\documentclass[prodmode,acmtods]{acmsmall} 

\usepackage[ruled]{algorithm2e}
\renewcommand{\algorithmcfname}{ALGORITHM}

\usepackage{amsmath}
\usepackage{amsfonts}
\usepackage{epsfig}
\usepackage{graphics}
\usepackage{amssymb}
\usepackage{color}
\usepackage{comment}
\usepackage{paralist}
\usepackage{multirow}
\usepackage{subfigure}

\SetAlFnt{\small}
\SetAlCapFnt{\small}
\SetAlCapNameFnt{\small}
\SetAlCapHSkip{0pt}
\IncMargin{-\parindent}

\begin{document}

\markboth{G. Gottlob et al.}{Query Rewriting and Optimization for Ontological Databases}

\markboth{}{Query Rewriting and Optimization for Ontological
Databases}



\newcommand{\A}{\mathcal{A}} \newcommand{\B}{\mathcal{B}}
\newcommand{\C}{\mathcal{C}} \newcommand{\D}{\mathcal{D}}
\newcommand{\E}{\mathcal{E}} \newcommand{\F}{\mathcal{F}}
\newcommand{\G}{\mathcal{G}} \renewcommand{\H}{\mathcal{H}}
\newcommand{\I}{\mathcal{I}} \newcommand{\J}{\mathcal{J}}
\newcommand{\K}{\mathcal{K}} \renewcommand{\L}{\mathcal{L}}
\newcommand{\M}{\mathcal{M}} \newcommand{\N}{\mathcal{N}}
\renewcommand{\O}{\mathcal{O}} \renewcommand{\P}{\mathcal{P}}
\newcommand{\Q}{\mathcal{Q}} \newcommand{\R}{\mathcal{R}}
\renewcommand{\S}{\mathcal{S}} \newcommand{\T}{\mathcal{T}}
\newcommand{\U}{\mathcal{U}} \newcommand{\V}{\mathcal{V}}
\newcommand{\W}{\mathcal{W}} \newcommand{\X}{\mathcal{X}}
\newcommand{\Y}{\mathcal{Y}} \newcommand{\Z}{\mathcal{Z}}


\newcommand{\setone}[2][1]{\set{#1\cld #2}}
\newcommand{\eset}{\emptyset}
\newcommand{\ol}[1]{\overline{#1}}                
\newcommand{\ul}[1]{\underline{#1}}               
\newcommand{\uls}[1]{\underline{\raisebox{0pt}[0pt][0.45ex]{}#1}}

\newcommand{\ra}{\rightarrow}
\newcommand{\Ra}{\Rightarrow}
\newcommand{\la}{\leftarrow}
\newcommand{\La}{\Leftarrow}
\newcommand{\lra}{\leftrightarrow}
\newcommand{\Lra}{\Leftrightarrow}
\newcommand{\lora}{\longrightarrow}
\newcommand{\Lora}{\Longrightarrow}
\newcommand{\lola}{\longleftarrow}
\newcommand{\Lola}{\Longleftarrow}
\newcommand{\lolra}{\longleftrightarrow}
\newcommand{\Lolra}{\Longleftrightarrow}
\newcommand{\ua}{\uparrow}
\newcommand{\Ua}{\Uparrow}
\newcommand{\da}{\downarrow}
\newcommand{\Da}{\Downarrow}
\newcommand{\uda}{\updownarrow}
\newcommand{\Uda}{\Updownarrow}


\newcommand{\incl}{\subseteq}
\newcommand{\imp}{\rightarrow}
\newcommand{\deq}{\doteq}
\newcommand{\dleq}{\dot{\leq}}                   


\newcommand{\per}{\mbox{\bf .}}                  

\newcommand{\cld}{,\ldots,}                      
\newcommand{\ld}[1]{#1 \ldots #1}                 
\newcommand{\cd}[1]{#1 \cdots #1}                 
\newcommand{\lds}[1]{\, #1 \; \ldots \; #1 \,}    
\newcommand{\cds}[1]{\, #1 \; \cdots \; #1 \,}    

\newcommand{\dd}[2]{#1_1,\ldots,#1_{#2}}             
\newcommand{\ddd}[3]{#1_{#2_1},\ldots,#1_{#2_{#3}}}  
\newcommand{\dddd}[3]{#1_{11}\cld #1_{1#3_{1}}\cld #1_{#21}\cld #1_{#2#3_{#2}}}

\newcommand{\ldop}[3]{#1_1 \ld{#3} #1_{#2}}   
\newcommand{\cdop}[3]{#1_1 \cd{#3} #1_{#2}}   
\newcommand{\ldsop}[3]{#1_1 \lds{#3} #1_{#2}} 
\newcommand{\cdsop}[3]{#1_1 \cds{#3} #1_{#2}} 


\newcommand{\quotes}[1]{{\lq\lq #1\rq\rq}}
\newcommand{\set}[1]{\{#1\}}                      
\newcommand{\Set}[1]{\left\{#1\right\}}
\newcommand{\bigset}[1]{\Bigl\{#1\Bigr\}}
\newcommand{\bigmid}{\Big|}
\newcommand{\card}[1]{|{#1}|}                     
\newcommand{\Card}[1]{\left| #1\right|}
\newcommand{\cards}[1]{\sharp #1}
\newcommand{\sub}[1]{[#1]}
\newcommand{\tup}[1]{\langle #1\rangle}            
\newcommand{\Tup}[1]{\left\langle #1\right\rangle}


\def\qed{\hfill{\qedboxempty}      
  \ifdim\lastskip<\medskipamount \removelastskip\penalty55\medskip\fi}

\def\qedboxempty{\vbox{\hrule\hbox{\vrule\kern3pt
                 \vbox{\kern3pt\kern3pt}\kern3pt\vrule}\hrule}}

\def\qedfull{\hfill{\qedboxfull}   
  \ifdim\lastskip<\medskipamount \removelastskip\penalty55\medskip\fi}

\def\qedboxfull{\vrule height 4pt width 4pt depth 0pt}

\newcommand{\markfull}{\qedboxfull}
\newcommand{\markempty}{\qed}



\newcommand{\assign}{:=}
\newcommand{\alg}[1]{$\mathsf{#1}$}
\newcommand{\tool}[1]{\textsc{#1}}
\newcommand{\tgdrewrite}{\alg{XRewrite}\xspace}
\newcommand{\iris}{\tool{Iris}\xspace}
\newcommand{\irispm}{\tool{Sysname}\xspace}
\newcommand{\alaska}{\tool{Alaska}\xspace}
\newcommand{\nyaya}{\tool{Nyaya}\xspace}


\newcommand{\norm}[1]{\mathsf{N}(#1)}
\newcommand{\cmpitem}{\noindent\ensuremath{-}\xspace}

\newcommand{\lin}[2]{\mathsf{b}(#1,#2)}
\newcommand{\mgu}[2]{\textsf{MGU}(#1,#2)}
\newcommand{\f}[1]{\textsf{f}(#1)}
\newcommand{\ff}[1]{\textsf{ff}(#1)}
\newcommand{\fr}[2]{\textsf{fr}_{#1}(#2)}

\newcommand{\DB}{\mathit{DB}} \newcommand{\wrt}[0]{w.r.t.\ }

\newcommand{\ifdirection}{``$\Leftarrow$'' {}}
\newcommand{\onlyifdirection}{``$\Rightarrow$'' {}}

\renewcommand{\emptyset}{\varnothing}

\newcommand{\arc}{^\curvearrowright}

\newcommand{\qans}[3]{\mathsf{QAns}(#1,#2,#3)} 

\newcommand{\relevent}[2]{\substack{#1,#2 \\ \twoheadrightarrow}} 


\newcommand{\rel}[1]{\mathsf{#1}}
\newcommand{\attr}[1]{\mathit{#1}}
\newcommand{\const}[1]{\mathit{#1}}
\newcommand{\vett}[1]{\vec{#1}}

\newcommand{\ext}[2]{#1^{#2}}


\newcommand{\pred}[1]{\mathit{pred}(#1)}

\newcommand{\atom}[1]{\underline{#1}}
\newcommand{\tuple}[1]{\mathbf{#1}}

\newcommand{\parent}[1]{\mathit{par}(#1)}
\newcommand{\dept}[1]{\mathit{depth}(#1)}

\newcommand{\mar}[1]{\hat{#1}}


\newcommand{\dom}{\Gamma}
\newcommand{\freshdom}{\Gamma_N}
\newcommand{\variables}{\Gamma_V}
\newcommand{\adom}[1]{\mathit{terms}(#1)}
\newcommand{\var}[1]{\mathit{var}(#1)}
\newcommand{\varpos}[2]{\mathit{var}(#1,#2)}


\newcommand{\entities}{\mathit{Ent}}
\newcommand{\relationships}{\mathit{Rel}}
\newcommand{\attributes}{\mathit{Att}}
\newcommand{\symbols}{\mathit{Sym}}


\newcommand{\dep}{\Sigma}
\newcommand{\tdep}{\Sigma_T}
\newcommand{\edep}{\Sigma_E}
\newcommand{\fdep}{\Sigma_F}
\newcommand{\kdep}{\Sigma_K}
\newcommand{\ndep}{\Sigma_\bot}

\newcommand{\key}[1]{\mathit{key}(#1)}
\newcommand{\isa}[1]{\mathit{ISA}}


\newcommand{\Var}[1]{\mathit{Var}(#1)}
\newcommand{\head}[1]{\mathit{head}(#1)}
\newcommand{\heads}[1]{\mathsf{heads}(#1)}
\newcommand{\body}[1]{\mathit{body}(#1)}
\newcommand{\arity}[1]{\mathit{arity}(#1)}
\newcommand{\conj}[1]{\mathit{conj}(#1)}

\newcommand{\cover}[1]{\mathit{cover}(#1)}

\newcommand{\ans}[3]{\mathit{ans}(#1,#2,#3)}


\newcommand{\ins}[1]{\mathbf{#1}}
\newcommand{\insA}{\ins{A}}
\newcommand{\insB}{\ins{B}}
\newcommand{\insC}{\ins{C}}
\newcommand{\insD}{\ins{D}}
\newcommand{\insX}{\ins{X}}
\newcommand{\insY}{\ins{Y}}
\newcommand{\insZ}{\ins{Z}}
\newcommand{\insK}{\ins{K}}
\newcommand{\insU}{\ins{U}}
\newcommand{\insT}{\ins{T}}
\newcommand{\insV}{\ins{V}}


\newcommand{\chase}[2]{\mathit{chase}(#1,#2)}
\newcommand{\ochase}[2]{\mathit{Ochase}(#1,#2)}
\newcommand{\rchase}[2]{\mathit{Rchase}(#1,#2)}
\newcommand{\instant}{\mathit{inst}}
\newcommand{\base}{\mathit{base}}

\newcommand{\pchase}[3]{\mathit{chase}^{#1}(#2,#3)}
\newcommand{\apchase}[3]{\mathit{chase}^{[#1]}(#2,#3)}
\newcommand{\aprchase}[3]{\mathit{Rchase}^{[#1]}(#2,#3)}
\newcommand{\level}[1]{\mathit{level}(#1)}
\newcommand{\freeze}[1]{\mathit{fr}(#1)}
\newcommand{\mods}[2]{\mathit{mods}(#1,#2)}
\newcommand{\subs}[2]{\gamma_{#1,#2}}


\renewcommand{\paragraph}[1]{\textbf{#1}}
\newenvironment{proofsk}{\textsc{Proof (sketch).}}{$\Box$\newline}
\newenvironment{proofrsk}{\textsc{Proof (rough sketch).}}{$\Box$\newline}
\newenvironment{proofidea}{\textsc{Proof idea.}}{$\Box$\newline}


\newcommand{\dg}[2]{\mathit{DG}(#1,#2)}
\newcommand{\rank}[1]{\mathit{rank}(#1)}


\title{Query Rewriting and Optimization for Ontological Databases}
\author{GEORG GOTTLOB
\affil{University of Oxford} GIORGIO ORSI \affil{University of
Oxford} ANDREAS PIERIS \affil{University of Oxford}}

\begin{abstract}
Ontological queries are evaluated against a knowledge base
consisting of an extensional database and an ontology (i.e., a set
of logical assertions and constraints which derive new intensional
knowledge from the extensional database), rather than directly on
the extensional database. The evaluation and optimization of such
queries is an intriguing new problem for database research. In this
paper, we discuss two important aspects of this problem: query
rewriting and query optimization.
Query rewriting consists of the compilation of an ontological query
into an equivalent first-order query against the underlying
extensional database. We present a novel query rewriting algorithm
for rather general types of ontological constraints which is
well-suited for practical implementations. In particular, we show
how a conjunctive query against a knowledge base, expressed using
linear and sticky existential rules, that is, members of the
recently introduced Datalog$^\pm$ family of ontology languages, can
be compiled into a union of conjunctive queries (UCQ) against the
underlying database.
Ontological query optimization, in this context, attempts to improve
this rewriting process so to produce possibly small and
cost-effective UCQ rewritings for an input query.
\end{abstract}



\terms{Algorithms, Theory, Languages, Performance}


%
%

\maketitle

\section{Introduction}\label{sec:introduction}

\subsection{Ontological Database Management Systems}

The use of ontological reasoning in companies, governmental
organizations, and other enterprises has become widespread in recent
years. An ontology is an explicit specification of a
conceptualization of an area of interest, and consists of a formal
representation of knowledge as a set of concepts within a domain,
and the relationships between instances of those concepts.
Moreover, ontologies have been adopted as high-level conceptual
descriptions of the data contained in data repositories that are
sometimes distributed and heterogeneous in the data models. Due to
their high expressive power, ontologies are also replacing more
traditional conceptual models such as UML class diagrams and Entity
Relationship schemata.

We are currently witnessing the marriage of ontological reasoning
and database technology, which gives rise to a new type of database
management systems, the so-called ontological database management
systems, equipped with advanced reasoning and query processing
mechanisms~\cite{CDLL*07,CGLP11}.
More precisely, an extensional database $D$ is combined with an
ontology $\dep$ which derives new intensional knowledge from the
extensional database.  An input conjunctive query is not just
answered against the database, as in the classical setting, but
against the logical theory (a.k.a. ontological database) $D \cup
\dep$ --- recall that conjunctive queries correspond to the
select-project-join fragment of relational algebra, and form one of
the most natural and commonly used languages for querying relational
databases~\cite{AbHV95}. Therefore, the answer to a conjunctive
query $\exists \insY \, \varphi(\insX,\insY)$ with distinguished
variables $\insX$ over the ontological database consists of all
tuples $\tuple{t}$ of constants such that, when we substitute the
variables $\insX$ with $\tuple{t}$, $\exists \insY
\varphi(\tuple{t},\insY)$ evaluates to $\mathit{true}$ in every
model of $D \cup \dep$, i.e., in every instance which contains $D$
and satisfies $\dep$.

This amalgamation of different technologies stems from the need for
semantically enhancing existing databases with ontological
constraints. Indeed, database technology providers have recognized
this need, and have recently started to build ontological reasoning
modules on top of their existing software with the aim of delivering
effective database management solutions to their customers.
For example, Oracle Inc. offers a system, called Oracle Database
11g, enhanced by modules performing ontological reasoning
tasks\footnote{http://www.oracle.com/technetwork/database/enterprise-edition/overview/index.html}.
Also, Ontotext offers a family of semantic repositories, called
OWLIM\footnote{http://www.ontotext.com/owlim}, and Semafora Systems
develops an inference machine, called
Ontobroker\footnote{http://www.semafora-systems.com/en/products/ontobroker/},
for processing ontologies that support all of the World Wide Web
Consortium (W3C) recommendations.
Enhancing databases with ontologies is also at the heart of several
research-based systems such as QuOnto~\cite{ACDL*05} and
Quest~\cite{RoCa12}.

\subsection{Ontology Languages}

Ontologies are modeled using formal languages called
ontology languages.
Description Logics (DLs)~\cite{BCMN+03} are a family of knowledge
representation languages widely used in ontological modeling. In
fact, DLs model a domain of interest in terms of concepts and roles,
which represent classes of individuals and binary relations on
classes of individuals, respectively. Interestingly, DLs provide the
logical underpinning for the Web Ontology Language (OWL), and its
revision OWL~2, as standartized by the
W3C\footnote{http://www.w3.org/TR/owl2-overview/}.
Unfortunately, in order to achieve favorable computational
properties, DLs are able only to describe knowledge for which the
underlying relational structure is treelike. Moreover, they usually
support only unary and binary relations.
The overcoming of the above limitations, through the definition of
expressive rule-based ontology languages, has become the last years
a field of intense research in the KR and database communities. In
fact, traditional database constraints such as
\emph{tuple-generating dependencies (TGDs)} (a.k.a.
\emph{existential rules} and \emph{Datalog$^\pm$ rules}) of the form
$\forall \insX \forall \insY \, \varphi(\insX,\insY) \ra \exists
\insZ \, \psi(\insX,\insZ)$, where $\varphi$ and $\psi$ are
conjunctions of atoms over a relational schema, appeared to be a
suitable formalism for ontological modeling and reasoning
--- examples of such languages can be found in~\cite{BLMS11,KrRu11,CaGL12,CaGP12}.

A vital computational property of an ontology language, apart from
ensuring the decidability, is to guarantee the tractability of
conjunctive query answering w.r.t.~the \emph{data complexity}, i.e.,
the complexity calculated by considering only the database as part
of the input. Indeed, the data complexity of query answering is
widely regarded as more meaningful and relevant in practice than the
\emph{combined complexity} (calculated by considering everything as
part of the input), since the query and the ontology are typically
of a size that can be productively assumed to be fixed, and usually
are much smaller than a typical relational database.
Several lightweight DLs have been proposed which guarantee that
conjunctive query answering is feasible in polynomial time w.r.t.
the data complexity. Such DLs are $\mathcal{EL}$~\cite{Baad03} and
the members of the \emph{DL-Lite} family~\cite{CDLL*07,PLCD*08},
i.e., DL-Lite$_{\R}$, DL-Lite$_{\F}$ and DL-Lite$_{\A}$. These
languages can be seen as tractable sublanguages of OWL; in fact, the
language DL-Lite$_\R$ forms the OWL 2
QL\footnote{http://www.w3.org/TR/owl2-profiles/} profile of OWL 2.
It was convincingly argued that, despite their simplicity,
$\mathcal{EL}$ and the DL-Lite formalisms are powerful enough for
modeling an overwhelming number of real-life scenarios.
More recently, several classes of TGDs have been identified which
guarantee the same low data complexity for conjunctive query
answering. For example, the class of \emph{guarded} TGDs, inspired
by the guarded fragment of first-order logic~\cite{AnBN98}, which is
noticeably more general than $\mathcal{EL}$ and the members of the
DL-Lite family, has been investigated in~\cite{CaGK08}
--- extensions of guarded TGDs can be found in~\cite{BLMS11,KrRu11}.
Moreover, the classes of \emph{linear} and \emph{sticky} TGDs, which
both encompass the DL-Lite family, have been proposed
in~\cite{CaGL12} and~\cite{CaGP12}.

\subsection{First-Order Rewritability}

\begin{figure}[t]
\begin{center}
\includegraphics[]{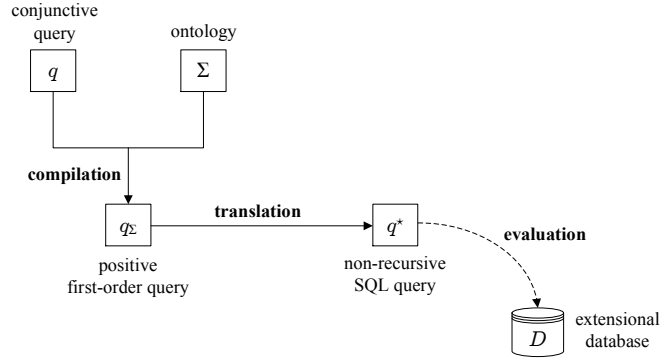}
\caption{Answering queries via rewriting.}\label{fig:rewriting}
\end{center}
\end{figure}

Polynomial time tractability is often considered not to be good
enough for efficient query processing. Ideally, one would like to
achieve the same complexity as for processing first-order queries,
or, equivalently, (non-recursive) SQL queries.
An ontology language $\mathcal{L}$ guarantees the \emph{first-order
rewritability} of conjunctive query answering if, for every
conjunctive query $q$ and ontology $\dep$ expressed in
$\mathcal{L}$, a positive first-order query $q_{\dep}$, called
\emph{perfect rewriting}\footnote{In general, there exist more than
one perfect rewritings. However, for query answering, all the
possible rewritings are equivalent, and thus we can refer to
\emph{the} perfect rewriting.}, can be constructed such that, given
a database $D$, $q_{\dep}$ evaluated over $D$ yields exactly the
same result as $q$ evaluated against the ontological database $D
\cup \dep$~\cite{CDLL*07}.
Since answering first-order queries is in \textsc{ac}$_{0}$ in data
complexity~\cite{Vard95}, it immediately follows that query
answering under ontology languages that guarantee the first-order
rewritability of the problem is also in \textsc{ac}$_{0}$ in data
complexity.

First-order rewritability is a most desirable property since it
ensures that the query answering process can be largely decoupled
from data access. In fact, as depicted in
Figure~\ref{fig:rewriting}, to answer a query $q$ over an
ontological database $D \cup \dep$, a separate software can compile
$q$ into $q_{\dep}$, then translate $q_\dep$ into a standard SQL
query $q^{\star}$, and finally submit it to the underlying
relational database management system holding $D$, where it is
evaluated and optimized in the usual way.

\begin{example}\label{exa:query-rewriting}
Consider the set $\dep$ consisting of the TGD:
\[
\forall X \forall Y \, \mathit{project}(X),\mathit{inArea}(X,Y)\
\ra\ \exists Z \, \mathit{hasCollaborator}(Z,Y,X),
\]
asserting that each project has an external collaborator specialized
in the area of the project.
We can ask for projects in the area of databases for which there are
external collaborators by posing the CQ $\exists A \,
\mathit{hasCollaborator}(A,db,B)$. Intuitively, due to the above
TGD, not only we have to query $\mathit{hasCollaborator}$, but we
also need to look for projects in the area of databases, as such
projects will necessarily have an external collaborator. The perfect
rewriting $q_\dep$ will thus be the union of CQs:
\[
\left(\exists A \, \mathit{hasCollaborator}(A,\mathit{db},B)\right)\
\vee\ \left(\mathit{project}(B) \wedge
\mathit{inArea}(B,\mathit{db})\right).
\]
Assuming the schema $\mathit{project}(\mathsf{p\_id}),
\mathit{inArea}(\mathsf{p\_id},\mathsf{area}),
\mathit{hasCollaborator}(\mathsf{c\_id},\mathsf{area},\mathsf{p\_id})$,
it is clear that $q_\dep$ can be written in SQL as shown in
Figure~\ref{fig:sql-query}. \hfill\markfull
\end{example}

\begin{figure}
\[
\begin{array}{lll}
&& \footnotesize \textrm{SELECT~C.p\_id~FROM~hasCollaborator~C}\\
&& \footnotesize \textrm{WHERE~C.area~=~'db'}\\
&& \footnotesize \textrm{UNION}\\
&& \footnotesize \textrm{SELECT~P.p\_id~FROM~project~P,~inArea~A}\\
&& \footnotesize \textrm{WHERE~A.area~=~'db'~AND~P.p\_id~=~A.p\_id.}
\end{array}
\]
\caption{The SQL query of Example~\ref{exa:query-rewriting}.}
\label{fig:sql-query}
\end{figure}

Interestingly, the members of the DL-Lite family of DLs, as well as
the classes of linear and sticky TGDs, guarantee the first-order
rewritability of conjunctive query answering. Actually, the above
languages guarantee a stronger property than first-order
rewritability: given a conjunctive query $q$, and an ontology $\dep$
expressed in one of the above formalisms, the perfect rewriting
$q_\dep$ can be expressed as a union of conjunctive queries, i.e.,
we do not need the full expressive power of positive first-order
queries.
As we explain below, the main problem that we address in this paper
is precisely the question of how to compute $q_\dep$ correctly and
efficiently, when the input ontology $\dep$ is expressed as a set of
linear or sticky TGDs.

\subsection{Aims and Objectives}\label{sec:aims-objectives}

The advantage of first-order rewritability is obvious, that is,
conjunctive query answering can be deferred to a standard query
language such as SQL, which in turn allows us to exploit mature and
efficient existing database technology that is accessible via the
underlying database management system.
However, there is a drawback in this approach: if the algorithm
which constructs the perfect rewriting inflates the query
excessively, and creates from a reasonably sized ontological query a
massive exponentially sized SQL query, then even the best database
management system may be of little use.
This problem gave rise to a flourishing research activity in the DL
community. A remarkable number of rewriting algorithms, with the aim
of compiling a conjunctive query and a DL-Lite ontology into a
``small'' union of conjunctive queries, have been proposed the last
five years (see, e.g.,~\cite{CDLL*07,PeMH10,ChTS11,KiKZ12,VeSS13})
--- see Section~\ref{sec:related-work}.

Surprisingly, before the conference version of the present
paper~\cite{ano}, no practical algorithm, able to efficiently
compile a conjunctive query and an ontology modeled using an
expressive TGD-based language into a union of conjunctive queries,
was available. It is the precise aim of this work to fill this gap
for linear and sticky TGDs. Both linearity and stickiness are
well-accepted paradigms:
\begin{itemize}
\item A TGD is
called \emph{linear} if it has only one body-atom~\cite{CaGL12};
notice that the body is the left-hand side of the implication.
Despite its simplicity, linearity forms a robust language with
several applications.
Linear TGDs are strictly more expressive than the description logic
DL-Lite$_{\R}$~\cite{CDLL*07} which, as already said, forms the OWL
2 QL profile of W3Cs standard ontology language for modeling
Semantic Web ontologies.
Importantly, linear TGDs, in contrast to DL-Lite$_{\R}$, can be used
with relational database schemas of arbitrary arity. The usefulness
of schemas of higher arity (not just unary and binary relations) has
been recognized by the DL community, and as evident we mention
DLR-Lite~\cite{CGLL+13}, a recent generalization of DL-Lite to
arbitrary arity, which is also captured by linear TGDs.
Also, linear TGDs generalize inclusion dependencies, a well-known
class of relational constraints; in fact, inclusion dependencies can
be equivalently written as TGDs with just one body-atom and one
head-atom without repeated variables.
Moreover, linear TGDs are powerful enough to express conditional
inclusion dependencies which extend traditional inclusion
dependencies by enforcing bindings of semantically related data
values, and they are useful in data cleaning and contextual schema
mapping~\cite{BEFF06,BrFM07}; in fact, conditional inclusion
dependencies can be written as linear TGDs with constant values in
the body.
Furthermore, linear TGDs generalize local-as-view (LAV) TGDs which
are employed in data exchange and data integration to define schema
mappings, i.e., specifications that describe how data for a source
schema can be transformed into data for a target schema; see,
e.g.,~\cite{CaKo09}.
Finally, linear TGDs can be used in schema evolution, and in
particular for expressing the decompose operator, with the aim of
splitting a table into smaller tables~\cite{CMDZ13}.

\item Stickiness~\cite{CaGP12} allows joins to appear in rule-bodies which
are not expressible via linear TGDs, let alone via DL(R)-Lite
assertions; more details are given in
Section~\ref{sec:preliminaries}.
Interestingly, sticky TGDs are able to capture well-known data
modeling constructs such as (conditional) inclusion and multivalued
dependencies. Furthermore, sticky TGDs, in contrast to linear TGDs
(and most of the existing DLs) allow to describe knowledge for which
the underlying relational structure is not treelike.
This is mainly due to the fact that sticky TGDs are expressive
enough for encoding the cartesian product of two tables; e.g., the
set of sticky TGDs consisting of $\forall X \forall Y \, p_i(X,Y)
\ra \exists Z \, p_i(Y,Z),s_i(Z)$, for each $i \in \{1,2\}$, and
$\forall X \forall Y \, s_1(X),s_2(Y) \ra r(X,Y)$, computes the
cartesian product of $s_1$ and $s_2$ which forms an infinite clique,
and thus the underlying relational structure has infinite treewidth.
As already observed by the DL community, there are some natural
ontological statements, e.g., ``all elephants are bigger than all
mice''~\cite{RuKH08}, which are expressible only via cartesian
product assertions. Notice that the above statement can be captured
by the sticky TGD $\forall X \forall Y \,
\mathit{elephant}(X),\mathit{mouse}(Y) \ra
\mathit{biggerThan}(X,Y)$.
Finally, sticky TGDs can also be used for schema evolution purposes,
and in particular for expressing the merge operator, with the aim of
putting together two or more tables~\cite{CMDZ13}.
\end{itemize}

Apart from designing a practical rewriting algorithm for linear and
sticky TGDs, we would also like to investigate the possibility of
improving the computation of the perfect rewriting on multi-core
architectures commonly available in modern database servers.
On the long term, we envision relational database systems able to
handle ontological constraints natively, as it is done today for
traditional data dependencies such as primary and foreign keys. A
key difference is that ontological constraints are not supposed to
be enforced by the DBMS as classical integrity constraints, but
rather to be taken into consideration during the evaluation of a
query.
This paper is a significant step towards this direction.

\subsection{The Existing Approach}

Although it is known that both linear and sticky TGDs guarantee the
first-order rewritability of conjunctive query answering, the
existing algorithms are of theoretical nature, and it is generally
accepted that there is no obvious way how they will lead to better
practical rewriting algorithms.
The key property of linear and sticky TGDs which implies the
first-order rewritability of conjunctive query answering is the
so-called \emph{bounded derivation-depth property
(BDDP)}~\cite{CaGL12}. As we shall see in
Section~\ref{sec:preliminaries}, to compute the answer to a
conjunctive query $q$ over an ontological database $D \cup \dep$,
where $\dep$ is a linear or sticky ontology, it suffices to evaluate
$q$ over a special model of $D \cup \dep$ which can be
homomorphically embedded into every other model of $D \cup \dep$.
Such a model, called \emph{universal model} (a.k.a. \emph{canonical
model}), always exists and can be constructed by applying the
\emph{chase procedure}, a powerful tool for reasoning about data
dependencies --- intuitively, the chase adds new atoms to the
extensional database $D$, possibly involving null values which act
as witnesses for the existentially quantified variables, until the
final result, denoted $\chase{D}{\dep}$, satisfies $\dep$. However,
$\chase{D}{\dep}$ is in general infinite, and thus not explicitly
computable. The BDDP implies that it suffices to evaluate $q$ over
an initial finite part of $\chase{D}{\dep}$ which depends only on
$q$ and $\dep$. Roughly, $\chase{D}{\dep}$ can be decomposed into
levels, where database atoms have level zero, while an inferred atom
has level $k+1$ if it is obtained due to atoms with maximum level
$k$; we refer to the part of the chase up to level $k$ as
$\pchase{k}{D}{\dep}$. Thus, the BDDP implies that there exists $k
\geqslant 0$ such that, for every database $D$, the answer to $q$
over $D \cup \dep$ coincides with the answer to $q$ over
$\pchase{k}{D}{\dep}$.
An algorithm for computing the prefect rewriting $q_\dep$ by
exploiting the above property has been presented in~\cite{CaGL12}.
Roughly, one can enumerate all the possible database ancestors
$D_1,\ldots,D_n$ of the image of the given query, and then, starting
from each $D_i$, construct $\pchase{k}{D}{\dep}$, where $k$ is the
depth provided by the BDDP, which will give rise to a query in the
final rewriting.
It is evident that such a  procedure is computationally expensive,
and also the obtained queries are usually very large and cannot be
effectively materialized. Notice that the goal of~\cite{CaGL12} was
to establish that classes of TGDs which enjoy the BDDP guarantee the
first-order rewritability of conjunctive query answering, without
taking into account implementation issues. It is apparent that we
had to look for new rewriting procedures which substantially deviate
from the one described above.

\subsection{Summary of Contributions}

Our contributions can be summarized as follows:

\begin{enumerate}
\item We propose a novel query rewriting algorithm, called
$\mathsf{XRewrite}$, which is based on backward-chaining resolution.
In fact, $\mathsf{XRewrite}$ uses the TGDs as rewriting rules, with
the aim of simulating, independently from the extensional database,
the chase derivations which are responsible for the generation of
the image of the input query. Such an algorithm is better for
practical applications than the one described above since, during
the rewriting process, we only explore the part of the chase which
is needed in order to entail the query, i.e., the proof of the
query, and thus we avoid the generation of a non-negligible number
of useless atoms.
Interestingly, $\mathsf{XRewrite}$ is sound and complete even if we
consider an arbitrary set of TGDs without any syntactic
restrictions; however, in this general case, the termination of the
algorithm is not guaranteed. We show that, if the input set of TGDs
is linear or sticky, then $\mathsf{XRewrite}$ terminates, and thus
it forms a practical query rewriting algorithm for linear and sticky
TGDs; recall that the designing of such an algorithm is the main
research challenge of this work.

\vspace{2mm}

\item We present a parallel version of $\mathsf{XRewrite}$, called
$\mathsf{XRewriteParallel}$, with the aim of reducing the overall
execution time for computing the final rewriting by exploiting
multi-core architectures. To the best of our knowledge, this is the
first attempt to design a parallel query rewriting algorithm.
The key idea is to decompose the input query $q$ into smaller
queries $q_1,\ldots,q_m$, where $m \geqslant 1$, in such a way that
each $q_i$ can be rewritten independently by concurrent rewriters
into a query $Q_{q_i}$, and then merge the queries
$Q_{q_1},\ldots,Q_{q_m}$ in order to obtain the final rewriting.

\vspace{2mm}

\item We propose a technique, called query elimination, aiming at
optimizing the final rewritten query under linear TGDs. Query
elimination, which is an additional step during the execution of
$\mathsf{XRewrite}$, reduces \emph{(i)} the size of the final
rewriting, \emph{(ii)} the number of atoms in each query of the
rewriting, and \emph{(iii)} the number of joins to be executed.
The key idea underlying query elimination is that the linearity of
TGDs allows us to effectively identify atoms in the body a query
which are logically implied (w.r.t.~a given set of TGDs) by other
atoms in the same query.

\vspace{2mm}

\item After implementing our algorithm, we have
analyzed its behavior, and we have spotted certain operations, such
as the computation of the most general unifier for a set of atoms,
that might benefit from caching.
%
%
We also perform an extensive analysis on the impact of our
optimizations on the rewriting process, and we show that all of them
reduce the number of redundant queries in the final rewriting.
We finally compare our system with \textsc{Alaska} (i.e., the
reference implementation of~\cite{KLMT12}) which is the only known
system which supports ontological query rewriting under arbitrary
TGDs.
We observe that both systems return minimal rewritings on the given
test cases.
However, query elimination allows us to perform a better exploration
of the rewriting search space on most of the given test cases.
Interestingly, even for the cases where \tool{Alaska} performs a
better exploration of the search space, our algorithm achieves
better performance due to the caching mechanism.
Notably, on certain test cases, the parallelization of the rewriting
provides a fundamental contribution towards making the rewriting
manageable as the number of explored and generated queries is
drastically reduced.
\end{enumerate}

\paragraph{\small{\textsf{Roadmap.}}}
After a review of previous work on query rewriting in
Section~\ref{sec:related-work}, and some technical definitions and
preliminaries in Section~\ref{sec:preliminaries}, we proceed with
our new results.
In Section~\ref{sec:ucq-rewriting}, we present the rewriting
algorithm $\mathsf{XRewrite}$, and in Section~\ref{sec:parallelize}
its parallel version.
In Section~\ref{sec:ucq-optimization}, we present the query
elimination technique.
Implementation issues are discussed in
Section~\ref{sec:implementation}, while the experimental evaluation
is presented in Section~\ref{sec:experimental-evaluation}.
We conclude in Section~\ref{sec:conclusions} with a brief outlook on
further research.

\section{Related Work on Query Rewriting}\label{sec:related-work}

An early query rewriting algorithm for the DL-Lite family of DLs,
introduced in~\cite{CDLL*07} and implemented in the QuOnto system,
reformulates the given query into a union of conjunctive queries.
The size of the reformulated query is unnecessarily large. This is
mainly due to the fact that the factorization step (which is needed,
as we shall see, to guarantee completeness) is applied in a
``blind'' way, even if it is not needed, and as a result many
superfluous queries are generated.
In~\cite{PeMH10} an alternative resolution-based rewriting algorithm
for DL-Lite$_{\R}$ is proposed, implemented in the Requiem system,
that addressed the issue of the useless factorizations (and
therefore of the redundant queries generated due to this weakness)
by directly handling existential quantification through proper
functional terms
--- notice that this algorithm works also for more expressive DLs,
which do not guarantee first-order rewritability of query answering;
in this case, the computed rewriting is a (recursive) Datalog query.
A query rewriting algorithm for DL-Lite$_{\R}$, called Rapid, which
is more efficient than the one in~\cite{PeMH10}, is presented
in~\cite{ChTS11}. The efficiency of Rapid is based on the selective
and stratified application of resolution rules; roughly, it takes
advantage of the query structure and applies a restricted sequence
of resolutions that may lead to useful and redundant-free
rewritings.
An alternative query rewriting technique for DL-Lite$_{\R}$ is
presented in~\cite{KiKZ12} --- although the obtained rewritings are,
in general, not correct and of exponential size, in most practical
cases the rewritings are correct and of polynomial size.
In~\cite{VeSS13}, the problem of computing query rewritings for
DL-Lite$_{\R}$ in an incremental way is investigated. More
precisely, a technique which computes an extended query by
``extending'' a previously computed rewriting of the initial query
(and thus avoiding recomputation) is proposed.

The algorithms mentioned above leverage specificities of DLs, such
as the limit to unary and binary predicates only and the absence of
variable permutations in the axioms. Therefore, they cannot be
easily extended to more general TGD-based languages; in fact,
DL-based systems often resort to case-by-case analysis on the
syntactic form of the DL axioms.
Following a more general approach, the
works~\cite{ano,KLMT12,KLMT13} presented a backward-chaining
rewriting algorithm which is able to deal with arbitrary TGDs,
providing that the language under consideration satisfies suitable
syntactic restrictions that guarantee the termination of the
algorithm.
Other works, which follow a different approach, and instead of
computing a union of conjunctive queries the rewritings are
expressed in some other query language, such as non-recursive
Datalog, can be found in the
literature~\cite{RoAl10,OrPi11,GoSc12,KKPZ12,Thom13}.

A distantly related research field is that of database query
reformulation in presence of views and
constraints~\cite{DePT99,Hale01}. Given a conjunctive query $q$, and
a set of constraints $\dep$, the goal is to find all the minimal
equivalent reformulations of $q$ w.r.t.~$\dep$. The most interesting
approach in this respect is the chase \& backchase
algorithm~\cite{DePT99}, implemented in the MARS
system~\cite{DeTa03}.
The relationship of the chase \& backchase algorithm with this work
is discussed in Section~\ref{sec:ucq-optimization}.

\section{Definitions and Background}\label{sec:preliminaries}

\subsection{Technical Definitions}\label{sec:technical-definitions}

We present background material necessary for this paper. We recall
some basics on relational databases, relational queries,
tuple-generating dependencies, and the chase procedure relative to
such dependencies. For further details on the above notions we refer
the reader to~\cite{AbHV95}.

\paragraph{\small{\textsf{Alphabets.}}}
We define the following pairwise disjoint (countably infinite) sets
of symbols: a set $\dom$~of \emph{constants} (constitute the
``normal'' domain of a database), a set $\freshdom$~of \emph{labeled
nulls} (used as placeholders for unknown values, and thus can be
also seen as (globally) existentially quantified variables), and a
set $\variables$ of (regular) \emph{variables} (used in queries and
dependencies). Different constants represent different values
(\emph{unique name assumption}), while different nulls may represent
the same value. A fixed lexicographic order is assumed on $\dom \cup
\freshdom$ such that every value in $\freshdom$ follows all those in
$\dom$. We denote by $\insX$ sequences (or sets, with a slight abuse
of notation) of variables $X_1,\ldots,X_k$, with $k \geqslant 1$.
Throughout, let $[n] = \{1,\ldots,n\}$, for any integer $n \geqslant
1$.

\paragraph{\small{\textsf{Relational Model.}}}
A \emph{relational schema} $\R$ (or simply \emph{schema}) is a set
of \emph{relational symbols} (or \emph{predicates}), each with its
associated arity.  We write $r/n$ to denote that the predicate $r$
has arity $n$. By $\arity{\R}$ we refer to the maximum arity over
all predicates of $\R$. A \emph{position} $r[i]$ (in $\R$) is
identified by a predicate $r \in \R$ and its $i$-th argument (or
attribute). A \emph{term} $t$ is a constant, null, or variable. An
\emph{atomic formula} (or simply \emph{atom}) has the form $r(t_{1},
\ldots, t_{n})$, where $r/n$ is a relation, and $t_{1}, \ldots,
t_{n}$ are terms. For an atom $\atom{a}$, we denote as
$\adom{\atom{a}}$ and $\var{\atom{a}}$ the set of its terms and the
set of its variables, respectively. These notations naturally extend
to sets of atoms. Conjunctions of atoms are often identified with
the sets of their atoms. An \emph{instance} $I$ for a schema $\R$ is
a (possibly infinite) set of atoms of the form $r(\tuple{t})$, where
$r/n \in \R$ and $\tuple{t} \in (\dom \cup \freshdom)^{n}$. A
\emph{database} $D$ is a finite instance such that $\adom{D} \subset
\dom$.

\paragraph{\small{\textsf{Substitutions.}}}
A \emph{substitution} from a set of symbols $S$ to a set of symbols
$S'$ is a function \mbox{$h : S \rightarrow S'$} defined as follows:
$\emptyset$ is a substitution (empty substitution), and if $h$ is a
substitution, then $h \cup \{t \rightarrow t'\}$ is a substitution,
where $t \in S$ and $t' \in S'$; if $t \rightarrow t'\, \in\, h$,
then we write $h(t) = t'$. An assertion of the form $t \rightarrow
t'$ is called \emph{mapping}. The \emph{restriction} of $h$ to $T
\subseteq S$, denoted $h|_T$, is the substitution $h' = \{t \ra
h(t)~|~t \in T\}$.
A \emph{homomorphism} from a set of atoms $A$ to a set of atoms $A'$
is a substitution $h: \dom \cup \freshdom \cup \variables
\rightarrow \dom \cup \freshdom \cup \variables$ such that: if $t
\in \dom$, then $h(t) = t$, and if $r(t_{1}, \ldots, t_{n}) \in A$,
then $h(r(t_{1}, \ldots, t_{n})) = r(h(t_{1}), \ldots, h(t_{n})) \in
A'$.
A set of atoms $A = \{\atom{a}_1,\ldots,\atom{a}_n\}$, where $n
\geqslant 2$, \emph{unifies} if there exists a substitution
$\gamma$, called \emph{unifier} for $A$, such that
$\gamma(\atom{a}_1) = \ldots = \gamma(\atom{a}_n)$. A \emph{most
general unifier (MGU)} for $A$ is a unifier for $A$, denoted as
$\gamma_{A}$, such that for each other unifier $\gamma$ for $A$,
there exists a substitution $\gamma'$ such that $\gamma = \gamma'
\circ \gamma_{A}$. Notice that if a set of atoms unify, then there
exists a MGU. Furthermore, the MGU for a set of atoms is unique
(modulo variable renaming).

\paragraph{\small{\textsf{Datalog.}}}
A \emph{Datalog rule} $\rho$ is an expression of the form
$\atom{a}_0 \la \atom{a}_1, \ldots, \atom{a}_n$, for $n \geqslant
0$, where $\atom{a}_i$ is an atom containing constants of $\dom$ and
variables of $\variables$, and every variable occurring in
$\atom{a}_0$ must appear in at least one of the atoms $\atom{a}_1,
\ldots, \atom{a}_n$; the latter is known as the safety condition.
The atom $\atom{a}_0$ is called the \emph{head} of $\rho$, denoted
as $\head{\rho}$, while the set of atoms $\{\atom{a}_1, \ldots,
\atom{a}_n\}$ is called the \emph{body} of $\rho$, denoted as
$\body{\rho}$.
A \emph{Datalog program} $\Pi$ over a schema $\R$ is a set of
Datalog rules such that, for each $\rho \in \Pi$, the predicate of
$\head{\rho}$ does not occur in $\R$. The program $\Pi$ is
\emph{non-recursive} if there is some ordering
$\rho_1,\ldots,\rho_n$ of the rules of $\Pi$ so that the predicate
in the head of $\rho_i$ does not occur in the body of a rule
$\rho_j$, for each $j \leqslant i$.
The \emph{extensional database (EDB)} predicates are those that do
not occur in the head of any rule of $\Pi$; all the other predicates
are called \emph{intensional database (IDB)} predicates.
A \emph{model} of $\Pi$ is an instance $I$ for $\R$ such that, for
every Datalog rule of the form $\atom{a}_0 \la \atom{a}_1, \ldots,
\atom{a}_n$ appearing in $\Pi$, $I$ satisfies the first-order
formula $\forall \insX (\atom{a}_1 \wedge \ldots \wedge \atom{a}_n
\ra \atom{a}_0)$, where $\insX$ are the variables occurring in
$\rho$. In other words, whenever there exists a homomorphism $h$
such that $h(\{\atom{a}_1,\ldots,\atom{a}_n\}) \subseteq I$,
$h(\atom{a}_0) \in I$.
The semantics of $\Pi$ w.r.t.~a database $D$ for $\R$, denoted as
$\Pi(D)$, is the minimum model of $\Pi$ containing $D$ (which is
unique and always exists).
%
%

\paragraph{\small{\textsf{Queries.}}}
An $n$-ary \emph{Datalog query} $Q$ over a schema $\R$ is a pair
$\tup{\Pi,p}$, where $\Pi$ is a Datalog program over $\R$, and $p$
is an $n$-ary (output) predicate which occurs in the head of at
least one rule of $\Pi$.
$Q$ is a \emph{non-recursive Datalog query} if $\Pi$ is
non-recursive. $Q$ is a \emph{union of conjunctive queries (UCQs)}
if $\Pi$ is non-recursive, $p$ is the only IDB predicate in $\Pi$,
and for each rule $\rho \in \Pi$, $p$ does not occur in
$\body{\rho}$. Finally, $Q$ is a \emph{conjunctive query (CQ)} if it
is a union of CQs, and $\Pi$ contains exactly one rule.
The \emph{answer} to an $n$-ary Datalog query $Q = \tup{\Pi,p}$ over
a database $D$ is the set $\{\tuple{t} \in \dom^n~|~p(\tuple{t}) \in
\Pi(D)\}$, denoted $Q(D)$.
Since the output predicate of a (U)CQ is clear from the syntax of
the query, in the rest of the paper, for brevity, a CQ is seen as a
Datalog rule, while a UCQ is seen as a Datalog program (instead of a
pair consisting of a program and a predicate).
The variables occurring in the head of a CQ are its
\emph{distinguished variables}.
The answer to a CQ $q$\footnote{Henceforth, for clarity, we usually
use lower case letters for CQs and upper case letters for UCQs.}
over a (possibly infinite) instance $I$ can be equivalently defined
as the set of all tuples of constants $\tuple{t}$ for which there
exists a homomorphism $h$ such that $h(\body{q}) \subseteq I$ and
$h(\insX) = \tuple{t}$, where $\insX$ are the distinguished
variables of $q$.
The answer to a UCQ $Q$ over $I$ can be equivalently defined as the
set of tuples $\{\tuple{t}~|~\textrm{there~exists~} q \in Q
\textrm{~such~that~} \tuple{t} \in q(I)\}$.

\paragraph{\small{\textsf{Tuple-Generating Dependencies.}}}
A \emph{tuple-generating dependency (TGD)} $\sigma$ over  a schema
$\R$ is a first-order formula $\forall \insX \forall
\insY\,\varphi(\insX,\insY)\, \rightarrow\, \exists
\insZ\,\psi(\insX,\insZ)$, where $\insX \cup \insY \cup \insZ
\subset \variables$, and $\varphi,\psi$ are conjunctions of atoms
over $\R$ (possibly with constants). Formula $\varphi$ is the
$\emph{body}$ of $\sigma$, denoted $\body{\sigma}$, while $\psi$ is
the \emph{head} of $\sigma$, denoted $\head{\sigma}$. Henceforth,
for brevity, we will omit the universal quantifiers in front of
TGDs.
Such $\sigma$ is satisfied by an instance $I$ for $\R$, written $I
\models \sigma$, if the following holds: whenever there exists a
homomorphism $h$ such that $h(\varphi(\insX,\insY)) \subseteq I$,
then there exists a homomorphism $h' \supseteq h|_{\insX}$, called
\emph{extension} of $h|_{\insX}$, such that $h'(\psi(\insX,\insZ))
\subseteq I$.
An instance $I$ satisfies a set $\dep$ of TGDs, denoted $I \models
\dep$, if $I \models \sigma$ for each $\sigma \in \dep$.
A set $\dep$ of TGDs is in \emph{normal form} if each of its TGDs
has a single head-atom which contains only one occurrence of an
existentially quantified variable.
As shown, e.g., in~\cite{CaGP12}, every set $\dep$ of TGDs over a
schema $\R$ can be transformed in logarithmic space into a set
$\norm{\dep}$ over a schema $\R_{\norm{\dep}}$ in normal form of
size at most quadratic in $|\dep|$, such that $\dep$ and
$\norm{\dep}$ are equivalent w.r.t.~query answering --- for more
details see Section~\ref{appsec:technical-definitions}.

\paragraph{\small{\textsf{Conjunctive Query Answering under TGDs.}}}
Given a database $D$ for a schema $\R$, and a set $\dep$ of TGDs
over $\R$, the answers we consider are those that are true in
\emph{all} models of $D$ w.r.t.~$\dep$. Formally, the \emph{models}
of $D$ w.r.t.~$\dep$, denoted as $\mods{D}{\dep}$, is the set of all
instances $I$ such that $I \supseteq D$ and $I \models \dep$.
The \emph{answer} to an $n$-ary CQ $q$ w.r.t.~$D$ and $\dep$,
denoted as $\ans{q}{D}{\dep}$, is the set of $n$-tuples
$\{\tuple{t}~|~\tuple{t} \in q(I), \textrm{~for~each~} I \in
\mods{D}{\dep}\}$; the answer to an $n$-ary UCQ is defined
analogously.
Notice that the associated decision problem, which asks whether a
tuple of constants belongs to the answer of a CQ w.r.t.~a database
and a set of TGDs, is undecidable under arbitrary
TGDs~\cite{BeVa81}; in fact, it remains undecidable even when the
schema and the set of TGDs are fixed~\cite{CaGK08}, or even when the
set of TGDs is a singleton~\cite{BLMS11}.
Concrete classes of TGDs which are of special interest for the
current work, and also guarantee the decidability of query
answering, are presented in
Section~\ref{sec:concrete-classes-of-tgds}.

\paragraph{\small{\textsf{The TGD Chase Procedure.}}}
The \emph{chase procedure} (or simply \emph{chase}) is a fundamental
algorithmic tool introduced for checking implication of
dependencies~\cite{MaMS79}, and later for checking query
containment~\cite{JoKl84}. Informally, the chase is a process of
repairing a database w.r.t.~a set of dependencies so that the
resulted instance satisfies the dependencies. By abuse of
terminology, we shall use the term ``chase'' interchangeably for
both the procedure and its result. The chase works on an instance
through the so-called \emph{TGD chase rule}:
\begin{description}\itemsep-\parsep
\item[TGD chase rule] Consider an instance $I$ for a schema $\R$, and a TGD $\sigma :
\varphi(\insX,\insY) \rightarrow \exists \insZ\,\psi(\insX,\insZ)$
over $\R$. We say that $\sigma$ is \emph{applicable} to $I$ if there
exists a homomorphism $h$ such that $h(\varphi(\insX,\insY))
\subseteq I$.
The result of \emph{applying} $\sigma$ to $I$ with $h$ is $I' = I
\cup h'(\psi(\insX,\insZ))$, and we write $I \tup{\sigma,h} I'$,
where $h'$ is an extension of $h|_{\insX}$ such that $h'(Z)$ is a
``fresh'' labeled null of $\freshdom$ not occurring in $I$, and
following lexicographically all those in $I$, for each $Z \in
\insZ$. In fact, $I \tup{\sigma,h} I'$ defines a single TGD chase
step.
\end{description}

Let us now give the formal definition of the \emph{chase} of a
database w.r.t.~a set of TGDs. A \emph{chase sequence} of a database
$D$ w.r.t.~a set $\dep$ of TGDs is a sequence of chase steps $I_i
\tup{\sigma_i,h_i} I_{i+1}$, where $i \geqslant 0$, $I_0 = D$ and
$\sigma_i \in \dep$. The chase of $D$ w.r.t.~$\dep$, denoted
$\chase{D}{\dep}$, is defined as follows:
\begin{itemize}\itemsep-\parsep
\item[--] A \emph{finite chase} of $D$ w.r.t.~$\dep$ is a finite chase
sequence $I_i \tup{\sigma_i,h_i} I_{i+1}$, where $0 \leqslant i <
m$, and there is no $\sigma \in \dep$ which is applicable to $I_m$;
let $\chase{D}{\dep} = I_m$.
\item[--] An infinite chase sequence $I_i \tup{\sigma_i,h_i} I_{i+1}$, where
$i \geqslant 0$, is \emph{fair} if whenever a TGD $\sigma :
\varphi(\insX,\insY) \rightarrow \exists \insZ\,\psi(\insX,\insZ)$
is applicable to $I_i$ with homomorphism $h$, then there exists an
extension $h'$ of $h|_{\insX}$ and $k > i$ such that
$h'(\head{\sigma}) \subseteq I_k$.
An \emph{infinite chase} of $D$ w.r.t.~$\dep$ is a fair infinite
chase sequence $I_i \tup{\sigma_i,h_i} I_{i+1}$, where $i \geqslant
0$; let $\chase{D}{\dep} = \bigcup_{i=0}^{\infty} I_i$.
\end{itemize}
Let $\apchase{k}{D}{\dep}$ be the instance constructed after $k
\geqslant 0$ applications of the TGD chase step. An example of the
chase procedure can be found in
Section~\ref{appsec:technical-definitions}.
It is well-known that the chase of $D$ w.r.t.~$\dep$ is a
\emph{universal model} of $D$ w.r.t.~$\dep$, i.e., for each $I \in
\mods{D}{\dep}$, there exists a homomorphism $h_I$ such that
$h_I(\chase{D}{\dep}) \subseteq I$~\cite{FKMP05,DeNR08}.
Using this universality property, it can be shown that the chase is
a formal algorithmic tool for query answering under TGDs. More
precisely, the answer to a CQ $q$ w.r.t.~a database $D$ and a set of
TGDs $\dep$ coincides with the answer to $q$ over the chase of $D$
w.r.t.~$\dep$, i.e., $\ans{q}{D}{\dep} = q(\chase{D}{\dep})$.

The TGD chase rule given above is known as \emph{oblivious} since it
``forgets'' to check whether the TGD under consideration is already
satisfied, i.e., it adds atoms to the given instance even if it is
not necessary. The version of the TGD chase rule which applies
stricter criteria to the applicability of TGDs, with the aim of
adding atoms to the given instance only if it is necessary, is
called \emph{restricted}.
The universality property was originally shown for the restricted
version of the chase~\cite{FKMP05,DeNR08}, which is considered as
the standard one. However, as explicitly stated in~\cite{CaGK13},
the universality property holds also for the oblivious chase; this
was established by showing the existence of a homomorphism from the
oblivious to the restricted chase. Thus, for our purposes, we can
safely consider the oblivious chase. This is done for technical
clarity and simplicity. As discussed in~\cite{JoKl84}, even in the
simple case of inclusion dependencies, things become technically
more complicated if the restricted chase is employed, since the
applicability of a TGD depends on the presence of other atoms
previously constructed by the chase.

\subsection{Query Answering via Rewriting}\label{sec:query-rewriting}

A fundamental property that a class of TGDs should enjoy is to
guarantee the decidability of (the decision version) of conjunctive
query answering; recall that in general this problem is undecidable.
However, as already discussed in Section~\ref{sec:introduction}, to
be able to work with very large data sets, decidability of query
answering is not enough. We need also high tractability in data
complexity, i.e., when both the query and the set of TGDs are fixed,
and possibly feasible by the use of relational query processors.
First-order rewritability, introduced in the context of description
logics~\cite{CDLL*07}, guarantees the above desirable properties.
Roughly speaking, given a CQ and a set of TGDs, a (finite)
first-order query can be constructed, called \emph{perfect
rewriting}, that takes into account the semantic consequences of the
TGDs. Then, the answer to the input query w.r.t.~a database $D$ and
the set of TGDs is obtained by evaluating the perfect rewriting
directly over $D$.
Formally, the problem of conjunctive query answering under a set of
TGDs $\dep$ is \emph{first-order rewritable} if, for every CQ $q$, a
(finite) positive first-order query $q_\dep$ can be constructed such
that, for every database $D$, $\ans{q}{D}{\dep} = q_\dep(D)$.
Unfortunately, the problem of deciding whether a set of TGDs
guarantees the first-order rewritability of CQ answering is
undecidable; for more details see
Section~\ref{appsec:query-rewriting}.

It is well-known that the evaluation of first-order queries is in
the highly tractable class $\textsc{ac}_0$ in data
complexity~\cite{Vard95}. Recall that this is the complexity class
of recognizing words in languages defined  by constant-depth Boolean
circuits with (unlimited fan-in) AND and OR gates (see,
e.g.,~\cite{Papa94}). Consequently, CQ answering under sets of TGDs
which guarantee the first-order rewritability of the problem is in
$\textsc{ac}_0$ in data complexity.
Given that every first-order query can be equivalently written in
(non-recursive) SQL, in practical terms this means that CQ answering
can be deferred to a standard query language such as SQL. This
allows us to exploit all the optimization capabilities of the
underlying RDBMS.

\subsection{Concrete Classes of
TGDs}\label{sec:concrete-classes-of-tgds}

Since the problem of identifying first-order rewritability is
undecidable, it is not possible to syntactically characterize the
fragment of TGDs which guarantees the first-order rewritability of
CQ answering. However, several sufficient syntactic conditions have
been proposed
--- the two main conditions are linearity and stickiness.

\paragraph{\small{\textsf{Linearity.}}}
Linear TGDs have been proposed in~\cite{CaGL12}. A TGD $\sigma$ is
called \emph{linear} if $\sigma$ has only one body-atom. The class
of linear TGDs, i.e., the set of all possible sets of linear TGDs,
is denoted $\mathsf{LINEAR}$.
Despite its simplicity, as already discussed in
Section~\ref{sec:aims-objectives}, $\mathsf{LINEAR}$ is quite
natural with several applications.
%
%
Linear TGDs guarantee the first-order rewritability of CQ
answering~\cite{CaGL12}; this is also implicit in~\cite{BLMS11},
where atomic-hypothesis rules, which coincide with linear TGDs, are
investigated. This result was established by showing that
$\mathsf{LINEAR}$ enjoys the BDDP. However, as already remarked in
Section~\ref{sec:introduction}, the techniques based on the BDDP do
not lead to practical query rewriting algorithms.

\paragraph{\small{\textsf{Stickiness.}}}
The class of sticky sets of TGDs, denoted $\mathsf{STICKY}$, has
been proposed in~\cite{CaGP12} with the aim of identifying an
expressive class that allows for meaningful joins in rule-bodies.
\begin{figure}[t]
  \epsfclipon \centerline {\hbox{
      \leavevmode \epsffile{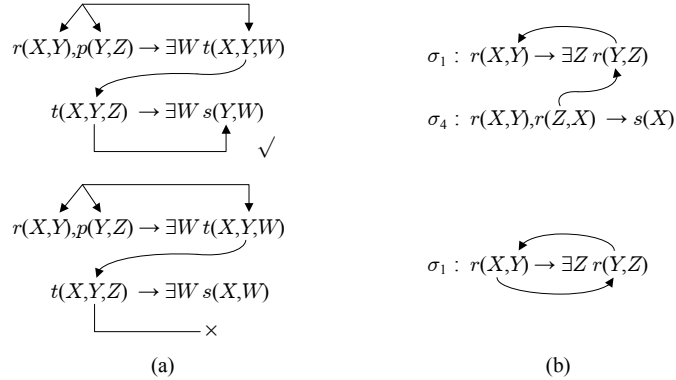} }}
  \epsfclipoff \caption{Sticky property and propagation step.}
  \label{fig:sticky-property-propagation-step}
  \vspace{-1.5mm}
\end{figure}
The key idea underlying stickiness is to ensure that, during the
chase, terms which are associated with body-variables that appear
more than once (i.e., join variables) always are propagated (or
``stick'') to the inferred atoms; this is illustrated in
Figure~\ref{fig:sticky-property-propagation-step}(a).

The formal definition of sticky sets of TGDs hinges on a
variable-marking procedure called $\mathsf{SMarking}$. This
procedure accepts as input a set $\dep$ of TGDs, and returns the
same set after marking some of its body-variables.
For notational convenience, given a TGD $\sigma$, an atom $\atom{a}
\in \head{\sigma}$, and a universally quantified variable $V$ of
$\sigma$, $\mathit{pos}(\sigma,\atom{a},V)$ is the set of positions
in $\atom{a}$ at which $V$ occurs.
$\mathsf{SMarking}(\dep)$ is constructed as follows. First, we apply
on $\dep$ the \emph{initial marking} step: for each $\sigma \in
\dep$, and for each variable $V \in \var{\body{\sigma}}$, if there
exists an atom $\atom{a} \in \head{\sigma}$ such that $V \not\in
\var{\atom{a}}$, then each occurrence of $V$ in $\body{\sigma}$ is
marked.
$\mathsf{SMarking}(\dep)$ is obtained by applying exhaustively
(i.e., until a fixpoint is reached) on $\dep$ the \emph{propagation}
step: for each pair $\tup{\sigma,\sigma'} \in \dep \times \dep$, for
each atom $\atom{a} \in \head{\sigma}$, and for each universally
quantified variable $V \in \var{\atom{a}}$, if there exists an atom
$\atom{b} \in \body{\sigma'}$ in which a marked variable occurs at
each position of $\mathit{pos}(\sigma,\atom{a},V)$, then each
occurrence of $V$ in $\body{\sigma}$ is marked.

\begin{example}
Consider the set $\dep$ consisting of
\[
\begin{array}{lcl}
\sigma_1\ :\ r(X,Y) \ra \exists Z \, r(Y,Z) &\qquad& \sigma_3\ :\ s(X),s(Y) \ra p(X,Y)\\
\sigma_2\ :\ r(X,Y) \ra s(X) &\qquad& \sigma_4\ :\ r(X,Y),r(Z,X) \ra s(X).\\
\end{array}
\]
By applying the initial marking step the body-variables of $\dep$
are marked with a cap (i.e., $\hat{V}$), and due to the propagation
step are marked with a double-cap as follows:
\[
\begin{array}{lcl}
\sigma_1\ :\ r(\mar{X},\mar{\mar{Y}}) \ra \exists Z \, r(Y,Z) &\qquad& \sigma_3\ :\ s(X),s(Y) \ra p(X,Y)\\
\sigma_2\ :\ r(X,\mar{Y}) \ra s(X) &\qquad& \sigma_4\ :\
r(X,\mar{Y}),r(\mar{Z},X) \ra s(X).
\end{array}
\]
Figure~\ref{fig:sticky-property-propagation-step}(b) depicts the two
ways of propagating the marking to the variable $Y$ of $\sigma_1$.
\hfill\markfull
\end{example}

A set $\dep$ of TGDs is called \emph{sticky} if, for every $\sigma
\in \mathsf{SMarking}(\dep)$, each marked variable appears only
once.
Stickiness guarantees the first-order rewritability of CQ
answering~\cite{CaGL12}. As for linear TGDs, this was established by
showing that the BDDP holds, and hence all the drawbacks of this
approach are inherited.\\

\paragraph{\small{\textsf{Normal Form.}}}
Notice that the normalization procedure for TGDs, presented in
Section~\ref{appsec:technical-definitions}, preserves linearity and
stickiness. In other words, given a linear (resp., sticky) set
$\dep$ of TGDs, the set $\norm{\dep}$ is linear (resp., sticky).
Thus, in the rest of the paper we assume, without loss of
generality, that TGDs have only one head-atom with at most one
existentially quantified variable which occurs once. This assumption
will allow us to simplify our later technical definitions and
proofs.
Given a TGD $\sigma$, we refer to the position of the (single)
existentially quantified variable by $\pi_{\exists}(\sigma)$; if
there is no existentially quantified variable, then
$\pi_{\exists}(\sigma) = \varepsilon$.

\section{UCQ Rewriting}\label{sec:ucq-rewriting}

In this section, we tackle the problem of CQ answering under linear
and sticky sets of TGDs. Our goal is to design a rewriting algorithm
which is well-suited for practical applications. In particular, we
present a backward-chaining rewriting algorithm which constructs a
union of conjunctive queries.
Let us say that our techniques apply immediately even if we
additionally consider a limited form of functional dependencies, and
negative constraints of the form $\forall \insX \, \varphi(\insX)
\ra \bot$, where $\varphi$ is a conjunction of atoms. Notice that
these modeling features are vital for ontological reasoning
purposes. Due to space reasons, we omit the details and we refer the
reader to Section~\ref{appsec:additional-modeling-features}.

\subsection{An Informal Description}\label{sec:xrewrite-informal-description}

Given a CQ $q$ and a set $\dep$ of TGDs, the actual computation of
the rewriting is done by exhaustively applying a backward
resolution-based step, called \emph{rewriting step}, which uses the
rules of $\dep$ as rewriting rules whose direction is right-to-left.
More precisely, a rewriting step is applied on a CQ, starting from
the given query $q$, and gives rise to a new CQ which will be part
of the final rewriting. Intuitively, a rewriting step simulates, in
the reverse direction (hence the term ``backward''), an application
of a TGD during the construction of the chase. In other words, by
applying the rewriting step we bypass an application of a TGD during
the chase, and the obtained query is one level closer to the
database-level. This is done until there are no other TGD chase
steps to bypass, which means that we reached the database-level, as
required.

\begin{example}[Rewriting Step]\label{exa:rewriting-step}
Consider the TGD and CQ given in Example~\ref{exa:query-rewriting}
(which are also given here):
\[
\begin{array}{rcl}
\sigma &:& \mathit{project}(X),\mathit{inArea}(X,Y)\ \ra\ \exists Z
\, \mathit{hasCollaborator}(Z,Y,X),\\
q &:& p(B)\ \la \mathit{hasCollaborator}(A,\mathit{db},B).
\end{array}
\]
Observe that $\head{\sigma}$ and $\body{q}$ unify, and $\gamma = \{X
\ra B, Y \ra \mathit{db}, Z \ra A\}$ is their MGU. This intuitively
means that an atom of the form
$\mathit{hasCollaborator}(t_1,\mathit{db},t_2)$, where $t_1$ and
$t_2$ are terms, to which $\body{q}$ can be homomorphically mapped,
may be obtained during the construction of the chase by applying
$\sigma$. Such a TGD chase step can be simulated (or bypassed) by
applying the rewriting step on $q$ using $\sigma$. This consists of
replacing $\body{q}$ with $\body{\sigma}$, and then applying
$\gamma$ on the obtained query. The result of such a rewriting step
is the CQ:
\[
q'\ :\ p(B)\ \la \mathit{project}(B),\mathit{inArea}(B,\mathit{db}),
\]
and the final rewriting of $q$ w.r.t.~$\{\sigma\}$ is the UCQ
$\{q,q'\}$. \hfill\markfull
\end{example}

The fact that a set $S \subseteq \body{q}$ unifies with
$\head{\sigma}$ indicates that an atom $\atom{a}$, to which $S$ can
be homomorphically mapped, may be obtained during the chase by
applying $\sigma$. However, this is not always true and may lead to
erroneous rewriting steps, which in turn will generate unsound
rewritings. Let us illustrate the two cases, via a simple example,
where the blind application of the rewriting step, without checking
whether further conditions are satisfied, leads to unsound
rewritings.

\begin{example}[Unsound Rewritings]\label{exa:unsound-rewritings}
Consider the same TGD $\sigma$ as in
Example~\ref{exa:rewriting-step}, and the CQ
\[
q_1\ :\ p(B)\ \la \mathit{hasCollaborator}(c,\mathit{db},B),
\]
where $c \in \dom$. Since $\head{\sigma}$ and $\body{q_1}$ unify,
with $\gamma = \{X \ra B,Y \ra \mathit{db},Z \ra c\}$ be their MGU,
we proceed with the rewriting step. This will result to the CQ:
\[
q'\ :\ p(B)\ \la \mathit{project}(B),\mathit{inArea}(B,\mathit{db}).
\]
Consider now the database $D =
\{\mathit{project}(a),\mathit{inArea}(a,b)\}$. The CQ $q'$ maps to
$D$ and we conclude that $\tup{a} \in q'(D)$. However, the original
query $q_1$ does not map to $\chase{D}{\{\sigma\}}$, since there is
no atom of the form $\mathit{hasCollaborator}(c,\mathit{db},t)$ in
$\chase{D}{\{\sigma\}}$, and thus $\ans{q_1}{D}{\{\sigma\}} =
\emptyset$. Therefore, any rewriting containing $q'$ is not a sound
rewriting of $q_1$ w.r.t.~$\{\sigma\}$. This is because the constant
$c$ is associated with the existentially quantified variable $Z$ and
thus, after applying the rewriting step, the information about the
constant $c$ occurring in the original query is lost.

Consider now the CQ
\[
q_2\ :\ p(B)\ \la \mathit{hasCollaborator}(B,\mathit{db},B),
\]
As above, $\head{\sigma}$ and $\body{q}$ unify, and $\gamma = \{X
\ra B,Y \ra \mathit{db},Z \ra B\}$ is their MGU. After applying the
rewriting step we get again the CQ $q'$, and $\tup{a} \in q'(D)$.
However, there is no atom of the form
$\mathit{hasCollaborator}(t,\mathit{db},t)$, i.e., an atom where the
same term occurs at the first and the last position, which means
that $\ans{q_2}{D}{\{\sigma\}} = \emptyset$. Hence, any rewriting
containing $q'$ is not a sound rewriting of $q_2$
w.r.t.~$\{\sigma\}$. The reason for this is because one occurrence
of the variable $B$ which is in a self-join, i.e., occurs more than
once in $\body{q}$, is associated with the existentially quantified
variable $Z$ and hence, after applying the rewriting step, the fact
that the variable $B$ is in a self-join is lost. \hfill\markfull
\end{example}

The blind application of the rewriting step may also cause the
generation of unsafe queries, i.e., queries where a distinguished
variable does not occur in the body. This may happen if a
distinguished variable of the query to be rewritten is associated
with an existentially quantified variable of the TGD under
consideration.
%
%
From the above informal discussion we conclude that the rewriting
step can be applied on a set $S \subseteq \body{q}$ using a TGD
$\sigma$ (or simply, $\sigma$ is applicable to $S$) if the following
hold: (1) $S$ and $\head{\sigma}$ unify; and (2) their MGU does not
associate the constants, the join variables, and the distinguished
variables of $q$ with the existentially quantified variable of
$\sigma$. This is the so-called \emph{applicability condition}, and
its formal definition will be given in the next section.
Although the applicability condition is crucial for the soundness of
the final rewriting, it may prevent the generation of queries which
are vital for the completeness of the rewriting. This is illustrated
in the following example:

\begin{example}[Incomplete Rewritings]\label{exa:incomplete-rewritings}
Consider the set $\dep$ consisting of the TGDs
\[
\begin{array}{rcl}
\sigma_1 &:& \mathit{project}(X),\mathit{inArea}(X,Y)\ \ra\ \exists Z \, \mathit{hasCollaborator}(Z,Y,X),\\
\sigma_2 &:& \mathit{hasCollaborator}(X,Y,Z)\ \ra\
\mathit{collaborator}(X),
\end{array}
\]
and the CQ
\[
q\ :\ p(B,C)\ \la\
\underbrace{\mathit{hasCollaborator}(A,B,C)}_{\atom{a}},\underbrace{\mathit{collaborator}(A)}_{\atom{b}}.
\]
The only viable strategy in this case is to apply $\sigma_2$ to
$\{\atom{b}\}$, since $\sigma_1$ is not applicable to $\{\atom{a}\}$
due to the join variable $A$. The obtained query is
\[
q'\ :\ p(B,C)\ \la\
\mathit{hasCollaborator}(A,B,C),\mathit{hasCollaborator}(A,E,F),
\]
where $E$ and $F$ are fresh variables. Notice that the variable $A$
remains a join variable, and thus $\sigma_1$ is not applicable since
the applicability condition is violated. However, $q'$ has the same
semantic meaning as
\[
q''\ :\ p(B,C)\ \la\ \mathit{hasCollaborator}(A,B,C),
\]
in which $A$ occurs only once. Since $\sigma_1$ is applicable to
$\body{q''}$ we get the query
\[
q'''\ :\ p(B,C)\ \la\ \mathit{project}(C),\mathit{inArea}(C,B).
\]
The query $q''$ is the result of unifying the body-atoms of $q'$,
and thus this unification step is critical for generating $q'''$.
Let us now show that indeed $q'''$ is crucial for the completeness
of the final rewriting. Consider the database $D =
\{\mathit{project}(a),\mathit{inArea}(a,b)\}$. Clearly,
$\chase{D}{\dep} = D \cup
\{\mathit{hasCollaborator}(z,b,a),\mathit{collaborator}(z)\}$, where
$z \in \freshdom$, and hence $\tup{b,a} \in \ans{q}{D}{\dep}$.
Observe that without the query $q'''$, there is no way to have the
tuple $\tup{b,a}$ in the answer to the final rewriting over $D$,
which implies that $q'''$ is needed for the completeness of the
rewriting. \hfill\markfull
\end{example}

From the above discussion we conclude that, apart from the rewriting
step, an additional unification step is needed to convert some join
variables into non-join ones. The purpose of this step, which we
call \emph{factorization step}, is to satisfy the applicability
condition, and thus guarantee the completeness of the final
rewriting.
To sum up, the prefect rewriting of a CQ $q$ w.r.t.~a set $\dep$ of
TGDs is computed by exhaustively applying the two steps discussed
above, namely rewriting and factorization.

\subsection{The Algorithm XRewrite}\label{sec:algorithm-tgd-rewrite}

We proceed with the formal definition of our rewriting algorithm,
called $\mathsf{XRewrite}$. Before going into the details of the
algorithm, we first need to formalize the applicability condition
and the notion of factorizability.
We assume, without loss of generality, that the variables occurring
in queries and those appearing in TGDs constitute two disjoint sets.
Given a CQ $q$, a variable is called \emph{shared} in $q$ if it
occurs more than once in $q$. Notice that the distinguished
variables of $q$ are trivially shared since, by definition, they
occur both in $\body{q}$ and $\head{q}$.

\begin{definition}[Applicability]\label{def:applicability}
Consider a CQ $q$ and a TGD $\sigma$. Given a set of atoms $S
\subseteq \body{q}$, we say that $\sigma$ is \emph{applicable} to
$S$ if the following conditions are satisfied:
\begin{enumerate}
\item the set $S \cup \{\head{\sigma}\}$ unifies, and
\item for each $\atom{a} \in S$, if the term at position $\pi$ in
$\atom{a}$ is either a constant or a shared variable in $q$, then
$\pi \neq \pi_{\exists}(\sigma)$. \hfill\markfull
\end{enumerate}
\end{definition}

Let us now focus on factorizability which will be at the basis of
the factorization step. Recall that the factorization step is
necessary in order to convert some shared variables into non-shared
ones, with the aim of satisfying the applicability condition. In
general, this can be achieved by exhaustively unifying all the atoms
that unify in the body of a query. However, some of these
unifications do not contribute in any way in satisfying the
applicability condition, and as a result many superfluous queries
are generated. We illustrate this situation by means of an example.

\begin{example}\label{exa:superfluous-queries}
Consider the following TGD and query:
\[
\sigma\ :\ s(X)\ \ra\ \exists Y \, r(X,Y) \qquad\qquad
q\ :\  p(A)\ \la\ r(A,B),r(C,B),r(B,E).
\]
Since $\sigma$ is applicable to $\{r(B,E)\}$ we obtain the query
\[
q'\ :\ p(A)\ \la\ \underbrace{r(A,B),r(C,B)}_{S},s(B).
\]
Due to the shared variable $B$, $\sigma$ is not applicable to $S$.
One can proceed with the unification of $r(A,B)$ and $r(C,B)$ in
order to make $B$ non-shared and satisfy the applicability
condition; clearly, the query
\[
q''\ :\ p(A)\ \la\ r(A,B),s(B)
\]
is obtained. However, the variable $B$ is still shared and there is
no way to make it non-shared. Thus, the unification of $r(A,B)$ and
$r(C,B)$ does not contribute in satisfying the applicability
condition, and the query $q''$ is not needed. \hfill\markfull
\end{example}

Clearly, the exhaustive unification produces a non-negligible number
of redundant queries. It is thus necessary to apply a restricted
form of factorization that generates a possibly small number of CQs
which are vital for the completeness of the rewriting algorithm.
This corresponds to the identification of all the atoms in the query
whose shared existential variables come from the same atom in the
chase, and they can be unified with no loss of information. Summing
up, the key idea underlying our notion of factorizability is as
follows: in order to apply the factorization step, there must exist
a TGD that can be applied to its output.

\begin{definition}[Factorizability]\label{def:factorizability}
Consider a CQ $q$ and a TGD $\sigma$. Given a set of atoms $S
\subseteq \body{q}$, where $|S| \geqslant 2$, we say that $S$ is
\emph{factorizable} w.r.t.~$\sigma$ if the following conditions are
satisfied:
\begin{enumerate}
\item $S$ unifies,
\item $\pi_{\exists}(\sigma) \neq \varepsilon$, and
\item there exists a variable $V \not\in \var{\body{q} \setminus S}$
which occurs in every atom of $S$ only at position
$\pi_{\exists}(\sigma)$. \hfill\markfull
\end{enumerate}
\end{definition}

\begin{example}\label{exa:factorization}
Consider the TGD $\sigma\ :\ s(X),r(X,Y)\, \rightarrow\, \exists Z
\, t(X,Y,Z)$ and the CQs
\[
\begin{array}{rcl}
q_1 &:& p(A)\, \leftarrow\, \underbrace{t(a,A,C),t(B,a,C)}_{S_1},\\
q_2 &:& p(A)\, \leftarrow\, s(C),\underbrace{t(A,B,C),t(A,E,C)}_{S_2},\\
q_3 &:& p(A)\, \leftarrow\, \underbrace{t(A,B,C),t(A,C,C)}_{S_3},
\end{array}
\]
where $a \in \dom$. The set $S_1$ is factorizable w.r.t.~$\sigma$
since the substitution $\{A \ra a, B \ra a\}$ is a unifier for
$S_1$, and also $C$ appears in both atoms of $S_1$ only at position
$\pi_{\exists}(\sigma) = t[3]$.
On the other hand, $S_2$ and $S_3$, although they unify, are not
factorizable w.r.t.~$\sigma$ since in $q_2$ the variable $C$ occurs
also outside $S_2$, while in $q_3$ the variable $C$ appears not only
at position $\pi_{\exists}(\sigma)$ but also at position $t[2]$.
\hfill\markfull
\end{example}

Let us clarify that the notion of factorizability is incomparable to
the notion of query minimization~\cite{ChMe77}. Recall that the goal
of query minimization is to construct a query which is equivalent to
the original one, and at the same time is minimal. Observe that
$q_1$, given in Example~\ref{exa:factorization}, is already minimal
since there is no endomorphism that can be applied on $q_1$ and make
it smaller, but $S_1 \subseteq \body{q_1}$ is factorizable
w.r.t.~$\sigma$ and the obtained query is $p(A) \la t(a,a,C)$ which
is not equivalent to $q_1$.
On the other hand, $q_2$ is not minimal since by applying the
endomorphism $\{E \ra B\}$ we get an equivalent and smaller query,
but the factorization step is not applied.

Having the above key notions in place, we are now ready to present
the algorithm $\mathsf{XRewrite}$, which is depicted in
Algorithm~\ref{alg:tgd-rewrite}. As said above, the perfect
rewriting of a CQ $q$ w.r.t.~a set $\dep$ of TGDs is computed by
exhaustively applying (i.e., until a fixpoint is reached) the
rewriting and the factorization steps.
Notice that the CQs which are the result of the factorization step,
are nothing else than auxiliary queries which are critical for the
completeness of the final rewriting, but are not needed in the final
rewriting. Thus, during the iterative procedure, we label the
queries with $\mathsf{r}$ (resp., $\mathsf{f}$) in order to keep
track which of them are generated by the rewriting (resp.,
factorization) step. The input query, although is not a result of
the rewriting step, is labeled by $\mathsf{r}$ since it must be part
of the final rewriting.
Moreover, once we apply exhaustively on a CQ the two crucial steps,
it is not necessary to revisit it since this will lead to redundant
queries. Hence, we also label the queries with $\mathsf{e}$ (resp.,
$\mathsf{u}$) indicating that a query is already explored (resp.,
unexplored).
Let us now describe the two main steps of the algorithm. In the
sequel, fix a triple $\tup{q,x,y}$, where $\tup{x,y} \in
\{\mathsf{r},\mathsf{f}\} \times \{\mathsf{e},\mathsf{u}\}$ (this is
how we indicate that $q$ is labeled by $x$ and $y$), and a TGD
$\sigma \in \dep$. We assume that $q$ is of the form $p(\insX) \la
\varphi(\insX,\insY)$.

\begin{algorithm}[t]
\caption{The algorithm $\mathsf{XRewrite}$ \label{alg:tgd-rewrite}}
\small
    \KwIn{a CQ $q$ over a schema $\R$ and a set $\Sigma$ of TGDs over $\R$}
    \KwOut{the perfect rewriting of $q$ w.r.t.~$\Sigma$}
    \vspace{2mm}
    $i := 0$\;
    $Q_{\textsc{rew}} := \{\langle q,\mathsf{r},\mathsf{u} \rangle\}$\;
    \Repeat{$Q_{\textsc{temp}} = Q_{\textsc{rew}}$}{
        $Q_{\textsc{temp}} := Q_{\textsc{rew}}$\;
        \ForEach{$\langle q,x,\mathsf{u} \rangle \in Q_{\textsc{temp}},$ \emph{where} $x \in \{\mathsf{r},\mathsf{f}\}$}{
            \ForEach{$\sigma \in \Sigma$}{
                \tcc{\textrm{rewriting step}}
                \ForEach{$S \subseteq \body{q}$ \emph{such that} $\sigma$ \emph{is applicable to} $S$}{
                        $i := i + 1$\;
                        $q^{\prime} := \gamma_{S,\sigma^i}(q[S/\body{\sigma^i}])$\;
                        \If{\emph{there is no} $\tup{q'',\mathsf{r},\star} \in Q_{\textsc{rew}}$ \emph{such that} $q' \simeq q''$}{
                            $Q_{\textsc{rew}} := Q_{\textsc{rew}} \cup \{\langle q',\mathsf{r},\mathsf{u} \rangle\}$\;
                        }
                }
                \tcc{\textrm{factorization step}}
                \ForEach{$S \subseteq \body{q}$ \emph{which is factorizable w.r.t.} $\sigma$}{
                    $q^{\prime} := \gamma_{S}(q)$\;
                    \If{\emph{there is no} $\tup{q'',\star,\star} \in Q_{\textsc{rew}}$ \emph{such that} $q' \simeq q''$}{
                        $Q_{\textsc{rew}} := Q_{\textsc{rew}} \cup \{\langle q^{\prime},\mathsf{f},\mathsf{u} \rangle\}$\;
                    }
                }
            }
            \tcc{\textrm{query $q$ is now explored}}
            $Q_{\textsc{rew}} := (Q_{\textsc{rew}} \setminus \{\tup{q,x,\mathsf{u}}\}) \cup
            \{\tup{q,x,\mathsf{e}}\}$\;
        }
    }
    $Q_{\textsc{fin}} := \{q ~|~ \langle q,\mathsf{r},\mathsf{e} \rangle \in Q_{\textsc{rew}}\}$\;
    \Return{$Q_{\textsc{fin}}$}
\end{algorithm}

\begin{itemize}
\item[\textbf{Rewriting Step.}] For each $S \subseteq \body{q}$ such
that $\sigma$ is applicable to $S$, the $i$-th application of the
rewriting step generates the query $q' =
\gamma_{S,\sigma^i}(q[S/\body{\sigma^i}])$, where $\sigma^i$ is the
TGD obtained from $\sigma$ by replacing each variable $X$ with
$X^i$, $\gamma_{S,\sigma^i}$ is the MGU for the set $S \cup
\{\head{\sigma^i}\}$ (which is the identity on the variables that
appear in the body but not in the head of $\sigma^i$), and
$q[S/\body{\sigma^i}]$ is obtained from $q$ be replacing $S$ with
$\body{\sigma^{i}}$, i.e., is the query with $p(\insX)$ as its head
and $(\varphi(\insX,\insY) \setminus S) \cup \body{\sigma^{i}}$ as
its body.
By considering $\sigma^i$ (instead of $\sigma$) we actually rename,
using the integer $i$, the variables of $\sigma$. This renaming step
is needed in order to avoid undesirable clutters among the variables
introduced during different applications of the rewriting step.
Finally, if the there is no $\tup{q'',\mathsf{r},\star} \in
Q_{\textsc{rew}}$, i.e., an (explored or unexplored) query which is
a result of the rewriting step, such that $q'$ and $q''$ are the
same (modulo bijective variable renaming), denoted $q' \simeq q''$,
then $\tup{q',\mathsf{r},\mathsf{u}}$ is added to $Q_{\textsc{rew}}$.\\

\item[\textbf{Factorization Step.}] For each $S \subseteq \body{q}$
which is factorizable w.r.t. $\sigma$, the factorization step
generated the query $q' = \gamma_S(q)$, where $\gamma_S$ is the MGU
for $S$. Then, if there is no $\tup{q'',\star,\star} \in
Q_{\textsc{rew}}$, i.e., a query which is a result of the rewriting
or the factorization step, and is explored or unexplored, such that
$q' \simeq q''$, then $\tup{q', \mathsf{f},\mathsf{u}}$ is added to
$Q_{\textsc{rew}}$.
\end{itemize}

It is important to say that, if the input set of TGDs is sticky,
then both $\gamma_{S,\sigma^i}$ and $\gamma_S$ are defined in such a
way that, for each of their mapping $V \ra U$, $V \in \var{q}$
implies $U \in \var{q}$; there existence is guaranteed by stickiness
(see the proof of Lemma~\ref{lem:property}). The reason why we
employ these MGUs (instead of arbitrary ones) is to ensure a crucial
syntactic property of each query generated during the rewriting
process (see Lemma~\ref{lem:property}), which in turn will allow us
to establish the termination of $\mathsf{XRewrite}$ under sticky
sets of TGDs.
Before we proceed further, let us briefly discuss the relationship
of our approach, and the one employed in~\cite{KLMT12} which is
based on the so-called piece-unifier. Roughly, a piece-based
rewriting step, the building block of the algorithm
in~\cite{KLMT12}, simulates a factorization and a rewriting step of
$\mathsf{XRewrite}$. Let us illustrate this via a simple example.

\begin{example}
Consider the TGD and the CQ
\[
\sigma\ :\ r(X)\ \ra\ \exists Y \, s(X,Y) \qquad q\ :\ p\ \la\
\underbrace{s(A,B),s(C,B),s(C,D)}_{S},t(A,C).
\]
A pair $(S,\gamma)$, where $\gamma$ is an MGU for the set $S \cup
\{\head{\sigma}\}$, is called piece-unifier of $q$ with $\sigma$ if
\emph{(i)} the universally quantified variables of $\sigma$, denoted
$\mathit{var}_{\forall}(\sigma)$, are mapped by $\gamma$ to
$\mathit{var}_{\forall}(\sigma)$, and \emph{(ii)} each variable of
$\var{S} \cap \var{\body{q} \setminus S}$ is mapped by $\gamma$ to
$\mathit{var}_{\forall}(\sigma)$. Such an MGU is $\gamma = \{A \ra
X, B \ra Y, C \ra X, D \ra Y\}$.
The existence of the piece-unifier $(S,\gamma)$ implies that $S$ can
be rewritten at a single (piece-based) rewriting step using
$\sigma$, and the query $q' : p \la r(X),t(X,X)$ is obtained.

Now, observe that the set $\{s(A,B),s(C,B)\} \subseteq  \body{q}$ is
factorizable w.r.t.~$\sigma$, and after applying the factorization
step we get the query $p\ \la\ s(A,B),s(C,D),t(A,C)$. Then, $\sigma$
is applicable to $\{s(A,B),s(C,D)\}$, and after applying the
rewriting step we get the query $p\ \la\ r(A),t(A,A)$ which
coincides (modulo variable renaming) with $q'$. \hfill\markfull
\end{example}


\subsection{Termination of XRewrite}\label{sec:termination-of-tgd-rewrite}

Let us now establish the termination of $\mathsf{XRewrite}$.
We first establish a key syntactic property of the constructed
rewritten query. In the sequel, for notational convenience, given a
CQ $q$ and a set $\dep$ of TGDs, we denote by $q_{\dep}$ the
rewritten query $\mathsf{XRewrite}(q,\dep)$.

\begin{lemma}\label{lem:property}
Consider a CQ $q$ over a schema $\R$, and a set $\dep$ of TGDs over
$\R$. For each $q' \in q_{\dep}$ the following hold:
\begin{enumerate}
\item If $\dep \in \mathsf{LINEAR}$, then $|\body{q}| \geqslant
|\body{q'}|$, and

\item If $\dep \in \mathsf{STICKY}$, then every variable of $(\var{q'} \setminus \var{q})$ occurs only once in
$q'$.
\end{enumerate}
\end{lemma}

\begin{proof}
Part (1) follows immediately by definition of linear TGDs. In
particular, since each linear TGD has only one body-atom, during the
rewriting step we replace a set of atoms in the body of the CQ under
consideration with a single atom. Notice that during the
factorization step, since we unify atoms, we always decrease the
number of atoms in the body of the CQ.

Part (2) is established by induction on the number of applications
of the rewriting and factorization steps. We denote by
$q_{\dep}^{i}$ the part of $q_\dep$ obtained after $i$ applications
either of the factorization or the rewriting step. The proof is by
induction on $i \geqslant 0$.

\underline{Base step:} Clearly, $q_{\dep}^{0} = q$, and the claim
holds trivially.

\underline{Inductive step:} In case that $q_{\dep}^{i+1} =
q_{\dep}^{i}$, where $i > 0$, the claim follows immediately by
induction hypothesis. The interesting case is when $q_{\dep}^{i+1} =
q_{\dep}^{i} \cup \{p'\}$, where $p'$ was obtained from a CQ $p \in
q_{\dep}^{i}$ by applying either the rewriting or the factorization
step. Henceforth, we refer to the variables (not occurring in $q$)
introduced during the rewriting process as \emph{new variables}. We
identify the following two cases.

Case 1: First, assume that $p'$ was obtained during the $j$-th
application of the rewriting step, where $j \leqslant i+1$, because
the TGD $\sigma \in \dep$ is applicable to a set $S \subseteq
\body{p}$. Since, by induction hypothesis, each new variable in $S$
occurs only once, we can assume, without loss of generality, that,
for each mapping $V \ra U$ of $\gamma_{S,\sigma^j}$, $U$ is not a
new variable introduced during the first $j-1$ applications of the
rewriting step. Recall that, by construction, for each $V \ra U$ of
$\gamma_{S,\sigma^j}$, $V \in \var{q}$ implies $U \in \var{q}$. It
is easy to see that such a MGU always exists. In particular, if
$\gamma_{S,\sigma^j}$ does not satisfy the above condition, then we
can redefine it as $\mu \circ \gamma_{S,\sigma^j}$, where $\mu$ is
constructed as follows: for each $V \ra U$ of $\gamma_{S,\sigma^j}$,
if $V \in \var{q}$, $U \not\in \var{q}$ and there is no mapping $U
\ra V'$ in $\mu$, then we add to $\mu$ the mapping $U \ra V$.
We proceed by case analysis on the reason why a new variable may
appear in $p'$. We identify the following two cases:
\begin{enumerate}
\item A variable $V$ occurs in $\body{\sigma^j}$ but not in $\head{\sigma^j}$.
By construction, $V \ra U \in \gamma_{S,\sigma^{j}}$ implies $U =
V$. Thus, $V$ is a new variable that appears in $\body{p'}$. Since
$\dep \in \mathsf{STICKY}$, $V$ occurs in $\body{\sigma^j}$ only
once, and hence $V$ appears in $p'$ only once.

\item A new variable $V \in \var{S}$, $\gamma_{S,\sigma^{j}}(V) = U$, where $U$
occurs in the body and in the head of $\sigma^j$, and there is no
assertion $U \ra V'$ in $\gamma_{S,\sigma^{j}}$, where $V' \in
\var{q}$.
By induction hypothesis, $V$ occurs only once in $p$, and thus does
not occur in $p'$. Since $U$ does not appear in the left-hand side
of an assertion of $\gamma_{S,\sigma^{j}}$, we get that $U$ is a new
variable that appears in $\body{p'}$ due to the fact that it occurs
in $\body{\sigma^j}$ and $\head{\sigma^j}$. Notice that $U$, after
applying $\mathsf{SMarking}$, is marked; thus, $U$ occurs only once
in $\body{\sigma^j}$ since $\dep \in \mathsf{STICKY}$. This implies
that $U$ appears in $p'$ only once.
\end{enumerate}

Case 2: Now, suppose that $p'$ was obtained by applying the
factorization step. This implies that there exists a set $S
\subseteq \body{p}$, where $|S| \geqslant 2$, that unifies, and $p'
= \gamma_S(p)$. Recall that, by construction, for each mapping $V
\ra U$ of $\gamma_S$, $V \in \var{q}$ implies $U \in \var{q}$. The
existence of such a MGU is guaranteed since, by induction
hypothesis, each new variable in $S$ occurs only once; in fact,
$\gamma_S$ can be defined as the MGU for $S'$, where $S'$ is
obtained as follows: if a new variable $W$ occurs in an atom
$\atom{a} \in S$ at position $\pi$, and there exists a set
$\{\atom{b}_1,\ldots,\atom{b}_n\}$, where $n \geqslant 1$, such that
at position $\pi$ of each $\atom{b}_i$ a variable $W_i \in \var{q}$
occurs, then replace $W$ with $W_1$. It is now straightforward to
see, by definition of $\gamma_S$, that each new variable in $p'$
occurs only once.
%
\end{proof}


We now show that our rewriting algorithm terminates under linear and
sticky TGDs:

\begin{theorem}\label{the:tgdrewrite-termination}
Consider a CQ $q$ over a schema $\R$, and a set $\dep$ of TGDs over
$\R$. If $\dep \in \mathsf{LINEAR}$ or $\dep \in \mathsf{STICKY}$,
then $\mathsf{XRewrite}(q,\dep)$ terminates.
\end{theorem}

\begin{proof}
Assume first that $\dep \in \mathsf{LINEAR}$. By
Lemma~\ref{lem:property}, we get that, for each $q' \in q_\dep$,
$|\body{q}| \geqslant |\body{q'}|$. This implies that each $q' \in
q_\dep$ can be equivalently rewritten as a CQ with at most $k =
|\body{q}| \cdot \arity{\R}$ variables. Therefore, $q_\dep$ contains
(modulo variable renaming) at most $k$ variables. Since the maximum
number of CQs that can be constructed using $k$ variables and $|\R|$
predicates is finite, and also since the algorithm does not drop
queries that it has generated, the claim follows.

Suppose now that $\dep \in \mathsf{STICKY}$. Given a CQ $p \in
q_\dep$, let $p^{\star}$ be the query obtained from $p$ by replacing
each variable of $\var{p} \setminus \var{q}$ with the symbol
$\star$. Since, by Lemma~\ref{lem:property}, each variable of
$\var{p} \setminus \var{q}$ occurs only once in $p$, we get the
following: for each pair of CQs $p_1$ and $p_2$ of $q_\dep$, if
$p_{1}^{\star} = p_{2}^{\star}$, then $p_1$ and $p_2$ are the same
modulo bijective variable renaming. Therefore, the maximum number of
CQs that can be constructed during the execution of
$\mathsf{XRewrite}$ is bounded by the number of different CQs that
can be constructed using terms of $T = (\adom{q} \cup \{\star\})$
and predicates of $\R$. Since both $T$ and $\R$ are finite, and also
since the algorithm does not drop queries that it has generated, we
conclude that $\mathsf{XRewrite}$ terminates under sticky sets of
TGDs.
\end{proof}

Clearly, the check that the obtained query is not already present
(modulo bijective variable renaming) each time the rewriting or the
factorization step is applied, is crucial in order to guarantee the
termination of $\mathsf{XRewrite}$.
An alternative way, which is actually the one that we employ in the
implementation of our algorithm, is to maintain an auxiliary set of
CQs $Q_{\mathit{can}}$ which stores the generated queries in a
canonical form, i.e., after applying a canonical renaming step, and
run the algorithm until a fixpoint of $Q_{\mathit{can}}$ is reached.
Formally, given a CQ $q$, assuming that $\dep$ is the input set of
TGDs and $\R$ the underlying schema, a canonical renaming
$\mathit{can_q} : \adom{\body{q}} \ra (\Gamma_q \cup \Delta_{q})$,
where $\Gamma_q \subset \dom$ are the constants occurring in $q$,
and $\Delta_q \subset \freshdom$ is such that $(\Delta_q \cap
\var{q}) = \emptyset$, $|\Delta_q| = |\body{q}| \cdot \arity{\R}$ if
$\dep \in \mathsf{LINEAR}$, and $|\Delta_q| = |\R| \cdot
(|\adom{q}|+1)^{\arity{\R}} \cdot \arity{\R}$ if $\dep \in
\mathsf{STICKY}$, is a one-to-one substitution which maps each
constant of $\Gamma_q$ to itself, and each variable of $\var{q}$ to
the first unused element of $\Delta_q$; a lexicographic order is
assumed on $\Delta_q$.
It is easy to see that, given two CQs $q$ and $p$,
$\mathit{can_q}(q) = \mathit{can_p}(p)$ implies that $q$ and $p$ are
the same query (modulo bijective variable renaming).

\subsection{The Size of the Rewriting}\label{sec:size-of-rewriting}

By exploiting the analysis in the proof of
Theorem~\ref{the:tgdrewrite-termination}, it is easy to establish an
upper bound on the size of the rewriting constructed by
$\mathsf{XRewrite}$.

\begin{theorem}\label{the:size-of-rewriting}
Consider a CQ $q$ over a schema $\R$, and a set $\dep$ of TGDs over
$\R$. The following hold:
\begin{enumerate}
\item $|q_{\dep}|
\in \O\left(\left(|\R| \cdot (\arity{\R} \cdot
|\body{q}|)^{\arity{\R}}\right)^{|\body{q}|}\right)$ if $\dep \in
\mathsf{LINEAR}$, and

\item $|q_{\dep}|
\in 2^{\O\left(|\R| \cdot \left(\arity{\R} \cdot
|\body{q}|\right)^{\arity{\R}}\right)}$ if $\dep \in
\mathsf{STICKY}$.
\end{enumerate}
\end{theorem}

\begin{proof}
Assume first that $\dep \in \mathsf{LINEAR}$. As discussed in the
proof of Theorem~\ref{the:tgdrewrite-termination}, the number of
variables that can appear in $q_{\dep}$ is bounded by $(\arity{\R}
\cdot |\body{q}|)$. Thus, the number of atoms that can appear in
$q_{\dep}$ is at most $|\R| \cdot (\arity{\R} \cdot
|\body{q}|)^{\arity{\R}}$. Since $|\body{q'}| \leqslant |\body{q}|$,
for each $q' \in q_{\dep}$, we immediately get that $|q_\dep|
\leqslant (|\R| \cdot (\arity{\R} \cdot
|\body{q}|)^{\arity{\R}})^{|\body{q}|}$, and part (1) follows.
Assume now that $\dep \in \mathsf{STICKY}$. As discussed in the
proof of Theorem~\ref{the:tgdrewrite-termination}, the number of
variables that can appear in $q_{\dep}$ is bounded by
$|\mathit{terms}(q)| + 1 \leqslant (\arity{\R} \cdot |\body{q}|) +
1$, and hence the number of atoms that can appear in $q_\dep$ is at
most $|\R| \cdot ((\arity{\R} \cdot |\body{q}|) + 1)^{\arity{\R}}$.
Since a CQ $q' \in q_\dep$ can have in its body any subset of those
atoms, we conclude that $|q_\dep| \leqslant 2^{(|\R| \cdot
((\arity{\R} \cdot |\body{q}|) + 1)^{\arity{\R}})}$, and part (2)
follows.
\end{proof}

An interesting question is whether the exponential (resp.,
double-exponential) size of the UCQ-rewriting is unavoidable when we
consider linear (resp., sticky) sets of TGDs. In what follows, we
give an affirmative answer to this question.

\begin{theorem}
The following hold:
\begin{enumerate}
\item There exists a CQ $q$ over a schema $\R$, and a set $\dep \in
\mathsf{LINEAR}$ over $\R$ such that, for any UCQ-rewriting $Q$ of
$q$ w.r.t. $\dep$, $|Q| \in
\Omega\left(\left(|\R|\right)^{|\body{q}|}\right)$,

\item There exists a CQ $q$ over a schema $\R$, and a set $\dep \in
\mathsf{STICKY}$ over $\R$ such that, for any UCQ-rewriting $Q$ of
$q$ w.r.t. $\dep$, $|Q| \in
\Omega\left(2^{\left(2^{\arity{\R}}\right)}\right)$.
\end{enumerate}
\end{theorem}

\begin{proof}
For part (1), let $\R = \{p_0,\ldots,p_m\}$ and consider the CQ and
the set of TGDs
\[
q\ :\ p\ \la\ p_0(A_1),\ldots,p_0(A_n) \qquad \dep\ =\
\left\{p_i(X)\ \ra\ p_0(X)\right\}_{i \in [m]}.
\]
It is not difficult to see that any UCQ-rewriting of $q$
w.r.t.~$\dep$ must contain a CQ $q'$ such that $\body{q'} \in
\left(\{p_i(A_1)\}_{i \in [m]} \times \{p_i(A_2)\}_{i \in [m]}
\times \ldots \times \{p_i(A_n)\}_{i \in [m]}\right)$. Since the
cardinality of the above set is $m^n = (|\R|)^{|\body{q}|}$, the
claim follows.

For part (2), let $\R = \{p_0,\ldots,p_n,s,r\}$ and consider the
atomic CQ $q : p \la p_0(0,\ldots,0)$, where $p_0$ is an $n$-ary
predicate, and the sticky set $\dep$ of TGDs
\[
\begin{array}{l}
\{p_i(X_1,\ldots,X_{i-1},0,X_{i+1},\ldots,X_n),p_i(X_1,\ldots,X_{i-1},1,X_{i+1},\ldots,X_n)\\
\hspace{60mm} \ra\ p_i(X_1,\ldots,X_{i-1},0,X_{i+1},\ldots,X_n)\}_{i \in [n]},\\
\{s_i(X_1,\ldots,X_n)\ \ra\ p_n(X_1,\ldots,X_n)\}_{i \in [2]}.
\end{array}
\]
%
It is easy to verify that any UCQ-rewriting of $q$ w.r.t.~$\dep$
must contain a CQ $q'$ such that $\body{q'} \in \times_{\tuple{t}
\in \{0,1\}^n} \{s_1(\tuple{t}),s_2(\tuple{t})\}$, and
$|\times_{\tuple{t} \in \{0,1\}^n}
\{s_1(\tuple{t}),s_2(\tuple{t})\}| = 2^{\left(2^n\right)} =
2^{\left(2^{\arity{\R}}\right)}$.
\end{proof}

\subsection{Correctness of XRewrite}\label{sec:correctness-of-tgd-rewrite}

We now establish the correctness of $\mathsf{XRewrite}$. Towards
this aim two auxiliary technical lemmas are needed. The first one,
which is used for soundness, states that the answer to the final
rewriting is a subset of the answer to the input query. In what
follows,
let $\insX^i$ be the sequence of variables obtained by replacing
each variable $X$ of $\insX$ with $X^i$.

\begin{lemma}\label{lem:sound-auxiliary-lemma}
Consider a CQ $q$ over a schema $\R$, a database $D$ for $\R$, and a
set $\dep$ of TGDs over $\R$. It holds that,
$\ans{q_{\dep}}{D}{\dep} \subseteq \ans{q}{D}{\dep}$.
\end{lemma}

\begin{proof}
It suffices to show that, for a tuple of constants $\tuple{t}$,
$\tuple{t} \in \ans{q_\dep}{D}{\dep}$ implies $\tuple{t} \in
\ans{q}{D}{\dep}$, or, equivalently, $\tuple{t} \in
q_{\dep}(\chase{D}{\dep})$ implies $\tuple{t} \in
q(\chase{D}{\dep})$.
It is straightforward to see that the factorization step does not
affect the soundness of our algorithm. Thus, we assume, without loss
of generality, that $q_\dep$ is the UCQ constructed without applying
the factorization step.
We denote by $q_{\dep}^{i}$ the part of $q_{\dep}$ obtained after $i
\geqslant 0$ applications of the rewriting step. The proof is by
induction on $i$.

\underline{Base step:} Clearly, $q_{\dep}^{0} = q$, and the claim
holds trivially.

\underline{Inductive step:} Suppose now that $\tuple{t} \in
q_{\dep}^{i}(\chase{D}{\dep})$, for $i \geqslant 0$. This implies
that there exists $p \in q_{\dep}^{i}$ and a homomorphism $h$ such
that $h(\body{p}) \subseteq \chase{D}{\dep}$ and $h(\insV) =
\tuple{t}$, where $\insV$ are the distinguished variables of $p$.
If $p \in q_{\dep}^{i-1}$, then the claim follows by induction
hypothesis. The interesting case is when $p$ was obtained during the
$i$-th application of the rewriting step from a CQ $p' \in
q_{\dep}^{i-1}$, i.e., $q_{\dep}^{i} = q_{\dep}^{i-1} \cup \{p\}$.
By induction hypothesis, it suffices to show that $\tuple{t} \in
q_{\dep}^{i-1}(\chase{D}{\dep})$.
Clearly, there exists a TGD $\sigma \in \dep$ of the form
$\varphi(\insX,\insY) \ra \exists Z \, r(\insX,Z)$ which is
applicable to a set $S \subseteq \body{p'}$, and $p$ is the query
$\gamma(p'[S / \body{\sigma^i}])$; let $\gamma$ be the MGU for $S
\cup \{\head{\sigma^i}\}$.
Observe that $h(\gamma(\varphi(\insX^i,\insY^i))) \subseteq
\chase{D}{\dep}$, and hence $\sigma$ is applicable to
$\chase{D}{\dep}$; let $\mu = h \circ \gamma$. Thus,
$\mu'(r(\insX^i,Z^i)) \in \chase{D}{\dep}$, where $\mu' \supseteq
\mu|_{\insX^i}$. We define the substitution $h' = h \cup
\{\gamma(Z^i) \ra \mu'(Z^i)\}$.
To establish that $h'$ is well-defined, it suffices to show that
$\gamma(Z^i) \not\in \dom$, and also that there is no mapping $V \ra
U \in h$ such that $\gamma(Z^i) = V$. Towards a contradiction,
suppose that $\gamma(Z^i)$ is either a constant or appears in the
left-hand side of an assertion of $h$. It is easy to verify that in
this case there exists an atom $\atom{a} \in S$ such that at
position $\pi_{\exists}(\sigma)$ in $\atom{a}$ occurs either a
constant or a variable which is shared in $p'$. But this contradicts
the fact that $\sigma$ is applicable to $S$, and hence $h'$ is
well-defined.
It remains to show that the substitution $h' \circ \gamma$ maps
$\body{p'}$ to $\chase{D}{\dep}$ and $h'(\gamma(\insV')) =
\tuple{t}$, where $\insV'$ are the distinguished variables of $p'$;
this immediately implies that $\tuple{t} \in
q_{\dep}^{i-1}(\chase{D}{\dep})$.
Clearly, $\gamma(\body{p'} \setminus S) \subseteq \body{p}$. Since
$h(\body{p}) \subseteq \chase{D}{\dep}$, we get that
$h'(\gamma(\body{p'} \setminus S)) \subseteq \chase{D}{\dep}$.
Moreover, $h'(\gamma(S)) = h'(\gamma(r(\insX^i,Z^i))) =
r(h'(\gamma(\insX^i)),h'(\gamma(Z^i))) = r(\mu(\insX^i),\mu'(Z^i)) =
\mu'(r(\insX^i,Z^i)) \in \chase{D}{\dep}$. Finally, since
$\gamma(\insV') = \insV$ and $h(\insV) = \tuple{t}$, we get that
$h'(\gamma(\insV')) = \tuple{t}$. 
\end{proof}

The second auxiliary lemma asserts that the answer to the final
rewriting is a subset of the set of tuples obtained by simply
evaluating it over the input database.

\begin{lemma}\label{lem:complete-auxiliary-lemma}
Consider a CQ $q$ over a schema $\R$, a database $D$ for $\R$, and a
set $\dep$ of TGDs over $\R$. It holds that,
$\ans{q_{\dep}}{D}{\dep} \subseteq q_{\dep}(D)$.
\end{lemma}

\begin{proof}
It suffices to show that, for a tuple of constants $\tuple{t}$,
$\tuple{t} \in \ans{q_\dep}{D}{\dep}$ implies $\tuple{t} \in q(D)$,
or, equivalently, $\tuple{t} \in q_{\dep}(\chase{D}{\dep})$ implies
$\tuple{t} \in q_{\dep}(D)$.
We proceed by induction on the number of
applications of the chase step.

\underline{Base step:} Clearly, $\apchase{0}{D}{\dep} = D$, and the
claim holds trivially.

\underline{Inductive step:} Suppose now that $\tuple{t} \in
q_{\dep}(\apchase{i}{D}{\dep})$, for $i \geqslant 0$. This implies
that there exists $p \in q_{\dep}$ and a homomorphism $h$ such that
$h(\body{p}) \subseteq \apchase{i}{D}{\dep}$ and $h(\insV) =
\tuple{t}$, where $\insV$ are the distinguished variables of $p$. If
$h(\body{p}) \subseteq \apchase{i-1}{D}{\dep}$, then the claim
follows by induction hypothesis. The non-trivial case is when the
atom $\atom{a}$, obtained during the $i$-th application of the chase
step by applying a TGD $\sigma : \varphi(\insX,\insY) \ra \exists Z
\, r(\insX,Z)$, belongs to $h(\body{p})$. Clearly, there exists a
homomorphism $\mu$ such that $\mu(\varphi(\insX,\insY)) \subseteq
\apchase{i-1}{D}{\dep}$ and $\atom{a} = \mu'(r(\insX,\insY))$, where
$\mu' \supseteq \mu|_{\insX}$. By induction hypothesis, it suffices
to show that $\tuple{t} \in q_{\dep}(\apchase{i-1}{D}{\dep})$.
Before we proceed further, we need an auxiliary claim; its proof can
be found in Section~\ref{appsec:complete-auxiliary-claim}.

\begin{claim}\label{cla:auxiliary-claim}
There exists a CQ $p' \in q_{\dep}$ and a set of atoms $S \subseteq
\body{p'}$ such that $\sigma$ is applicable to $S$, and also there
exists a homomorphism $\lambda$ such that $\lambda(\body{p'}
\setminus S) \subseteq \apchase{i-1}{D}{\dep}$, $\lambda(\insV') =
\tuple{t}$, where $\insV'$ are the distinguished variables of $p'$,
and $\lambda(S) = \atom{a}$.
\end{claim}

The above claim implies that there exists $i \geqslant 1$ such that
during the rewriting process eventually we will get a CQ $p''$ with
$\body{p''} = \gamma(\body{p'} \setminus S) \cup
\gamma(\varphi(\insX^i,\insY^i))$, where $\gamma$ is the MGU for $S
\cup \{\head{\sigma^i}\}$. It remains to show that there exists a
homomorphism that maps $\body{p''}$ to $\apchase{i-1}{D}{\dep}$ and
the distinguished variables $\insV''$ of $p''$ to $\tuple{t}$. Since
$\lambda \cup \mu'$ is a well-defined substitution, it is a unifier
for $S \cup \{\head{\sigma^i}\}$. By definition of the MGU, there
exists a substitution $\theta$ such that $\lambda \cup \mu' = \theta
\circ \gamma$. Observe that $\theta(\body{p''}) =
\theta(\gamma(\body{p'} \setminus S) \cup
\gamma(\varphi(\insX^i,\insY^i))) = (\lambda \cup \mu')(\body{p'}
\setminus S) \cup (\lambda \cup \mu')(\varphi(\insX^i,\insY^i)) =
\lambda(\body{p'} \setminus S) \cup \mu'(\varphi(\insX^i,\insY^i))
\subseteq \apchase{i-1}{D}{\dep}$. Finally, $\theta(\insV'') =
\theta(\gamma(\insV')) = (\lambda \cup \mu')(\insV') =
\lambda(\insV') = \tuple{t}$.
\end{proof}


We are now ready to establish the soundness and completeness of
$\mathsf{XRewrite}$:

\begin{theorem}\label{the:TGD-rewrite-sound-complete}
Consider a CQ $q$ over a schema $\R$, a database $D$ for $\R$, and a
set $\dep$ of TGDs over $\R$. It holds that, $q_{\dep}(D) =
\ans{q}{D}{\dep}$.
\end{theorem}

\begin{proof}
Since $D \subseteq \chase{D}{\dep}$, by monotonicity of CQs,
$q_{\dep}(D) \subseteq q_{\dep}(\chase{D}{\dep})$ which in turn
implies $q_{\dep}(D) \subseteq \ans{q_{\dep}}{D}{\dep}$. By
Lemma~\ref{lem:sound-auxiliary-lemma}, we immediately get that
$q_{\dep}(D) \subseteq \ans{q}{D}{\dep}$.
Conversely, since $q \in q_{\dep}$, we get that $\ans{q}{D}{\dep}
\subseteq \ans{q_{\dep}}{D}{\dep}$.
Lemma~\ref{lem:complete-auxiliary-lemma} implies that
$\ans{q}{D}{\dep} \subseteq q_{\dep}(D)$, and the claim follows.
\end{proof}

Let us conclude this section by noticing that $\mathsf{XRewrite}$
can treat even more expressive classes of TGDs than linear and
sticky TGDs, namely multi-linear~\cite{CaGL12} and
sticky-join~\cite{CaGP12} TGDs, which guarantee the first-order
rewritability of CQ answering.
The goal of multi-linearity was the definition of a natural
formalism which is strictly more expressive than
DL-Lite$_{\R,\sqcap}$, that is, the extended version of DL-Lite$_\R$
which allows for concept conjunction~\cite{CGLL*13}. Sticky-joiness
is the result of combining linearity and stickiness, with the aim of
identifying more expressive classes of TGDs.
For more details, we refer the reader to
Section~\ref{appsec:more-expressive-classes}.

\section{Parallelize the Rewriting Procedure}\label{sec:parallelize}

An interesting question that comes up is whether the overall time
that we need to compute the final rewriting can be reduced by
designing a parallel version of $\mathsf{XRewrite}$ which exploits
multi-core architectures. In this section, we present some
preliminary ideas and results regarding the parallelization of our
algorithm --- to the best of our knowledge, this is the first
attempt to design a parallel rewriting algorithm.
The key idea is to decompose the query $q$ into smaller queries
$q_1,\ldots,q_m$, where $m \geqslant 1$, in such a way that, if a
variable $V$ occurs in at least two queries of $\{q_1,\ldots,q_m\}$,
then each occurrence of $V$ occurs at a position that may host only
constants (in the instance constructed by the chase procedure). This
allows us to rewrite independently each query $q_i$ into $Q_{q_i}$,
where $i \in [m]$, and then merge the queries
$Q_{q_1},\ldots,Q_{q_m}$ in order to obtain the final rewriting.
Notice that the decomposition technique described above is a new
form of query decomposition which, in contrast to traditional
methods such as the ones in~\cite{ChRa00,GoLS02}, takes into account
a given set of TGDs, and is engineered to be used for parallelizing
our rewriting algorithm. Instead, the aim of existing techniques is
to suggest an efficient strategy for executing the given query. Let
us first give an informal description of our parallel procedure.

\subsection{An Informal Description}\label{sec:xrewriteparalle-informal-description}

Consider the following relational schema representing financial
information about companies and their stocks:
\[
\begin{array}{lcl}
\mathit{stock}(\mathsf{id},\mathsf{name},\mathsf{unit\_price})
&\quad&
\mathit{company}(\mathsf{name},\mathsf{country},\mathsf{segment})\\
\mathit{listComponent}(\mathsf{stock},\mathsf{list}) &\quad&
\mathit{stockPortfolio}(\mathsf{company},\mathsf{stock},\mathsf{quantity})\\
\mathit{finIndex}(\mathsf{name},\mathsf{type},\mathsf{reference\_market})
&\quad& \mathit{hasStock}(\mathsf{stock},\mathsf{comany})\\
\mathit{finInstrument}(\mathsf{stock}) &\quad&
\mathit{legalPerson}(\mathsf{company}).
\end{array}
\]
Let $\dep$ be the set consisting of the following linear TGDs; for
clarity, we use more than one existentially quantified variables in
the rule-heads:
\[
\begin{array}{rcl}
\sigma_1 &:& \mathit{stockPortfolio(X,Y,Z)}\ \rightarrow\ \exists V \exists W \, \mathit{company(X,V,W)}\\
\sigma_2 &:& \mathit{stockPortfolio(X,Y,Z)}\ \rightarrow \exists V \exists W \, \mathit{stock(Y,V,W)}\\
\sigma_3 &:& \mathit{listComponent(X,Y)}\ \rightarrow\ \exists Z \exists W \, \mathit{finIndex(Y,Z,W)}\\
\sigma_4 &:& \mathit{listComponent(X,Y)}\ \rightarrow\ \exists Z \exists W \, \mathit{stock(X,Z,W)}\\
\sigma_5 &:& \mathit{stockPortfolio(X,Y,Z)}\ \rightarrow\ \mathit{hasStock(Y,X)}\\
\sigma_6 &:& \mathit{hasStock(X,Y)}\ \rightarrow\ \exists Z \, \mathit{stockPortfolio(Y,X,Z)}\\
\sigma_7 &:& \mathit{stock(X,Y,Z)}\ \rightarrow\ \exists V \exists W \, \mathit{stockPortfolio(V,X,W)}\\
\sigma_8 &:& \mathit{stock(X,Y,Z)}\ \rightarrow\ \mathit{finInstrument(X)}\\
\sigma_9 &:& \mathit{company(X,Y,Z)}\ \rightarrow\
\mathit{legalPerson(X)}.
\end{array}
\]

The TGDs $\sigma_1$, $\sigma_2$, $\sigma_3$ and $\sigma_4$ set the
``domain'' and the ``range'' of the $\mathit{stockPortfolio}$ and
$\mathit{listComponent}$ relations, respectively. The TGDs
$\sigma_5$ and $\sigma_6$ assert that $\mathit{stockPortfolio}$ and
$\mathit{hasStock}$ are ``inverse relations'', while $\sigma_7$
expresses that each stock must belong to a stock portfolio. The TGDs
$\sigma_8$ and $\sigma_9$ model taxonomic relationships; in
particular, each stock is a financial instrument, and each company
is a legal person.
Consider also the following conjunctive query $q$ asking for all the
triples $\tup{a,b,c}$, where $a$ is a financial instrument owned by
the company $b$ and listed on $c$:
\[
\begin{array}{rcl}
p(A,B,C) &\leftarrow& \mathit{finInstrument}(A),
\mathit{stockPortfolio}(B,A,D), \mathit{company}(B,E,F),\\
&&\mathit{listComponent}(A,C), \mathit{finIndex}(C,G,H).
\end{array}
\]

Recall that our intention is to decompose $q$ into smaller
subqueries in such a way that, if a variable $V$ occurs in at least
two such subqueries, then each occurrence of $V$ occurs at a
position that may host only constants (in the instance constructed
by the chase procedure).
After a careful inspection of the set $\dep$, it is easy to verify
that, for every database $D$, if $q$ is mapped to $\chase{D}{\dep}$
via a homomorphism $h$, then the only join-variable occurring in $q$
that can be mapped by $h$ to a null value is $B$. More precisely,
due to $\sigma_7$ a null value may appear at position
$\mathit{stockPortfolio}[1]$, which in turn may be propagated to
position $\mathit{company}[1]$ after applying $\sigma_1$ --- those
positions are called affected w.r.t.~$\sigma_7$, which intuitively
means that they can have a null generated by $\sigma_7$.
The fact that only $B$ appears at an affected position, allows us to
decompose $q$ into four subqueries, and then rewrite each one of
them independently.
The result of such a decomposition, called existential-join
decomposition, is the following:
\[
\begin{array}{rcl}
q_1 &:& p_1(A)\ \la\ \mathit{finInstrument}(A)\\
q_2 &:& p_2(A,B)\ \la\ \mathit{stockPortfolio}(B,A,D),
\mathit{company}(B,E,F)\\
q_3 &:& p_3(A,C)\ \la\ \mathit{listComponent}(A,C)\\
q_4 &:& p_4(C)\ \la\ \mathit{finIndex}(C,G,H).
\end{array}
\]
Notice that, for each subquery $q_i$, the distinguished variables of
$q_i$ are the shared variables of $q$ which appear outside
$\body{q_i}$, i.e., in $\head{q}$ or in $\body{q} \setminus
\body{q_i}$.
%
%
%
%
The rewriting of $q_i$ w.r.t.~$\dep$, for each $i \in [4]$, is
denoted $Q_{q_i}$.
%
The last step is to merge the queries $Q_{q_1},\ldots,Q_{q_4}$. This
can be done via the reconciliation rule
\[
\rho\ :\ p(A,B,C)\ \la\ p_1(A),p_2(A,B),p_3(A,C),p_4(C),
\]
which intuitively says that the rewriting of $q$ w.r.t.~$\dep$ is
obtained by computing the cartesian product of the queries
$Q_{q_1},\ldots,Q_{q_4}$, while the variables $A$ and $C$, which
occur in more than one components, have the same semantic meaning,
i.e., the joins among different components are preserved.
More precisely, the final rewriting of $q$ w.r.t.~$\dep$ is obtained
by unfolding the non-recursive Datalog query $\tup{Q_{q_1} \cup
\ldots \cup Q_{q_4} \cup \{\rho\},p}$.

The UCQ obtained by employing the above technique, and
$\mathsf{XRewrite}(q,\dep)$ have exactly the same size. In other
words, the parallelization of the rewriting procedure does not
affect the size of the final rewriting. However, it significantly
affects the execution time of the rewriting algorithm.
The execution of $\mathsf{XRewrite}$ on $q$ and $\dep$ takes 194ms,
while the execution of the parallel version of $\mathsf{XRewrite}$
takes 81ms (47ms for constructing $\tup{Q_{q_1} \cup \ldots \cup
Q_{q_4} \cup \{\rho\},p}$ and 34ms for unfolding it).
%

\subsection{The Algorithm XRewriteParallel}\label{sec:xrewrite-parallel}

Let us now formalize the idea discussed above. First, we need to
define the notion of affected positions:

\begin{definition}[Affected Positions]\label{def:affected-postions}
Consider a set $\dep$ of TGDs over a schema $\R$. An \emph{affected
position} of $\R$ w.r.t. a pair $\tup{\sigma,\dep}$, where $\sigma
\in \dep$, is defined inductively as follows:
\begin{enumerate}
\item the position $\pi_{\exists}(\sigma)$ is affected w.r.t.
$\tup{\sigma,\dep}$, and
\item a position $\pi$ in the head of a TGD $\sigma' \in
\dep$ is affected w.r.t. $\tup{\sigma,\dep}$ if the same variable
appears at $\pi$, and in the $\body{\sigma'}$ only at positions
which are affected w.r.t. $\tup{\sigma,\dep}$. \hfill\markfull
\end{enumerate}
\end{definition}

\begin{example}
Consider the set $\dep$ of TGDs consisting of
\[
\sigma_1\ :\ p(X,Y),s(Y,Z)\ \ra\ \exists W \, t(Y,X,W) \qquad
\sigma_2\ :\ t(X,Y,Z)\ \ra\ \exists W \, p(W,Z).
\]
It is easy to verify that
\[
\tup{\sigma_1,\dep}\ =\ \{t[3],p[2]\} \qquad \tup{\sigma_2,\dep}\ =\
\{p[1],t[2]\}.
\]
Notice that, although the variable $Y$ in $\body{\sigma_1}$ occurs
at position $p[2] \in \tup{\sigma_1,\dep}$, $t[1]$ is not affected
w.r.t.~$\tup{\sigma_1,\dep}$ since $Y$ also occurs at position $s[1]
\not\in \tup{\sigma_1,\dep}$. \hfill\markfull
\end{example}

By having the above auxiliary notion in place, we are now ready to
define the key notion of the existential-join decomposition of a CQ
w.r.t.~a set of TGDs.

\begin{definition}[Existential-join Decomposition]\label{def:decomposition}
Consider a CQ $q$ over a schema $\R$, and a set $\dep$ of TGDs over
$\R$. An \emph{existential-join decomposition} of $q$ w.r.t. $\dep$
is a partition $P$ of $\body{q}$ such that the following holds: if a
variable $V \in \var{q}$ occurs in $\body{q}$ only at positions
which are affected w.r.t. $\tup{\sigma,\dep}$ for some $\sigma \in
\dep$, then there exists $S \in P$ such that $V \in \var{S}$ and $V
\not\in \var{P \setminus S}$.
We say that $P$ is \emph{optimal} if there is no $S \in P$ such that
$(P \setminus S) \cup \{S_1,S_2\}$, where $\{S_1,S_2\}$ is a
partition of $S$, is an existential-join decomposition of $q$ w.r.t.
$\dep$. \hfill\markfull
\end{definition}

\begin{algorithm}[t]
\caption{The algorithm $\mathsf{XRewriteParallel}$
\label{alg:parallel}} \small
    \KwIn{a CQ $q$ over a schema $\R$ and a set $\Sigma$ of TGDs over $\R$}
    \KwOut{the perfect rewriting of $q$ w.r.t.~$\Sigma$}
    \vspace{2mm}
    \tcc{\textrm{decomposition step}}
    $\tup{\{q_1,\ldots,q_m\},\rho} := \mathsf{decompose}(q,\dep)$\;
    \tcc{\textrm{parallel step}}
    \ForPar{$q \in  \{q_1,\ldots,q_m\}$}{
        $Q_q := \mathsf{XRewrite}(q,\dep)$\;
        }
    \tcc{\textrm{merging step}}
    $\Pi := Q_{q_1} \cup \ldots \cup Q_{q_m} \cup \{\rho\}$\;
    $Q_{\textsc{fin}} := \mathsf{unfold}(\tup{\Pi,p})$\;
    \Return{$Q_{\textsc{fin}}$}
\end{algorithm}

It is easy to see that the optimal existential-join decomposition of
a CQ w.r.t.~a set of TGDs is unique.
We are now ready to describe the parallel version of
$\mathsf{XRewrite}$. As already said, the key idea hinges on the
fact that each component of an existential-join decomposition can be
rewritten independently, and the final rewriting is obtained by
merging the obtained rewritings via a reconciliation (Datalog) rule.
Consider a CQ $q$ over a schema $\R$ and a set $\dep$ of TGDs over
$\R$; for notational convenience, we assume that $p(\insX)$ is the
head-atom of $q$, and $\var{q} = \{V_1,\ldots,V_n\}$. The parallel
version of $\mathsf{XRewrite}$, called $\mathsf{XRewriteParallel}$,
which is depicted in Algorithm~\ref{alg:parallel}, is consisting of
the following three steps:

\begin{itemize}
\item [\textbf{Decomposition Step.}] The optimal existential-join decomposition
$P$ of $q$ w.r.t. $\dep$ is computed; let $P = \{C_1,\ldots,C_m\}$.
Then, for each $i \in [m]$, we construct the CQ
\[
q_i\ :\ p_i(f_i(V_1,\ldots,V_n))\ \la\ C_i,
\]
where $p_i$ is an auxiliary predicate not occurring in $\R$, and
$f_i(V_1,\ldots,V_n)$ is defined as the tuple
$\tup{V_{j_1},\ldots,V_{j_k}}$, where $1 \leqslant k \leqslant n$,
such that \emph{(i)} $1 \leqslant j_1 < \ldots < j_k \leqslant n$,
and \emph{(ii)} for each $\ell \in [k]$, $V_{j_\ell} \in \var{C_i}
\cap (\insX \cup (\var{q} \setminus \var{C_i}))$.
Intuitively, $f_i(V_1,\ldots,V_n)$ is obtained from
$\tup{V_1,\ldots,V_n}$ by keeping only the variables of $\var{C_i}$
which are also distinguished variables of $q$, or they occur in a
component other than $C_i$.
Moreover, the reconciliation (Datalog) rule
\[
\rho\ :\ p(\insX)\ \la\
p_1(f_1(V_1,\ldots,V_n)),\ldots,p_m(f_m(V_1,\ldots,V_n))
\]
is constructed. The decomposition step is carried out by the
$\mathsf{decompose}$ function, which accepts as input the query $q$
and the set of TGDs $\dep$, and returns as output the pair
$\tup{\{q_1,\ldots,q_m\},\rho}$.\\

\item[\textbf{Parallel Step.}] We construct in $m$ parallel computations the perfect
rewriting $Q_q$ of each CQ $q \in \{q_1,\ldots,q_m\}$ w.r.t.~$\dep$
by exploiting the rewriting algorithm $\mathsf{XRewrite}$.\\

\item[\textbf{Merging Step.}] It is not difficult to verify that $\tup{\Pi,p}$, where
$\Pi = (Q_{q_1} \cup \ldots \cup Q_{q_m} \cup \{\rho\})$, is a
non-recursive Datalog query. It is well-known that such a query can
be \emph{unfolded} into a (finite) UCQ; for more details see,
e.g.,~\cite{AbHV95}. The perfect rewriting of the input CQ $q$
w.r.t.~$\dep$ is the UCQ obtained by unfolding $\tup{\Pi,p}$, which
is carried out by the $\mathsf{unfold}$ function.
\end{itemize}

It is easy to see that $\mathsf{XRewriteParallel}$ terminates under
linear and sticky sets of TGDs. The decomposition step terminates
since $q$ and $\dep$ are finite, the parallel step terminates since
$\mathsf{XRewrite}$ terminates under linear and sticky sets of TGDs,
and the merging step terminates since the unfolding of a finite
non-recursive Datalog query is finite.

\begin{theorem}\label{the:tgdrewriteparallel-termination}
Consider a CQ $q$ over a schema $\R$, and a set $\dep$ of TGDs over
$\R$. If $\dep \in \mathsf{LINEAR}$ or $\dep \in \mathsf{STICKY}$,
then $\mathsf{XRewriteParallel}(q,\dep)$ terminates.
\end{theorem}

The soundness and completeness of $\mathsf{XRewriteParallel}$
follows by construction. Instead of giving a formal proof (which is
rather long and uninteresting), we intuitively explain why
$\mathsf{XRewriteParallel}$ is sound and complete.
For brevity, given a CQ $q$ and a set $\dep$ of TGDs, we denote by
$q_{\dep}^{\shortparallel}$ the rewritten query
$\mathsf{XRewriteParallel}(q,\dep)$.
It is possible to show that $q_{\dep}^{\shortparallel}$ and $q_\dep$
are the same (modulo bijective variable renaming), which immediately
implies the soundness and completeness of
$\mathsf{XRewriteParallel}$. Let $P = \{C_1,\ldots,C_m\}$ be the
optimal existential-join decomposition of $q$ w.r.t. $\dep$.
Each rewriting step applied during the execution of
$\mathsf{XRewriteParallel}(q,\dep)$ corresponds to a rewriting step
of $\mathsf{XRewrite}(q,\dep)$. This holds since, by construction of
each $q_i : p_i(f_i(V_1,\ldots,V_n)) \la C_i$, where
$V_1,\ldots,V_n$ are the variables of $\var{q}$, a variable $V \in
\var{C_i}$ which is shared in $q$ is also shared in $q_i$. More
precisely, if $V$ is a distinguished variable of $q$, or occurs in a
component of $P$ other than $C_i$, then it also occurs in
$\head{q_i}$ and thus is shared in $q_i$; otherwise, if it occurs
only in $C_i$, then is trivially shared in $q_i$ since, by
hypothesis, it occurs more than once in $C_i$.
Conversely, each rewriting step applied during the execution of
$\mathsf{XRewrite}(q,\dep)$ corresponds to a rewriting step of
$\mathsf{XRewriteParallel}(q,\dep)$. Towards a contradiction, assume
that the above claim does not hold. This implies the during the
execution of $\mathsf{XRewriteParallel}(q,\dep)$ a valid rewriting
step is not applied due to a missing factorization step. But this
implies that a variable which occurs in $\body{q}$ only at positions
which are affected w.r.t. $\tup{\sigma,\dep}$, for some $\sigma \in
\dep$, appears in more than one components of $P$ which is a
contradiction.
Notice that the reconciliation rule preserves the joins among
different components of $P$ and the claim follows:

\begin{theorem}\label{the:tgdrewriteparallel-correctness}
Consider a CQ $q$ over a schema $\R$, a database for $\R$, and a set
$\dep$ of TGDs over $\R$. It holds that,
$q_{\dep}^{\shortparallel}(D) = \ans{q}{D}{\dep}$.
\end{theorem}

\section{Optimize the Rewriting for Linear TGDs}\label{sec:ucq-optimization}

Linearity of TGDs allows us to effectively identify atoms in the
body of a query which are logically implied (w.r.t. a given set of
TGDs) by other atoms in the same query. By exploiting this fact, we
propose a technique, called \emph{query elimination}, aiming at
optimizing the obtained rewritten query under the class of linear
TGDs. As we shall see in the experimental section, query elimination
(which is an additional step during the execution of
$\mathsf{XRewrite}$) reduces \emph{(i)} the number of CQs of the
perfect rewriting, \emph{(ii)} the number of atoms in each query of
the rewriting, and \emph{(iii)} the number of joins to be executed.
Let us first give a motivating example which exposes the key idea
underlying query elimination, and also illustrates its impact on the
final rewriting.

\subsection{A Motivating Example}\label{sec:motivating-example}
%


Consider the set $\dep$ of linear TGDs and the CQ $q$ given in
Section~\ref{sec:xrewriteparalle-informal-description}.
%
%
The complete rewriting of $q$ w.r.t.~$\dep$ contains 60 conjunctive
queries executing 300 joins. However, by exploiting the set of TGDs,
it is possible to eliminate redundant atoms in the generated
queries, and thus reduce the size of the final rewriting.
For example, it is possible to eliminate from the given query $q$
the atom $\mathit{finInstrument(A)}$ since, due to the existence of
the TGDs $\sigma_2$ and $\sigma_8$ in $\dep$, if the atom
$\mathit{stockPortfolio}(B,A,D)$ is satisfied, then immediately the
atom $\mathit{finInstrument(A)}$ is also satisfied.
Notice that by eliminating a redundant atom from a query, we also
eliminate all the queries that are generated starting from it during
the rewriting process.
Moreover, due to the TGD $\sigma_3$, if the atom
$\mathit{listComponent(A,C)}$ in $q$ is satisfied, then the atom
$\mathit{finIndex(C,G,H)}$ is also satisfied, and therefore can be
eliminated.
Finally, due to the TGD $\sigma_1$, if the atom
$\mathit{stockPortfolio}(B,A,D)$ is satisfied, then the atom
$\mathit{company(B,E,F)}$ is also satisfied, and hence the latter is
redundant.
The query that has to be considered as input of the rewriting
process is therefore
\[
\begin{array}{rcl}
p(A,B,C) &\leftarrow& \mathit{stockPortfolio}(B,A,D),
\mathit{listComponent}(A,C)
\end{array}
\]
which produces a perfect rewriting containing the following two
conjunctive queries executing only two joins:
\[
\begin{array}{rcl}
p(A,B,C) &\leftarrow& \mathit{listComponent}(A,C), \mathit{stockPortfolio}(B,A,D)\\
p(A,B,C) &\leftarrow& \mathit{listComponent}(A,C),
\mathit{hasStock}(A,B).
\end{array}
\]
It is evident that by eliminating redundant atoms from a query as
described above, we reduce the number of CQs of the perfect
rewriting, the number of atoms in each query of the rewriting, and
the number of joins to be executed.

\subsection{Atom Coverage}\label{sec:atom-coverage}

Before formalizing the idea described above, let us first introduce
some auxiliary technical notions.

\begin{definition}[Propagation Graph]\label{def:propagation-graph}
Consider a set $\dep$ of TGDs over a schema $\R$. The
\emph{propagation graph} of $\dep$, denoted $\mathit{PG}(\dep)$, is
a labeled directed multigraph $\tup{N,E,\lambda}$, where $N$ is the
node set, $E$ is the edge set, and $\lambda$ is a labeling function
$E \ra \dep$. The node set is the set of positions of $\R$. If there
exists $\sigma \in \dep$ such that the same variable appears at
position $\pi_b$ in $\body{\sigma}$ and at position $\pi_h$ in
$\head{\sigma}$, then the edge $e = \tup{\pi_b,\pi_h}$ belongs to
$E$ with $\lambda(e) = \sigma$; no other edges belong to $E$.
\hfill\markfull
\end{definition}

The propagation graph of a set of linear TGDs encodes all the
possible ways of propagating a term from one position to another
position during the chase. More precisely, the existence of a path
from $\pi_1$ to $\pi_2$ implies that there may be a way to propagate
a term from $\pi_1$ to $\pi_2$ during the construction of the chase.
Given a path $P = v_1 \ldots v_n$, where $n > 1$, of
$\mathit{PG}(\dep) = \tup{N,E,\lambda}$, we say that $P$ is
\emph{minimal} if the following condition is satisfied: there is no
$1 < i < n$ and $0 < j < i$ such that $v_{i-j} \ldots v_i = v_i
\ldots v_{i+j}$ and $\lambda(\tup{v_{i-j},v_{i-j+1}}) \ldots
\lambda(\tup{v_{i-1},v_{i}}) = \lambda(\tup{v_{i},v_{i+1}}) \ldots
\lambda(\tup{v_{i+j-1},v_{i+j}})$. The minimality condition
guarantees that cycles occurring in $\mathit{PG}(\dep)$ are
traversed at most once.

\begin{figure}[t]
\begin{center}
\includegraphics[]{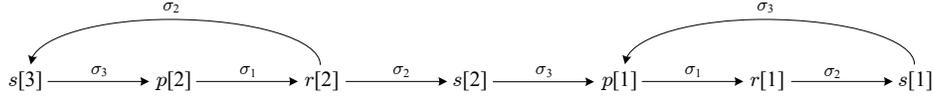}
\caption{Propagation graph for
Example~\ref{exa:propagation-graph}.}\label{fig:propagation-graph}
\end{center}
\end{figure}

\begin{example}\label{exa:propagation-graph}
Consider the set $\dep$ of linear TGDs consisting of
\[
\sigma_1 : p(X,Y) \rightarrow \exists Z \, r(X,Y,Z) \quad \sigma_2 :
r(X,Y,c) \rightarrow s(X,Y,Y) \quad \sigma_3 : s(X,X,Y) \rightarrow
p(X,Y).
\]
The propagation graph of $\dep$ (without the isolated node $r[3]$)
is depicted in Figure~\ref{fig:propagation-graph}.
The path $P = v_1 \ldots v_6$, where $v_1v_2v_3 = v_4v_5v_6 =
s[3]p[2]r[2]$ is minimal. However, the path $P' = v_1 \ldots v_9$,
where $v_1v_2v_3 = v_4v_5v_6 = v_7v_8v_9 = s[3]p[2]r[2]$ is not
minimal since the minimality condition is violated with $i = 4$ and
$j = 3$; clearly, $v_1v_2v_3v_4 = v_4v_5v_6v_7 = s[3]p[2]r[2]s[3]$
and
$\lambda(\tup{v_1,v_2})\lambda(\tup{v_2,v_3})\lambda(\tup{v_3,v_4})
= \lambda(\tup{v_4,v_5})\lambda(\tup{v_5,v_6})\lambda(\tup{v_6,v_7})
= \sigma_3 \sigma_2 \sigma_1$, which intuitively means that the
cycle $s[3]p[2]r[2]s[3]$ occurs in $P'$ twice. \hfill\markfull
\end{example}

Unfortunately, the existence of a path $P$ from $\pi_1$ to $\pi_2$
does not guarantee the propagation of a term from $\pi_1$ to
$\pi_2$.
For example, consider the TGDs $\sigma_1 : r(X,Y) \ra \exists Z \,
t(Y,Z)$ and $\sigma_2 : t(X,X) \ra s(X)$. It is easy to verify that,
although in $\mathit{PG}(\{\sigma_1,\sigma_2\})$ the path
$r[2]t[1]s[1]$ exists, there is no way to propagate a term from
$r[2]$ to $s[1]$ since the atom obtained by applying $\sigma_1$ does
not trigger $\sigma_2$.
Thus, the existence of such a path $P$ guarantees the propagation of
a term from $\pi_1$ to $\pi_2$ providing that, for each pair of
consecutive edges $e = \tup{\pi,\pi'}$ and $e' = \tup{\pi',\pi''}$
of $P$, where $e$ and $e'$ are labeled by the TGDs $\sigma$ and
$\sigma'$, respectively, the atom obtained during the chase by
applying $\sigma$ triggers $\sigma'$.
It is easy to verify that a natural sufficient condition for the
latter is as follows: for each pair of consecutive edges $e$ and
$e'$ of $P$ which are labeled by $\sigma$ and $\sigma'$,
respectively, there exists a homomorphism $h$ such that
$h(\body{\sigma'}) \subseteq \head{\sigma}$; notice that this
condition heavily relies on the linearity of the TGDs.
A sequence $\sigma_1,\ldots,\sigma_n$ of linear TGDs, where $n > 1$,
is called \emph{tight} if, for each $i \in [n-1]$, there exists a
homomorphism $h_i$ such that $h_i(\body{\sigma_{i+1}}) =
\head{\sigma_i}$; a sequence consisting of a single TGD is trivially
tight.
Furthermore, such a sequence is \emph{compatible} to an atom
$\atom{a}$ if there exists a homomorphism $h$ such that
$h(\body{\sigma_1}) = \atom{a}$.
We are now ready to introduce the central notion of \emph{atom
coverage}. For brevity, given an atom $\atom{a}$ and a term $t$,
$\mathit{pos}(\atom{a},t)$ is the set of positions at which $t$
occurs in $\atom{a}$; e.g., if $\atom{a} = r(X,Y,X)$, then
$\mathit{pos}(\atom{a},X) = \{r[1],r[3]\}$ and
$\mathit{pos}(\atom{a},Y) = \{r[2]\}$.
Moreover, given a CQ $q$ and an atom $\atom{a} \in \body{q}$, let
$T(q,\atom{a})$ be the maximal subset of $\adom{\atom{a}}$ which
contains only constants occurring in $q$ and variables which are
shared in $q$; e.g., if $q$ is the CQ $p(A) \la r(A,B,c)$, where $c
\in \dom$, then $T(q,r(A,B,c)) = \{A,c\}$.

\begin{definition}[Atom Coverage]\label{def:atom-coverage}
Consider a CQ $q$ over a schema $\R$, and a set $\dep \in
\mathsf{LINEAR}$ over $\R$. Let $\atom{a}$ and $\atom{b}$ be atoms
of $\body{q}$.
We say that $\atom{a}$ \emph{covers} $\atom{b}$ w.r.t.~$q$ and
$\dep$, written as $\atom{a} \prec_{\dep}^{q} \atom{b}$, if the
following conditions are satisfied:
\begin{enumerate}
\item $T(q,\atom{b}) \subseteq \mathit{terms}(\atom{a})$,
and

\item there exists a sequence $S = \sigma_1,\ldots,\sigma_m$ of TGDs
of $\dep$, for $m \geqslant 1$, such that:
\begin{enumerate}
\item $S$ is tight and compatible to $\atom{a}$;

\item for each $t \in T(q,\atom{b})$ and $\pi \in
\mathit{pos}(\atom{b},t)$, there exists a minimal path
$\pi_1\pi_2\ldots\pi_{m+1}$ in $\mathit{PG}(\dep)$ such that $\pi_1
\in \mathit{pos}(\atom{a},t)$, $\pi_{m+1} = \pi$ and
$\lambda(\tup{\pi_{j},\pi_{j+1}}) = \sigma_j$, for each $j \in [m]$.
\end{enumerate}
\end{enumerate}
The \emph{cover set} of an atom $\atom{a} \in \body{q}$ w.r.t.~$q$
and $\dep$, denoted $\mathit{cover}(\atom{a},q,\dep)$, is the set
$\{\atom{b}~|~\atom{b} \in \body{q} \setminus \{\atom{a}\}
\textrm{~and~} \atom{b} \prec_{\dep}^{q} \atom{a}\}$; when $q$ and
$\dep$ are obvious from the context, we shall denote the above set
as $\mathit{cover}(\atom{a})$. \hfill\markfull
\end{definition}

Intuitively speaking, the first condition of atom coverage ensures
that by removing $\atom{b}$ from $q$ we do not loose any constant,
and also all the joins between $\atom{b}$ and the other atoms of
$\body{q}$, except $\atom{a}$, are preserved. The second condition
guarantees that $\atom{b}$ is logically implied (w.r.t. $\dep$) by
$\atom{a}$, and thus can be safely eliminated.
The choice of considering only minimal paths in condition 2(b) is
crucial in order to be able to explicitly construct the cover set of
an atom without considering infinite paths. Notice that by
considering infinite paths we compute exactly the same cover sets.
More precisely, if $\atom{a} \curlyeqprec_{\dep}^{q} \atom{b}$
denotes the fact that $\atom{a}$ covers $\atom{b}$ w.r.t. $q$ and
$\dep$ if we consider infinite paths in
Definition~\ref{def:atom-coverage}, then it is easy to verify that
$\atom{a} \curlyeqprec_{\dep}^{q} \atom{b}$ implies $\atom{a}
\prec_{\dep}^{q} \atom{b}$. In fact, if $\atom{a}
\curlyeqprec_{\dep}^{q} \atom{b}$ because of a non-minimal path $P$,
then we can construct a minimal path $P'$ from $P$, by eliminating
the repeated cycles, which is a witness for the fact that $\atom{a}
\prec_{\dep}^{q} \atom{b}$.

\begin{lemma}\label{lem:atom-coverage}
Consider a CQ $q$ over a schema $\R$, and a set $\dep \in
\mathsf{LINEAR}$ over $\R$. Suppose that $\atom{a} \prec_{\dep}^{q}
\atom{b}$, where $\{\atom{a},\atom{b}\} \subseteq \body{q}$, and
$q'$ is obtained from $q$ by eliminating $\atom{b}$. Then, $q'(I)
\subseteq q(I)$, for each instance $I$ that satisfies $\dep$.
\end{lemma}

\begin{proof}
Fix a tuple of constants $\tuple{t}$. Suppose there exists a
homomorphism $h$ such that $h(\body{q'}) \subseteq I$ and $h(\insV)
= \tuple{t}$, where $\insV$ are the distinguished variables of $q'$.
We need to show that there exists a homomorphism $\hat{h}$ such that
$\hat{h}(\body{q}) \in I$ and $\hat{h}(\insV) = \tuple{t}$. Let us
first give an auxiliary technical claim; its proof can be found in
Section~\ref{appsec:atom-coverage-claim}.

\begin{claim}\label{cla:atom-coverage}
There exists a linear TGD $\sigma$ over $\R$ such that $\dep \models
\sigma$, a substitution $\lambda$, and a substitution $\mu$ which is
the identity on $\var{\body{q'}}$, such that $\lambda(\body{\sigma})
= \atom{a}$ and $\lambda(\head{\sigma}) = \mu(\atom{b})$.
\end{claim}

Since $\atom{a} \in \body{q'}$, Claim~\ref{cla:atom-coverage}
implies that $h(\lambda(\body{\sigma})) \in I$. Recall that $\dep
\models \sigma$, and thus $I \models \sigma$. This implies that
there exists $h' \supseteq h|_{\insX}$, where $\insX$ are the
variables that appear both in $\lambda(\body{\sigma})$ and
$\lambda(\head{\sigma})$, such that $h'(\lambda(\head{\sigma})) \in
I$. Therefore, $h'(\lambda(\head{\sigma})) = h'(\mu(\atom{b}))$.
Since $\mu$ is the identity on $\var{\body{q'}}$, $h$ and $h' \circ
\mu$ are compatible. Consequently, the substitution $\rho = h \cup
(h' \circ \mu)$ maps $\body{q}$ to $I$, and $\rho(\insV) = h(\insV)=
\tuple{t}$. The claim follows with $\hat{h} = \rho$.
\end{proof}

The above technical result provides the logical underpinning for the
query elimination technique. More precisely,
Lemma~\ref{lem:atom-coverage} suggests that, for each CQ $q$
obtained by applying the rewriting step of $\mathsf{XRewrite}$, the
atoms of $\body{q}$ that are logically implied (w.r.t.~$\dep$) by
some other atom of $\body{q}$ can be eliminated, and the obtained
subquery is equivalent to $q$ w.r.t. query answering.

\begin{example}\label{exa:query-elimination}
Consider the set $\dep$ constituted by the linear TGDs
\[
\begin{array}{rcl}
\sigma_1\ :\ t(X,Y) \ra \exists Z \, r(X,Y,Z), &\qquad& \sigma_3\ :\ s(X,Y,Z) \ra t(Z,X),\\
\sigma_2\ :\ r(X,Y,Z) \ra \exists W \, s(Y,W,X), &\qquad& \sigma_4\
:\ t(X,Y) \ra s(X,Y,Y).
\end{array}
\]
Let also $q$ be the CQ
\[
p(A)\ \la\
\underbrace{t(A,B)}_{\atom{a}},\underbrace{r(A,B,C)}_{\atom{b}},\underbrace{s(A,B,B)}_{\atom{c}}.
\]
By Definition~\ref{def:atom-coverage}, $\mathit{cover}(\atom{a}) =
\{\atom{b}\}$, $\mathit{cover}(\atom{b}) = \{\atom{a}\}$ and
$\mathit{cover}(\atom{c}) = \{\atom{a},\atom{b}\}$.
Thus, we can either eliminate $\atom{a},\atom{c}$ and get the CQ
$p(A) \la r(A,B,C)$, or eliminate $\atom{b},\atom{c}$ and get the CQ
$p(A) \la t(A,B)$. Both queries are equivalent to $q$ (for query
answering purposes). \hfill\markfull
\end{example}

\begin{algorithm}[t]
\caption{The algorithm $\mathsf{eliminate}$ \label{alg:eliminate}}
\small
    \KwIn{a CQ $q$, an elimination strategy $S$ for $q$, and a set $\dep$ of linear TGDs}
    \KwOut{the set of eliminable atoms from $\body{q}$ w.r.t.~$S$ and $\dep$}
    \vspace{2mm}
    $A := \emptyset$\;
    \ForEach{$i := 1$ {\rm to} $n$}{
        $\atom{a} := S[i]$\;
        \If{$\mathit{cover}(\atom{a},q,\dep) \neq \emptyset$}{
            $A := A \cup \{\atom{a}\}$\;
            \ForEach{$\atom{b} \in \body{q} \setminus A$}{
                $\mathit{cover}(\atom{b},q,\dep) := \mathit{cover}(\atom{b},q,\dep)
                \setminus \{\atom{a}\}$\;
            }
        }
    }
    \Return{$A$}
\end{algorithm}

\subsection{Unique Elimination
Strategy}\label{sec:unique-elimination-strategy}

The outcome of query elimination is not unique, as it heavily
depends on the order that we consider the atoms of the query under
consideration. In the above example, the order
$\atom{a},\atom{b},\atom{c}$ gives the subquery $p(A) \la r(A,B,C)$,
while the order $\atom{b},\atom{a},\atom{c}$ gives the subquery
$p(A) \la t(A,B)$. Before presenting the optimized version of
$\mathsf{XRewrite}$, let us first discuss which elimination strategy
best suits our needs.

An \emph{(atom) elimination strategy} for a CQ $q$ is a permutation
of its body-atoms. By exploiting the cover set of the atoms of
$\body{q}$, we associate to each elimination strategy $S$ for $q$ a
subset of $\body{q}$, denoted $\mathsf{eliminate}(q,S,\dep)$, which
is the set of atoms of $\body{q}$ that can be safely eliminated
(according to $S$) in order to obtain a logically equivalent query
(w.r.t.~$\dep$) with less atoms in its body. Formally,
$\mathsf{eliminate}(q,S,\dep)$ is computed by applying
Algorithm~\ref{alg:eliminate}; given an elimination strategy $S$,
$S[i]$ is the $i$-th element of $S$.
As already observed, given two strategies $S_1$ and $S_2$, in
general, $\mathsf{eliminate}(q,S_1,\dep) \neq
\mathsf{eliminate}(q,S_2,\dep)$. The question that comes up concerns
the choice of the elimination strategy. Since our goal is to
eliminate as many atoms as possible, we should choose an elimination
strategy which maximizes the number of eliminable atoms. However,
the process of finding such a strategy is computationally expensive;
in particular, given a query with $n$ body-atoms, we have to
enumerate the $n!$ different elimination strategies, and for each
one of them, compute the set of eliminable atoms.
Interestingly, such an expensive computation can be avoided since,
regardless of the chosen elimination strategy, always we eliminate
the same number of atoms, i.e., the strategy of eliminating atoms
from the body of a query is unique (modulo the number of the
eliminable atoms).
The proof of this result, that can be found in
Section~\ref{appsec:unique-elimination-lemma}, relies on the fact
that the binary relation $\prec_{\dep}^{q}$ is transitive.

%
%


\begin{lemma}\label{lem:unique-elimination-strategy}
Consider a CQ $q$, and a set $\dep \in \mathsf{LINEAR}$. Let $S_1$
and $S_2$ be arbitrary elimination strategies for $q$. It holds
that, $|\mathsf{eliminate}(q,S_1,\dep)| =
|\mathsf{eliminate}(q,S_2,\dep)|$.
\end{lemma}

Henceforth, given a CQ $q$ of the form $\atom{h} \la
\atom{a}_1,\ldots,\atom{a}_n$, we refer to \emph{the} atom
elimination strategy for $q$ denoted by $S_q$, and we denote by
${\lfloor q \rfloor}_{\dep}$ the CQ obtained from $q$ after
eliminating from $\body{q}$ the atoms of
$\mathsf{eliminate}(q,S_q,\dep)$.

\subsection{Query Elimination}\label{sec:query-elimination}

We are now ready to describe the optimized algorithm
$\mathsf{XRewriteEliminate}$. During the execution of
$\mathsf{XRewrite}$, after the rewriting and factorization steps,
the query elimination step is applied. $\mathsf{XRewriteEliminate}$
is obtained after modifying $\mathsf{XRewrite}$ as follows:
\begin{enumerate}
\item line 2 --- $Q_{\textsc{rew}} := \{{\langle \lfloor q \rfloor}_{\dep},\mathsf{r},\mathsf{u} \rangle\}$;
\item line 10 --- $q' := {\lfloor \gamma_{S,\sigma^i}(q[S/\body{\sigma^i}])
\rfloor}_{\dep}$; and
\item line 17 --- $q' := {\lfloor \gamma_{S}(q) \rfloor}_{\dep}$.
\end{enumerate}
Since $\mathsf{eliminate}$ terminates, and
$\mathsf{XRewriteEliminate}$ generates less queries than
$\mathsf{XRewrite}$, the termination of the optimized algorithm
follows by Theorem~\ref{the:tgdrewrite-termination}:

\begin{theorem}\label{the:tgdrewrite-star-termination}
Consider a CQ $q$ over a schema $\R$ and a set $\dep \in
\mathsf{LINEAR}$ over $\R$. Then,
$\mathsf{XRewriteEliminate}(q,\dep)$ terminates.
\end{theorem}

The next result establishes the correctness of
$\mathsf{XRewriteEliminate}$. For brevity, given a CQ $q$ and a set
$\dep$ of linear TGDs, the query
$\mathsf{XRewriteEliminate}(q,\dep)$ is denoted $q_{\dep}^{\star}$.

\begin{theorem}\label{the:TGD-rewrite-star-sound-complete}
Consider a CQ $q$ over a schema $\R$, a database $D$ for $\R$, and a
set $\dep \in \mathsf{LINEAR}$ over $\R$. It holds that,
$q_{\dep}^{\star}(D) = \ans{q}{D}{\dep}$.
\end{theorem}

\begin{proof}
Since $D \subseteq \chase{D}{\dep}$, by monotonicity of CQs,
$q_{\dep}^{\star}(D) \subseteq q_{\dep}^{\star}(\chase{D}{\dep})$;
thus, $q_{\dep}^{\star}(D) \subseteq
\ans{q_{\dep}^{\star}}{D}{\dep}$. By giving a proof similar to that
of Lemma~\ref{lem:sound-auxiliary-lemma}, and also by exploiting
Lemma~\ref{lem:atom-coverage}, we can show that
$\ans{q_{\dep}^{\star}}{D}{\dep} \subseteq \ans{\hat{q}}{D}{\dep}$,
where $\hat{q} = {\lfloor q \rfloor}_{\dep}$. Since $\chase{D}{\dep}
\models \dep$, Lemma~\ref{lem:atom-coverage} implies that
$\hat{q}(\chase{D}{\dep}) \subseteq q(\chase{D}{\dep})$; hence,
$\ans{\hat{q}}{D}{\dep} \subseteq \ans{q}{D}{\dep}$ which implies
$q_{\dep}^{\star}(D) \subseteq \ans{q}{D}{\dep}$.
Conversely, $\body{\hat{q}} \subset \body{q}$ implies
$\ans{q}{D}{\dep} \subseteq \ans{\hat{q}}{D}{\dep}$. Since, by
construction, $\hat{q} \in q_{\dep}^{\star}$, we immediately get
that $\ans{\hat{q}}{D}{\dep} \subseteq
\ans{q_{\dep}^{\star}}{D}{\dep}$. By devising a proof similar to
that of Lemma~\ref{lem:complete-auxiliary-lemma}, and also by
exploiting Lemma~\ref{lem:atom-coverage}, we can show that
$\ans{q_{\dep}^{\star}}{D}{\dep} \subseteq q_{\dep}^{\star}(D)$.
Therefore, $\ans{q}{D}{\dep} \subseteq q_{\dep}^{\star}(D)$, and the
claim follows.
\end{proof}

It is important to clarify that the above result does not hold if we
consider arbitrary TGDs. This is because
Lemma~\ref{lem:atom-coverage}, which is crucial in the proof of
Theorem~\ref{the:TGD-rewrite-star-sound-complete}, is heavily based
on the linearity of TGDs.
Notice that the algorithm $\mathsf{XRewriteEliminateParallel}$ can
be naturally defined by considering in the parallel step of
$\mathsf{XRewriteParallel}$ the algorithm
$\mathsf{XRewriteEliminate}$ instead of $\mathsf{XRewrite}$.

\subsection{The Chase \& Backchase Approach}

The task of finding all the minimal equivalent reformulations of a
CQ w.r.t.~a set of TGDs has been already investigated in databases.
The most interesting approach in this respect is the chase \&
backchase ($\mathsf{C\&B}$) algorithm~\cite{DePT99}.
During the \emph{chase phase}, the given CQ $q$ is chased using the
TGDs of the given set $\dep$, yielding a query $q_U$ called
\emph{universal plan}.
The \emph{backchase phase} enumerates all minimal subqueries of
$q_U$ which are equivalent to $q$ w.r.t.~$\dep$; henceforth, we
refer to $\dep$-minimal and $\dep$-equivalent subqueries. For a
subquery $q_S$ of $q_U$, to decide whether $q_S$ is
$\dep$-equivalent to $q$ it suffices to check whether $q_S
\subseteq_{\dep} q$, i.e., $q_S$ is contained in $q$ w.r.t~$\dep$,
which reduces to finding a containment mapping from $q$ to the query
obtained after chasing $q_S$ using $\dep$. Let us recall that,
instead of naively enumerating all the possible subqueries of $q_U$
during the backchase phase, one can employ a bottom-up approach,
starting with all subqueries with just one atom, continuing with
those consisting of two atoms, and so on, and stop as soon as a
subquery which is $\dep$-equivalent to $q$ is found. This is
possible due to the so-called \emph{pruning property}, which says
that, if a subquery $q_S$ of $q_U$ is $\dep$-equivalent to $q$, then
every subquery of $q_U$ which is a superquery of $q_S$ cannot be
both $\dep$-equivalent to $q$ and $\dep$-minimal.

It is obvious that $\mathsf{C\&B}$ is more general than our query
elimination technique. More precisely, given a CQ $q$ and a set
$\dep$ of linear TGDs, $\mathsf{C\&B}$ will definitely return the CQ
$\mathsf{eliminate}(q,S_q,\dep)$. Therefore, during the execution of
$\mathsf{XRewrite}$, the elimination of redundant atoms can be done
by exploiting the $\mathsf{C\&B}$ algorithm instead of relying on
our query elimination technique. Unfortunately, $\mathsf{C\&B}$
suffers from two major drawbacks which make it inappropriate for our
purposes.
The first one is the fact that it works only for classes of TGDs
which guarantee the termination of the chase. Recall that in both
phases of the algorithm we need to chase a query as long as no new
atoms can be obtained. Thus, if we consider, e.g., arbitrary linear
TGDs, then the termination of the procedure is not guaranteed.
The second one (assuming that we focus on a class which guarantees
the termination of the chase) is the fact that we need to apply the
chase procedure double-exponentially many times (in general), which
makes the whole procedure computationally expensive --- recall that
the main motivation underlying our backward-chaining algorithm was
precisely the avoidance of the explicit construction of the chase.
Therefore, although the $\mathsf{C\&B}$ algorithm can be used to
identify and eliminate redundant query atoms, the query elimination
approach is more appropriate for our purposes since it works for
arbitrary linear TGDs, and can effectively identify redundant atoms
without an explicit construction of the chase.
Conceptually speaking, our query elimination technique is a refined
version of the $\mathsf{C\&B}$ algorithm, specifically engineered
for the class of linear TGDs.

\section{Implementation}\label{sec:implementation}

We implemented \tgdrewrite and its optimizations in Java by
extending the \iris Datalog engine~\cite{BiFi08}. Throughout this
section we will refer to this implementation as \irispm. All data
used in our evaluation, together with the complete source code of
\irispm are publicly available\footnote{Omitted due to double-blind
review. }.

\subsection{System Architecture}\label{sec:system-architecture}

A high-level overview of the main architectural components of
\irispm and their interconnections is shown in
Figure~\ref{fig:system-architecture}(a).
The input of the system consists of a pair $\tup{Q,\dep}$; $Q$ is a
set of CQs to be executed against a (possibly incomplete) relational
database $D$, and $\dep$ is an ontology constituted by
non-conflicting TGDs and FDs, and negative constraints (NCs). The
\irispm \emph{parser} partitions $\dep$ into $\dep_T$ (the set
of TGDs), $\dep_F$ (the set of FDs), and $\dep_\bot$ (the set of
NCs). The \emph{constraints manager} accepts $\dep_T$, and
constructs (query-independent) support data structures based on
$\dep_T$. In particular, the cover graph of $\dep_T$, which is
basically the transitive closure of the propagation graph of
$\dep_T$ (see Definition~\ref{def:propagation-graph}) is constructed
--- more details are given in the following subsection.
The constraints manager accepts also $\dep_F$ and $\dep_\bot$, and
constructs a set $Q_F$ and $Q_\bot$, respectively, of \emph{check
queries}, which are actually unions of CQs, that will be used to
verify whether $D$ satisfies $\dep_F$ and $D \cup \dep$ satisfies
$\dep_\bot$. The \emph{query manager} takes as input the set $Q$,
and schedules the CQs of $Q$ for rewriting and execution.

Both input and check queries are handed over to the rewriting
engine. More precisely, given as input a CQ $q \in Q$, the union of
CQs $Q_\bot$, and the set of TGDs $\dep_T$ (along with the cover
graph of $\dep_T$), the rewriting engine rewrites $q$ and $Q_\bot$
using \tgdrewrite into a union of CQs $Q_q$ and
$\overline{Q}_\bot$, respectively.
Then, the \emph{SQL-Rewriter} accepts as input $Q_q$, $Q_F$ and
$\overline{Q}_\bot$, and rewrites them into equivalent
\emph{select-project-join} SQL queries $SQ_q$, $SQ_F$ and $SQ_\bot$,
respectively, to be executed against $D$.
A non-empty answer to the (rewritten) check query $SQ_F$ (resp.,
$SQ_\bot$) implies that a FD of $\dep_F$ (resp., a NC of
$\dep_\bot$) is violated, i.e., $D \cup \dep$ is inconsistent. In
this case, \irispm exits with an error and a list of violated
constraints together with the tuples of $D$ that ``witness'' the
violation. If for the check queries the answer is the empty set,
then \irispm executes the rewritten query $SQ_q$ over $D$.

\begin{figure}[t]
    \centering
        \includegraphics[width=1.0\columnwidth]{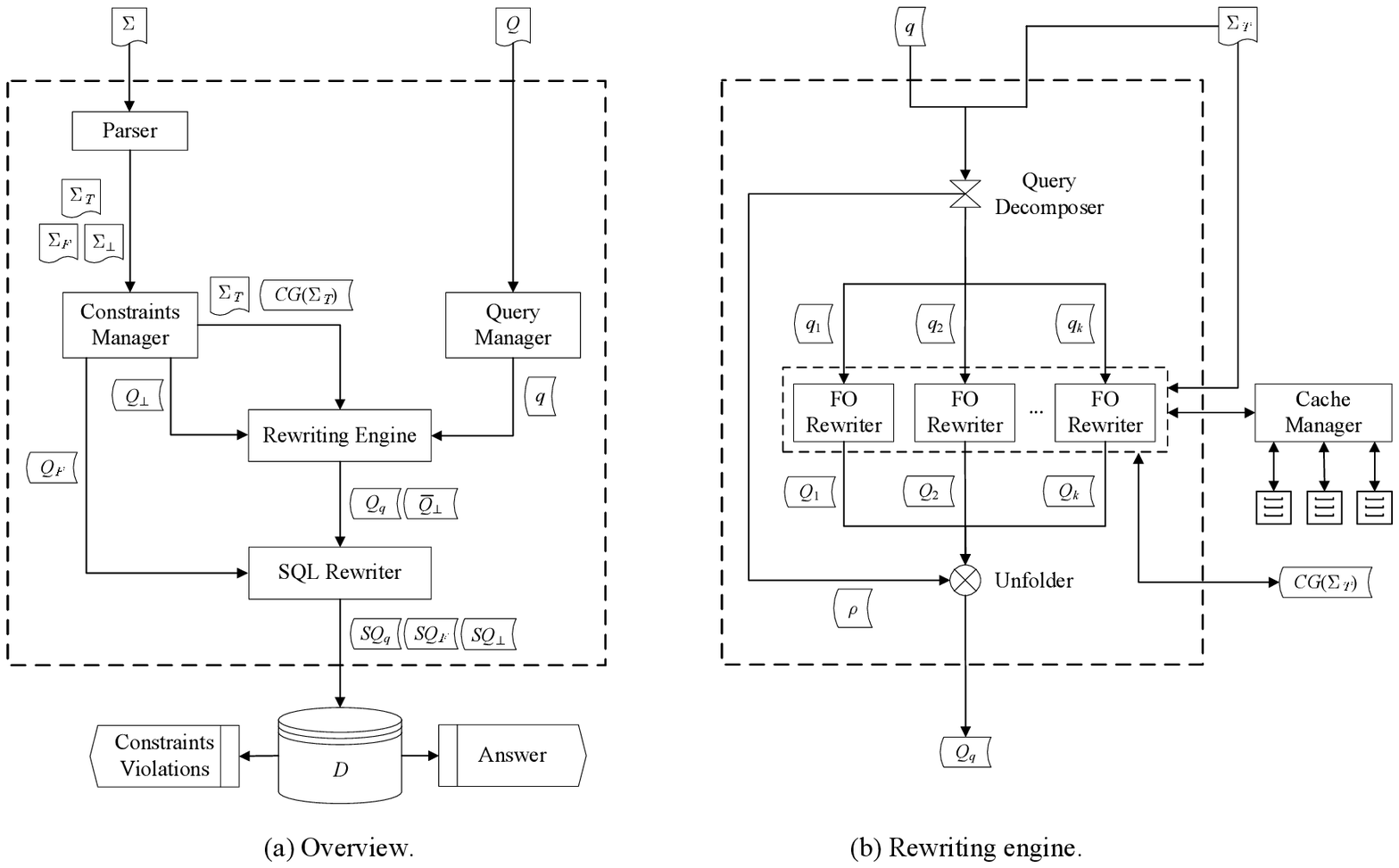}
\caption{\irispm architecture components.}
\label{fig:system-architecture}
\end{figure}

Figure~\ref{fig:system-architecture}(b) shows in more detail the
architectural structure of the \irispm rewriting engine. The
main module is the \emph{FO-Rewriter} which implements
\tgdrewrite. The engine receives as input a CQ $q$ and the
set $\dep_T$ along with the cover graph of $\dep_T$. First, hands
$q$ and $\dep_T$ over to the \emph{query decomposer} which decomposes
$q$ into components $q_{1}, q_{2}, \ldots, q_{k}$, according to the
procedure described in Section~\ref{sec:parallelize}, that can be
rewritten independently. The decomposer also computes the
reconciliation rule $\rho$. Each $q_{i}$, where $i \in [k]$, is then
handed over an independent FO-Rewriter that produces the rewriting
$Q_{i}$ for that particular component. Each atom of the
reconciliation rule $\rho$ is then unfolded using the corresponding
rewriting $Q_{i}$.
All FO-Rewriters share access to the graph $\mathit{CG}(\dep_T)$.
Moreover, during the execution of the rewriting procedure, an
additional data structure, called \emph{query graph}, is maintained,
which actually stores the CQs generated during the rewriting ---
more details are given in the following subsection.
Furthermore, the FO-Rewriters share access to caching facilities
which aim at avoiding to recompute several times the same piece of
information, e.g., the MGU for a set of atoms, which is needed for
the execution of the rewriting process --- this is discussed in more
details in Section~\ref{sec:caching-mechanism}. Notice that the
\emph{cache manager} ensures synchronized access to the caches.

It is worth noting that an indexing structure for TGDs is adopted.
More precisely, \emph{TGD-Index} is implemented as a map $M(K,V)$,
where a key $k \in K$ is a predicate symbol of the underlying
schema, and the value $M(k) \in V$ is a set of TGDs of $\dep$ having
$k$ as head-predicate. This allows a FO-Rewriter, during an
applicability check, to consider only those TGDs which may be
applicable. This is quite beneficial since, despite the fact that a
single applicability check is computationally easy, the rewriting
step iterates over each TGD $\sigma \in \dep$ checking applicability
of $\sigma$ to a set of atoms $S$ in the query $q$ being rewritten,
and therefore on large ontologies this iteration might result in an
unnecessary waste of time since only a few TGDs may be applicable to
$S$.

\subsection{Support Data Structures}\label{sec:support-data-structures}

As already said, \irispm makes use of support data structures,
namely, the query graph and the cover graph. In what follows, we
describe how these data structures are implemented, as well as how
they are used during the rewriting process.

\paragraph{\small{\textsf{Query Graph.}}}
The \emph{query graph} stores the queries generated during the
rewriting process. The formal definition follows:

\begin{definition}[Query Graph]\label{def:query-graph}
Consider a CQ $q$ over a schema $\R$, and set $\dep$ of TGDs over
$\R$. The \emph{query graph} of $q$ and $\dep$ is a labeled directed
acyclic graph $\tup{N,E,\lambda}$, where $N$ is the node set, $E$ is
the edge set, and $\lambda$ is a labeling function $N \ra$
$\mathsf{XRewrite}(q,\dep)$.
An edge $\tup{v,u}$ occurs in $E$ if there exists $\sigma \in \dep$
and $S \subseteq \body{q_v}$, where $\sigma$ is applicable to $S$
and $q_v = \lambda(v)$, and an integer $i \geqslant 1$, such that
$\lambda(u) = \gamma_{S,\sigma^i}(q_v[S/\body{\sigma^i}])$.
\hfill\markfull
\end{definition}

In other words, the above definition says that $\lambda(u)$ is
obtained during the execution of $\mathsf{XRewrite}(q,\dep)$ by
applying the rewriting step on $\lambda(v)$.
Interestingly, apart from storing the generated queries, the query
graph keeps also track of the provenance of the queries. This allows
us, whenever a generated query is recognized as redundant because of
a query subsumption check (that we are going to discuss in the next
section), to use the edges in the graph to eventually remove its
descendants, thus saving similar checks which are redundant. The
query graph is implemented using
JGraphT\footnote{\url{http://jgrapht.org}} that provides efficient
data structures for the representation of graph-like structures and
comes with efficient implementations of algorithms such as
reachability.

\paragraph{\small{\textsf{Cover Graph.}}}
As said, the cover graph of a set $\dep$ of TGDs is actually the
transitive closure of the propagation graph of $\dep$.
We denote by $\dep^{\star}$ the Kleene closure of $\dep$\footnote{By
abuse of notation, we consider $\dep$ as a set of symbols.}, i.e.,
the set of all strings over $\dep$ of any length $n > 0$. By abuse
of notation, if $\mathit{PG}(\dep) = \tup{N,E,\lambda}$ and $v_1
\ldots v_n$ is a path in $\mathit{PG}(\dep)$, then by $\lambda(v_1
\ldots v_n)$ we denote the string
$\lambda(\tup{v_1,v_2})\lambda(\tup{v_2,v_3}) \ldots
\lambda(\tup{v_{n-1},v_n}) \in \dep^{\star}$. The formal definition
follows:

\begin{definition}[Cover Graph]\label{def:cover-graph}
Consider a set $\dep$ of TGDs over a schema $\R$, and assume that
$\mathit{PG}(\dep) = \tup{N,E,\lambda}$. The \emph{cover graph} of
$\dep$, denoted $\mathit{CG}(\dep)$, is a labeled directed
multigraph $\tup{N,E',\lambda'}$, where $E' \supseteq E$ and
$\lambda' : E' \ra \dep^{\star}$. The edge set $E'$ is defined as
follows: \emph{(i)} if there exists a minimal path $v_1 \ldots v_n$,
where $n > 1$, in $\mathit{PG}(\dep)$ such that the sequence of TGDs
$\lambda(v_1 \ldots v_n)$ is tight, then in $E'$ there exists an
edge $e = \tup{v_1,v_n}$ with $\lambda'(e) = \lambda(v_1 \ldots
v_n)$, and \emph{(ii)} no other edges belong to $E'$.
\hfill\markfull
\end{definition}



The cover graph is used to check whether a certain position is
reachable from some other position of the underlying schema. More
precisely, it is used to check for the existence of a tight sequence
of TGDs as required by the definition of atom coverage (see
Definition~\ref{def:atom-coverage}). Moreover, it is used  for the
computation of the affected positions of the underlying schema (see
Definition~\ref{def:affected-postions}).
Notice that in the cases where query elimination is not applied,
e.g., when the input set of TGDs is not linear, and thus we do not
need to check for atom coverage, then we can consider only the
propagation graph (and not the cover graph).

The cover graph is implemented as a map $M(K,V)$, where a key $k \in
K$ is a pair $\tup{\pi,\pi'}$ of positions of the underlying schema
such that $\pi^{\prime}$ is reachable from $\pi$ via a sequence of
TGDs, and the value $M(k) \in V$ is the set of all sequences $s$ of
TGDs such that $\pi^{\prime}$ is reachable from $\pi$ via $s$.
This implementation of the cover graph proved to be a better
alternative than a traditional graph structure due to the
potentially high number of calls to the reachability procedure ---
by pre-computing the closure of the propagation graph, reachability
can be checked in constant time.


\section{Experimental Evaluation}\label{sec:experimental-evaluation}
We are now ready to perform an experimental evaluation of
$\mathsf{XRewrite}$. After describing the experimental setting, we
carried out an extensive internal evaluation in order to better
understand the impact of the proposed optimization techniques.
Finally, we compare our system with ALASKA, the reference
implementation of~\cite{KLMT12}, which is the only known system
supporting ontological query answering under general existential
rules.



\subsection{Experimental Setting}\label{sec:experimental-setting}
%

%
Since ontological query answering under existential rules is a
relatively recent area of research, no benchmark is currently
available. We therefore resorted to an established benchmark for
DL-based query rewriting systems used, e.g.,
in~\cite{PeMH10,KLMT12}.
The benchmark consists of five ontologies expressed in the
well-known description logic DL-Lite$_{\R}$. Notice that every set
of DL-Lite$_{\R}$ axioms can be translated into an equivalent set
(w.r.t.~query answering) of linear TGDs and NCs over a schema
consisting of unary and binary predicates; for details
see~\cite{CaGL12}. A brief description of the ontologies follows:

\begin{itemize}
\item[--] VICODI ($\mathsf{V}$) is an ontology of European history, developed within the VICODI project\footnote{http://www.vicodi.org.}. It
consists of 222 linear TGDs without constraints.
\item[--] STOCKEXCHANGE ($\mathsf{S}$) is an ontology of the domain
of financial institutions within the EU. It consists of 53
linear TGDs without constraints.
\item[--] UNIVERSITY ($\mathsf{U}$) is a DL-Lite$_\R$ version of
the LUBM
Benchmark\footnote{http://swat.cse.lehigh.edu/projects/lubm/.},
developed at Lehigh University, and describes the organizational
structure of universities. It consists of 87 linear TGDs without
constraints.
\item[--] ADOLENA ($\mathsf{A}$) (Abilities and Disabilities OntoLogy for ENhancing
Accessibility) developed for the South African National
Accessibility Portal, and describes abilities, disabilities and
devices. It consists of 154 linear TGDs and 19 NCs.
\item[--] The Path5 ($\mathsf{P5}$) ontology is a synthetic ontology encoding graph
structures, and used to generate an exponential-blowup of the size
of the rewritten queries. It consists of 13 linear TGDs without
constraints.
\end{itemize}



Since $\mathsf{XRewrite}$ supports general existential rules, we have complemented the above benchmark with two ontologies consisting of linear and sticky sets of TGDs, respectively, which
are not expressible using description logics.

\begin{itemize}
\item[--] Split-Full ($\mathsf{SF}$) is an ontology designed to test the ability of
a rewriting algorithm to exploit query decomposition. It consists of
60 linear TGDs over a schema with predicates of arity at most three.
\item[--] Clique ($\mathsf{CLQ}$) is an ontology representing $k$-cliques in
a graph, where $k \in [3]$, and has been devised to test the ability
of rewriting engines to handle sticky sets of TGDs. It consists of
34 TGDs over a schema with predicates of arity at most four.
\end{itemize}

Each ontology has an associated set of test queries (see
Section~\ref{appsec:test-queries}) either obtained via an analysis
of query logs or manually created.
%
Since $\mathsf{XRewrite}$ is provably sound and complete, it is necessary to
give some metrics for the quality of the rewriting:

\begin{itemize}
\item[--] \emph{Size.} When the target rewriting language is UCQs, the size represents the
number of CQs in the final rewriting. Some existing approaches and
systems, e.g.,~\cite{OrPi11,PeMH10,RoAl10} also support other
languages for the rewriting such as non-recursive or bounded
Datalog. In this case the size of the rewriting is the number of
rules in the Datalog program. Notice that the fact that Datalog
rewritings are syntactically more succinct than UCQs does not
immediately imply that they are preferable from a practical point of
view. One of the reasons is the necessity to resort to Datalog
engines or some form of pre-processing before being able to execute
a Datalog rewriting against standard relational database systems.
Other size-related metrics include the \emph{number of joins} and
the \emph{number of atoms} since they are an indication of the
effort necessary to execute the rewriting in practice. Since all
disjuncts in the rewriting must be executed, in the following we
always consider the total number of atoms and joins in the
rewriting.
\item[--] \emph{Rewriting time.} Assuming that the natural setting of
ontological query answering is a transactional environment, another
important metric is the time required to compute a final (and
executable) rewriting once a query is submitted to the system. In
this paper, we do not include in this metric the time required for
the construction of the cover graph, which does not depend on the
query itself and can be constructed beforehand. However, we include
query-dependent pre- and post-processing steps such as query
decomposition.


%
\item[--] \emph{Memory consumption.} Represents the peak memory usage reached
during the rewriting of a given query. This metric always includes
the memory consumption introduced by caches but not the memory
consumption of auxiliary data structures such as the cover graph.
\item[--] \emph{Search space.} Another typical metric for query rewriting
algorithms is the number of CQs explored and generated during the
rewriting~\cite{KLMT13}. In case of $\mathsf{XRewrite}$, the
explored queries are those labelled with $\mathsf{e}$, while the
generated ones are those obtained via a rewriting step (possibly
multiple times). Ideally, a rewriting algorithm should be able to
explore and generate only the necessary queries for the final
rewriting; as we shall see, this is not always the case.
\end{itemize}

The machine used for testing is a Dell Optiplex 9020 with 4
dual-core Intel i7-4770 processors at 3.40GHz (8 cores in total),
running Linux Mint v15 (Olivia) x86-64, Kernel 3.8.0-19. The machine
is equipped with 32Gb of RAM. We used a Java VM 1.7.0-45 provided
with 16Gb of maximum heap size.

\renewcommand{\tabcolsep}{4pt}
\begin{table}[t]
    \centering
    \tbl{The impact of caching (per-query averages).\label{tab:caching statistics}}{
        \begin{tabular}{|l|c|c|c|c|c|c|c|}
            \cline{2-8}
            \multicolumn{1}{l|}{~} & ~ & ~ & ~ & ~ & ~ & ~ & ~\\
            \multicolumn{1}{l|}{~} & $\mathsf{V}$ & $\mathsf{S}$ & $\mathsf{U}$ & $\mathsf{A}$ & $\mathsf{P5}$ & $\mathsf{SF}$ & $\mathsf{CLQ}$\\
            \multicolumn{1}{l|}{~} & ~ & ~ & ~ & ~ & ~ & ~ & ~\\
            \hline
            \textbf{Factorization}    & 0 & 1 & 0 & 37k & 9k & 0 & 0 \\
            Same invocation (\%)  & 0 & 0 & 0 & 0   & 0  & 0 & 0 \\
            Distinct input       & 0 & 1 & 0 & 37k & 9k & 0 & 0 \\
            \hline
            \textbf{Atom coverage}      & 12 & 18 & 12   & 642k & 202k & 10 & N/A \\
            Same invocation (\%)  & 0  & 12 & 11.2 & 86   & 63.4 & 0  & N/A \\
            Distinct input       & 12 & 2  & 9    & 3k   & 2k   & 10 & N/A \\
            \hline
            \textbf{Homomorphism check} & 447 & 12 & 8 & 863k & 145 & 351 & 1.4k \\
            Same invocation (\%)  & 0   & 0  & 0 & 0    & 0   & 0   & 0 \\
            Distinct input       & 447 & 12 & 8 & 863k & 145 & 351 & 1.4k \\
            \hline
            \textbf{MGU computation}              & 457  & 97   & 25   & 294k & 93k & 500k & 24k \\
            Same invocation (\%)  & 72.3 & 50.4 & 34.8 & 98.2 & 79.6 & 77.6 & 98.3 \\
            Distinct input       & 73   & 27   & 13   & 384  & 6k & 92 & 345 \\
            \hline
            \textbf{Canonical renaming}         & 5k   & 249  & 558  & 4.6M & 63k  & 17k & 6k \\
            Same invocation (\%)  & 90.6 & 66.8 & 64.6 & 97.2 & 78.8 & 63.4 & 75.6 \\
            Distinct input       & 372  & 69   & 106  & 76k  & 15k  & 8k &  2k \\
            \hline
            \hline
        \end{tabular}
    }
\end{table}

\subsection{Caching Mechanism}\label{sec:caching-mechanism}

During the rewriting process, several operations, such as the
computation of the MGU for a set of atoms, are likely to be applied
multiple times for the same input. This might occur either within a
single FO-Rewriter, e.g., because the same CQ is generated more than
once in different branches of the rewriting procedure, or due to
multiple FO-Rewriters exploring the same CQ in two different
branches of the search space.
For this reason, we have analyzed the behavior of
$\mathsf{XRewrite}$ to identify operations that might benefit from
caching. In order to determine these operations, and also to
dimension the caches, we set up an experiment recording, for each
target operation, the total number of invocations, the number of
invocations on the same input, and the number of distinct inputs
that these operations have been invoked on. The number of
invocations on the same inputs corresponds to the maximum number of
cache \emph{hits} we can achieve, while the number of invocations on
distinct inputs corresponds to the \emph{size} of the cache that is
necessary to obtain the maximum number of cache hits.
The target operations considered in this experiment are the
following:
\begin{itemize}
\item[--] \emph{Factorization:} given a CQ $q$, a TGD $\sigma$, and a set $S \subseteq \body{q}$
which is factorizable w.r.t.~$\sigma$, compute the query
$\gamma_S(q)$, i.e., the query obtained by factorizing $q$.

\item[--] \emph{Query elimination:} given a CQ $q$, and a set $\dep$ of TGDs,
compute the query ${\lfloor q \rfloor}_{\dep}$, that is, the query
obtained by applying the query elimination step on $q$.

\item[--] \emph{Homomorphism check:} given two sets of atoms $S_1$
and $S_2$, check whether there exists a homomorphism from $S_1$ to
$S_2$.

\item[--] \emph{MGU computation}: given a set of atoms $S$, compute
the MGU for $S$.

\item[--] \emph{Canonical renaming:} given a CQ $q$, compute the query
$\mathit{can_q}(q)$; for the definition of the canonical renaming
$\mathit{can_q}$ see the last paragraph of
Section~\ref{sec:termination-of-tgd-rewrite}.
\end{itemize}

%

Table~\ref{tab:caching statistics} summarizes the results of our
experiment. The values are reported as per-query averages on all
test ontologies. A first observation is that, despite the fact that
factorization and homomorphism check are very frequent, they are
mostly invoked on different inputs. As a consequence, caching the
output of these two operations would be rather ineffective, and thus
representing an unnecessary burden on the rewriting engine. This is
not necessarily a negative result, since it shows that \irispm
explores the rewriting search-space effectively without a caching
mechanism in place.
On the other hand, MGU computation and canonical renaming are often
invoked on the same input with a hit rate for an MGU cache ranging
between 50.4\% to 98.2\%, and a hit rate for a canonical renaming
cache ranging from 63.4\% to 97.2\%. Contrasting results have been
observed for the query elimination cache, with a hit rate ranging
from 0 (i.e., totally ineffective) on $\mathsf{V}$ and $\mathsf{SF}$
to 86\% on $\mathsf{A}$. The reason of this difference has to be
found in the query decomposition. In fact, on $\mathsf{A}$ the input
queries cannot be effectively decomposed thus making the query
elimination more likely to find queries that can be reduced via atom
coverage. On the contrary, for those ontologies where the queries
are highly decomposable, e.g., $\mathsf{V}$ and $\mathsf{SF}$, the
atom coverage is rarely applied on queries which are already small.
The fact that a query elimination cache can potentially be useful in
some of the ontologies, led us to keep it in our implementation.

We now come to the problem of dimensioning the various caches and
deciding suitable caching algorithms. Since the ontologies for which
caching is likely to be more effective are $\mathsf{A}$ and
$\mathsf{P5}$ due to the fact that their queries are poorly
decomposable, we dimensioned the caches at roughly the 75\% of the
optimal size, i.e., the MGU cache has been designed for 4.5k
entries, the canonical renaming cache for 55k entries, and the query
elimination cache for 2k entries.

Caches are implemented as maps $M(K,V)$, where the nature of keys
and values varies depending on the particular cache. The
$\mathsf{MGUCache}$ caches MGU computations; a key is a set of atoms
and the value is their MGU. The $\mathsf{EliminationCache}$ caches
the result of the application of query elimination on a query; a key
is a CQ $q$ and the value is the CQ ${\lfloor q \rfloor}_{\dep}$,
where $\dep$ is the input set of TGDs. Finally, the
$\mathsf{RenameCache}$ stores canonical renamings of queries; the
key is a CQ $q$ and the value is the CQ $\mathit{can}_q(q)$.

\subsection{Internal Evaluation}\label{sec:internal-evaluation}
The aim of the internal evaluation is to quantify the impact of our
optimizations on the rewriting. In particular, they aim at
\begin{inparaenum}[\itshape(i)]
\item reducing the number of redundant queries in the final rewriting while preserving its completeness, and
\item intelligently explore the rewriting search space, e.g., by avoiding the exploration of redundant queries.
\end{inparaenum}

\renewcommand{\tabcolsep}{0.5pt}
\begin{table}[t]
    \centering
    \tbl{The impact of query elimination on the rewriting. \label{tab:atom-coverage}}{
        \begin{tabular}{cc|c|c|c|c|c|c|c|c|c|c|c|c|c|c|}
            \cline{3-16}
            ~ & ~ & \multicolumn{2}{c|}{Size} & \multicolumn{2}{c|}{\#Atoms} & \multicolumn{2}{c|}{\#Joins} & \multicolumn{2}{c|}{Explored} & \multicolumn{2}{c|}{Generated} & \multicolumn{2}{c|}{Time (ms)} & \multicolumn{2}{c|}{Memory (MB)}\\
            \cline{3-16}
            ~&~&~&~&~&~&~&~&~&~&~&~&~&~&~&~\\
            ~ & ~ & \textsc{Base} & \textsc{QE} & \textsc{Base} & \textsc{QE} & \textsc{Base}& \textsc{QE} & \textsc{Base}& \textsc{QE} & \textsc{Base}& \textsc{QE} & \textsc{Base}& \textsc{QE} & \textsc{Base}& \textsc{QE}\\
            ~&~&~&~&~&~&~&~&~&~&~&~&~&~&~&~\\
            \hline
            \multicolumn{1}{|c|}{\multirow{5}{*}{$\mathsf{V}$}} & $q_1$ & 15 & 15 & 15 & 15 & 0 & 0 & 15 & 15 & 14 & 14 & 9 & 9 & 4.3 & 4.3 \\
            \multicolumn{1}{|c|}{}                      & $q_2$ & 10 & 10 & 30 & 30 & 30 & 30 & 10 & 10 & 9 & 9 & 7 & 7 & 4.3 & 6.3\\
            \multicolumn{1}{|c|}{}                      & $q_3$ & 72 & 72 & 216 & 216 & 144 & 144 & 72 & 72 & 71 & 71 & 44 & 45 & 4.7 & 6.7 \\
            \multicolumn{1}{|c|}{}                      & $q_4$ & 185 & 185 & 555 &555 & 370 & 370 & 185 & 185 & 184 & 184 & 111 & 115 & 5.4 & 7.4 \\
            \multicolumn{1}{|c|}{}                      & $q_5$  & 30 & 30 & 210 & 210 & 270 & 270 & 30 & 30 & 29 & 29 & 26 & 28 & 4.6 & 6.6\\
            \hline
            \multicolumn{1}{|c|}{\multirow{5}{*}{$\mathsf{S}$}} & $q_1$ & 6 & 6 & 6 & 6 & 0 & 0 & 6 & 6 & 7 & 7 & 2 & 2 & 4.1 & 4.1\\
            \multicolumn{1}{|c|}{}                      & $q_2$ & 160 & 2 & 480 & 2 & 320 & 0 & 160 & 2 & 244 & 1 & 43 & 2 & 5.9 & 8.2 \\
            \multicolumn{1}{|c|}{}                      & $q_3$ & 504 & 4 & 2,520 & 8 & 2,520 & 4 & 504 & 4 & 823 & 3 & 198 & 8 & 16.8 & 8.2 \\
            \multicolumn{1}{|c|}{}                      & $q_4$  & 960 & 4 & 4,800 & 8 & 4,800 & 4 & 960 & 4 & 1,445 & 3 & 363 & 2 & 25.3 & 8.2 \\
            \multicolumn{1}{|c|}{}                      & $q_5$ & 3,024 & 8 & 21,168 & 24 & 27,216 & 24 & 3,024 & 8 & 4,892 & 7 & 1.7s & 3 & 12.6 & 8.3 \\
            \hline
            \multicolumn{1}{|c|}{\multirow{5}{*}{$\mathsf{U}$}}    & $q_1$  & 2 & 2 & 4 & 4 & 2 & 2 & 5 & 5 & 4 & 4 & 3 & 3 & 4.1 & 6.2 \\
            \multicolumn{1}{|c|}{}                      & $q_2$  & 148 & 1 & 444 & 1 & 296 & 0 & 240 & 1 & 250 & 0 & 73 & 1 & 5.8 & 4.1 \\
            \multicolumn{1}{|c|}{}                      & $q_3$ & 224 & 4 & 1,344 & 16 & 2,016 & 20 & 1,008 & 12 & 1,007 & 11 & 432 & 7 & 18.5 & 8.3 \\
            \multicolumn{1}{|c|}{}                      & $q_4$ & 1,628 & 2 & 4,884 & 2 & 1,628 & 0 & 5,000 & 5 & 6,094 & 4 & 1.6s & 3 & 54.1 & 8.2 \\
            \multicolumn{1}{|c|}{}                      & $q_5$ & 3,009 & 10 & 12,036 & 20 & 18,054 & 20 & 8,154 & 25 & 11,970 & 24 & 3.2s & 8 & 119.2 & 8.4 \\
            \hline
            \multicolumn{1}{|c|}{\multirow{5}{*}{$\mathsf{A}$}}    & $q_1$ & 402 & 299 & 779 & 573 & 377 & 274 & 782 & 679 & 847 & 725 & 818 & 729 & 7.9 & 17.0 \\
            \multicolumn{1}{|c|}{}                      & $q_2$ & 103 & 94 & 256 & 238 & 153 & 144 & 1,784 & 1,772 & 1,783 & 1,783 & 1.1s & 1.2s & 19.1 & 33.4 \\
            \multicolumn{1}{|c|}{}                      & $q_3$ & 104 & 104 & 520 & 520 & 520 & 520 & 4,752 & 4,752 & 4,751 & 4,751 & 3.2s & 3.5s & 62.7 & 97.5 \\
            \multicolumn{1}{|c|}{}                      & $q_4$ & 492 & 456 & 1,288 & 1,216 & 796 & 760 & 7,110 & 6,740 & 7,110 & 6,838 & 3.8s & 3.5s & 67.8 & 65.8 \\
            \multicolumn{1}{|c|}{}                      & $q_5$ & 624 & 624 & 3,120 & 3,120 & 3,120 & 3,120 & 76,122 & 69,448 & 76,121 & 70,457 & 52.3s & 49.8s & 1.1G & 981.5 \\
            \hline
            \multicolumn{1}{|c|}{\multirow{5}{*}{$\mathsf{P5}$}}   & $q_1$ & 6 & 6 & 6 & 6 & 0 & 0 & 14 & 14 & 13 & 13 & 1 & 2 & 4.1 & 4.1 \\
            \multicolumn{1}{|c|}{}                      & $q_2$ & 10 & 10 & 16 & 16 & 6 & 6 & 77 & 77 & 76 & 80 & 55 & 8 & 4.5 & 8.6 \\
            \multicolumn{1}{|c|}{}                      & $q_3$ & 13 & 13 & 29 & 29 & 16 & 16 & 410 & 400 & 409 & 413 & 57 & 52 & 7.4 & 11.4 \\
            \multicolumn{1}{|c|}{}                      & $q_4$ & 15 & 15 & 44 & 44 & 29 & 29 & 2,275 & 2,210 & 2,274 & 2,273 & 368 & 403 & 30.3 & 33.5 \\
            \multicolumn{1}{|c|}{}                      & $q_5$  & 16 & 16 & 60 & 60 & 44 & 44 & 13,522 & 13,085 & 13,521 & 13,424 & 3.2s & 3.2s & 211.7 & 208.3 \\
            \hline
            \multicolumn{1}{|c|}{\multirow{5}{*}{$\mathsf{SF}$}} & $q_1$ & 1 & 1 & 3 & 3 & 2 & 2 & 1 & 1 & 0 & 0 & 1 & 1 & 0.053 & 2.1  \\
            \multicolumn{1}{|c|}{}                        & $q_2$ & 125 & 125 & 375 & 375 & 250 & 250 & 125 & 125 & 124 & 124 & 30 & 33 & 5.1 & 7.1 \\
            \multicolumn{1}{|c|}{}                        & $q_3$ & 1,000 & 1,000 & 3,000 & 3,000 & 2,000 & 2,000 & 1,000 & 1,000 & 999 & 999 & 227 & 237 & 12.6 & 14.7 \\
            \multicolumn{1}{|c|}{}                        & $q_4$ & 8,000 & 8,000 & 24,000 & 24,000 & 16,000 & 16,000 & 8,000 & 8,000 & 7,999 & 7,999 & 2s & 2.2s & 77.2 & 77.0 \\
            \multicolumn{1}{|c|}{}                        & $q_5$ & 27,000 & 27,000 & 162,000 & 162,000 & 108,000 & 108,000 & 27,000 & 27,000 & 26,999 & 26,999 & 12.4s & 12.4s & 560.7 & 561.6  \\
            \hline
            \hline
        \end{tabular}}
\end{table}

\renewcommand{\tabcolsep}{1pt}
\begin{table}[t]
    \centering
    \tbl{The impact of parallelization on the rewriting. \label{tab:decomposition}}{
        \begin{tabular}{cc|c|c|c|c|c|c|c|c|c|c|c|c|c|c|c|c|c|c|}
            \cline{4-16}
            \multicolumn{1}{c}{~} & \multicolumn{1}{c}{~} & \multicolumn{1}{c|}{~} & \multicolumn{2}{c|}{Size} & \multicolumn{2}{c|}{Explored} & \multicolumn{2}{c|}{Generated} & \multicolumn{5}{c|}{Time (ms)} &  \multicolumn{2}{c|}{{Memory (MB)}}\\
            \cline{3-16}
            ~&~&~&~&~&~&~&~&~&~&~&~&~&~&~&~\\
            ~ & ~ & \multicolumn{1}{c|}{Comp} & \textsc{Base} & \textsc{Para} & \textsc{Base} & \textsc{Para} & \textsc{Base} & \textsc{Para} & \textsc{Base} & \textsc{Para} & Rew & Split & Unfold & \textsc{Base} & \textsc{Para} \\
            ~&~&~&~&~&~&~&~&~&~&~&~&~&~&~&~\\
            \hline
            \multicolumn{1}{|c|}{\multirow{5}{*}{$\mathsf{V}$}} & $q_1$ & 1 & 15 & 15 & 15 & 15 & 14 & 14 & 9 & 14 & 14 & 0 & 0 & 4.3 & 4.3 \\
            \multicolumn{1}{|c|}{}                   & $q_2$ & 3 & 10 & 10 & 10 & 12 & 9 & 9 & 7 & 4 & 3 & 1 & 0 & 6.3 & 6.4 \\
            \multicolumn{1}{|c|}{}                   & $q_3$ & 3 & 72 & 72 & 72 & 28 & 71 & 25 & 45 & 25 & 24 & 1 & 2 & 6.7 & 6.7 \\
            \multicolumn{1}{|c|}{}                   & $q_4$ & 3 & 185 & 185 & 185 & 43 & 184 & 40 & 115 & 26 & 26 & 0 & 3 & 7.5 & 7.4 \\
            \multicolumn{1}{|c|}{}                   & $q_5$ & 7 & 30 & 30 & 30 & 14 & 29 & 7 & 28 & 16 & 16 & 0 & 2 & 6.6 & 6.7 \\
            \hline
            \multicolumn{1}{|c|}{\multirow{5}{*}{$\mathsf{S}$}} & $q_1$ & 1 & 6 & 6 & 6 & 6 & 7 & 7 & 2 & 2 & 2 & 0 & 0 & 4.2 & 4.2 \\
            \multicolumn{1}{|c|}{}                   & $q_2$ & 1 & 2 & 2 & 2 & 2 & 1 & 1 & 2 & 2 & 1 & 0 & 0 & 8.2 & 8.2 \\
            \multicolumn{1}{|c|}{}                   & $q_3$ & 1 & 4  & 4 & 4 & 4 & 3 & 3 & 8 & 3 & 2 & 0 & 0 & 8.2 & 8.2 \\
            \multicolumn{1}{|c|}{}                   & $q_4$ & 2 & 4 & 4 & 4 & 4 & 3 & 2 & 2 & 3 & 2 & 0 & 0 & 8.2 & 8.2 \\
            \multicolumn{1}{|c|}{}                   & $q_5$ & 2 & 8 & 8 & 8 & 6 & 7 & 4 & 3 & 4 & 3 & 1 & 0 & 8.3 & 8.3 \\
            \hline
            \multicolumn{1}{|c|}{\multirow{5}{*}{$\mathsf{U}$}} & $q_1$ & 2 & 2 & 2 & 5 & 6 & 4 & 4 & 3 & 4 & 3 & 0 & 0 & 6.2 & 6.2 \\
            \multicolumn{1}{|c|}{}                   & $q_2$ & 1 & 1 & 1 & 1 & 1 & 0 & 0 & 1 & 1 & 1 & 0 & 0 & 4.1 & 4.1 \\
            \multicolumn{1}{|c|}{}                   & $q_3$ & 4 & 4 & 4 & 12 & 9 & 11 & 5 & 7 & 4 & 3 & 1 & 1 & 8.3 & 8.3 \\
            \multicolumn{1}{|c|}{}                   & $q_4$ & 1 & 2 & 2 & 5 & 5 & 4 & 4 & 3 & 3 & 2 & 0 & 0 & 8.2 & 8.2 \\
            \multicolumn{1}{|c|}{}                   & $q_5$ & 2 & 10 & 10 & 25 & 10 & 24 & 8 & 8 & 5 & 5 & 0 & 0 & 8.3 & 8.3 \\
            \hline
            \multicolumn{1}{|c|}{\multirow{5}{*}{$\mathsf{A}$}} & $q_1$ & 1 & 299 & 299 & 679 & 679 & 725 & 725 & 729 & 282 & 281 & 1 & 0 & 17.0 & 17.0 \\
            \multicolumn{1}{|c|}{}                   & $q_2$ & 1 & 94 & 94 & 1,772 & 1,772 & 1,783 & 1,783 & 1.2s & 853 & 852 & 1 & 0 & 33.4 & 33.4 \\
            \multicolumn{1}{|c|}{}                   & $q_3$ & 3 & 104 & 104 & 4,752 & 4,754 & 4,751 & 4,751 & 3.5s & 2.5s & 2.5s & 3 & 7 & 97.5 & 49.8 \\
            \multicolumn{1}{|c|}{}                   & $q_4$ & 1 & 456 & 456 & 6,740 & 6,740 & 6,838 & 6,838 & 3.5s & 3.5s & 3.5s & 1 & 0 & 65.9 & 93.9 \\
            \multicolumn{1}{|c|}{}                   & $q_5$ & 2 & 624 & 624 & 69,448 & 69,449 & 70,457 & 70,486 & 49.8s & 43.4s & 43.4s & 5 & 18 & 981.5 & 865.0 \\
            \hline
            \multicolumn{1}{|c|}{\multirow{5}{*}{$\mathsf{P5}$}}& $q_1$ & 1 & 6 & 6 & 14 & 14 & 13 & 13 & 2 & 2 & 1 & 0 & 0 & 4.1 & 4.1 \\
            \multicolumn{1}{|c|}{}                   & $q_2$ & 1 & 10 & 10 & 77 & 77 & 80 & 80 & 8 & 9 & 8 & 0 & 0 & 8.6 & 8.6 \\
            \multicolumn{1}{|c|}{}                   & $q_3$ & 1 & 13 & 13 & 400 & 400 & 413 & 413 & 52 & 61 & 61 & 0 & 0 & 11.4 & 11.4 \\
            \multicolumn{1}{|c|}{}                   & $q_4$ & 1 & 15 & 15 & 2,210 & 2,210 & 2,273 & 2,273 & 403 & 400 & 399 & 1 & 0 & 33.5 & 33.5 \\
            \multicolumn{1}{|c|}{}                   & $q_5$ & 1 & 16 & 16 & 13,085 & 13,085 & 13,424 & 13,424 & 3.2s & 3.1s & 3.1s & 0 & 0 & 208.2 & 208.6 \\
            \hline
            \multicolumn{1}{|c|}{\multirow{5}{*}{$\mathsf{SF}$}}& $q_1$ & 3 & 1 & 1 & 1 & 3 & 0 & 0 & 1 & 3 & 2 & 1 & 1 & 2.1 & 6.2 \\
            \multicolumn{1}{|c|}{}                   & $q_2$ & 3 & 125 & 125 & 125 & 15 & 124 & 12 & 33 & 6 & 5 & 0 & 1 & 7.1 & 6.6 \\
            \multicolumn{1}{|c|}{}                   & $q_3$ & 3 & 1,000 & 1,000 & 1,000 & 30 & 999 & 27 & 237 & 15 & 14 & 1 & 10 & 14.7 & 9.1 \\
            \multicolumn{1}{|c|}{}                   & $q_4$& 3 & 8,000 & 8,000 & 8,000 & 60 & 7,999 & 57 & 2.2s & 82 & 82 & 0 & 73 & 770.4 & 274.4 \\
            \multicolumn{1}{|c|}{}                   & $q_5$ & 6 & 27,000 & 27,000 & 27,000 & 39 & 26,999 & 33 & 12.4s & 472 & 471 & 1 & 464 & 561.6 & 121.4 \\
            \hline
            \hline
            \multicolumn{1}{|c|}{\multirow{5}{*}{$\mathsf{CLQ}$}}& $q_1$ & 1 & 38 & 38 & 38 & 38 & 57 & 57 & 102 & 8 & 8 & 0 & 0 & 4.6 & 4.6 \\
            \multicolumn{1}{|c|}{}                   & $q_2$ & 2 & 38 & 38 & 38 & 39 & 54 & 56 & 140 & 15 & 14 & 1 & 1 & 4.6 & 4.7 \\
            \multicolumn{1}{|c|}{}                   & $q_3$ & 4 & 152 & 152 & 152 & 44 & 223 & 59 & 864 & 17 & 17 & 0 & 6 & 7.5 & 5.5 \\
            \multicolumn{1}{|c|}{}                   & $q_4$ & 5 & 5,776 & 5,776 & 5,776 & 82 & 9,871 & 112 & 48.3s & 317 & 316 & 1 & 304 & 287.4 & 87.08 \\
            \hline
            \hline
        \end{tabular}
    }
\end{table}

\paragraph{\small{\textsf{Query Elimination.}}}
The first optimization we consider is query elimination (introduced
in Section~\ref{sec:ucq-optimization}). Query elimination requires
linearity of the TGDs, therefore we exclude the $\mathsf{CLQ}$
ontology from the analysis. Table~\ref{tab:atom-coverage} quantifies
the gain produced by query elimination (\textsc{QE}) against a
baseline (\textsc{Base}), where $\mathsf{XRewrite}$ is run without
applying any additional optimization steps (see
Section~\ref{sec:algorithm-tgd-rewrite}).

Query elimination provides a substantial advantage in terms of the
size of the rewriting for the ontologies $\mathsf{U}$ and
$\mathsf{S}$. In particular, for $q_2$ in $\mathsf{U}$ and
$\mathsf{S}$, all but one atoms are eliminated from the input
queries, thus resulting in a 98\% reduction in the size of the
rewriting.
On the other side, query elimination is ineffective on $\mathsf{V}$
and $\mathsf{P5}$. For the ontology $\mathsf{V}$, the test queries,
as well as all the queries generated during the rewriting process,
are already ``minimal'' in the sense that no atoms are eliminated
after applying query elimination.
As a natural consequence, query elimination has also a beneficial
effect on the exploration of the rewriting search space since entire
branches of the exploration space are pruned. This also impacts the
running time and the memory consumption. Again, a substantial
improvement is observed on $\mathsf{S}$ and $\mathsf{U}$ both in
terms of explored and generated queries. For ontologies
$\mathsf{P5}$ and $\mathsf{A}$ we observe a gain in the exploration
and generation of queries, although this does not translates to a
substantially smaller size of the final rewriting. It is worth
noting that, even when query elimination is less effective (i.e.,
$\mathsf{A}$, $\mathsf{P5}$ and $\mathsf{V}$), the impact of the
additional checks on the rewriting time and memory consumption is
negligible.

\paragraph{\small{\textsf{Parallelize the Rewriting.}}}
%
We now discuss how the decomposition-based parallelization of the
rewriting procedure (Section~\ref{sec:parallelize}) impacts the
rewriting metrics. Differently from query elimination,
parallelization is applicable regardless of the expressive power of
the input ontology. Table~\ref{tab:decomposition} summarizes the
results, where \textsc{Para} denotes $\mathsf{XRewrite}$ with
parallelization (and query elimination). The comparison is carried
out against a baseline ({\textsc{Base}}), where only query
elimination is applied. Since the parallelization cannot reduce the
final size of the rewriting, we report the size of the rewriting
only to complement the results of Table~\ref{tab:atom-coverage} with
the size of the rewriting for $\mathsf{CLQ}$, where query
elimination is not applied.
The number of components (Comp), computed for each
query and for each ontology, is also reported. As before, we also
give the number of explored and generated CQs. Along with the
overall rewriting time, we also report the time to rewrite all
components (Rew), the time necessary to decompose the query under
consideration (Split), and to unfold the rewritten components
(Unfold). As usual, we also report the impact of the optimization on
memory consumption.

An immediate conclusion is that, when the input query is
decomposable, the rewriting search space can often be explored more
efficiently. For certain ontologies, such as $\mathsf{SF}$ and
$\mathsf{CLQ}$, the gain is substantial and is also generally
reflected into a lower rewriting time and memory consumption. For
other ontologies, such as $\mathsf{V}$, even if the input query is
fully decomposable into atomic components, e.g., $q_5$, the
decomposition could result in a loss of performance due to the
overhead introduced by multi-threaded execution of FO-Rewriters. On
the other hand, it is worth noting that this occurs for queries that
can be already rewritten very quickly even without applying query
elimination. The results on the ontology $\mathsf{A}$ deserve
further explanation. As it can be seen, for both $q_3$ and $q_5$,
the number of explored queries increases. The reason is that for
$q_3$ (resp., $q_5$) two (resp., one) of the computed components do
not get rewritten, and therefore they count as two (resp., one)
additional explored queries, but no substantial gain is obtained
from such a decomposition. This is not the case without
decomposition since they would have all be part of a unique query,
counting as a single explored query. In addition, parallelization
can potentially prevent applicability of query elimination if the
covered and covering atoms reside in two different components.
Another interesting observation is that the decomposition is more
effective when the rewriting search space can be partitioned into
fairly similar subsets that can be explored by rewriting
independently each component. This is the case, e.g., for $q_5$ on
$\mathsf{SF}$ but not for $q_3$ and $q_5$ on $\mathsf{A}$, where
some components do not generate any rewriting.
If we consider those tests where decomposition is more effective,
e.g., $\mathsf{SF}$ and $\mathsf{CLQ}$, we observe that most of the
time is spent unfolding the rewritten components into a UCQ. A
possible way of tackling this problem is to keep the rewriting
``folded'', i.e., as a non-recursive Datalog rewriting; more details
can be found in Section~\ref{appsec:nr-datalog-rewriting}.

\renewcommand{\tabcolsep}{0.4pt}
\begin{table}[t]
    \centering
    \tbl{The impact of subsumption check on the rewriting. \label{tab:subsumption-check}}{
        \begin{tabular}{cc|c|c|c|c|c|c|c|c|c|c|c|c|c|c|c|c|}
            \cline{3-18}
             ~ & ~ & \multicolumn{4}{c|}{{Size}} & \multicolumn{2}{c|}{Explored} & \multicolumn{2}{c|}{Generated} & \multicolumn{4}{c|}{Time (ms)} &  \multicolumn{4}{c|}{Memory (MB)}\\
            \cline{3-18}
            ~&~&~&~&~&~&~&~&~&~&~&~&~&~&~&~&~&~\\
            ~ & ~ & \textsc{Base} & \textsc{Tail} & \textsc{IDec} & \textsc{IRew} & \textsc{Base} & \textsc{IRew} & \textsc{Base} & \textsc{IRew} & \textsc{Base} & \textsc{Tail} & \textsc{IDec} & \textsc{IRew} & \textsc{Base} & \textsc{Tail} & \textsc{IDec} & \textsc{IRew}\\
            ~&~&~&~&~&~&~&~&~&~&~&~&~&~&~&~&~&~\\
            \hline
            \multicolumn{1}{|c|}{\multirow{5}{*}{$\mathsf{V}$}} & $q_1$ & 15 & 15 & 15 & 15 & 15 & 15 & 14 & 14 & 14 & 12 & 13 & 10 & 4.3 & 4.3 & 4.3 & 4.3 \\
            \multicolumn{1}{|c|}{}                   & $q_2$ & 10 & 10 & 10 & 10 & 12 & 12 & 9 & 9 & 4 & 9 & 7 & 26 & 6.4 & 6.4 & 6.4 & 6.4 \\
            \multicolumn{1}{|c|}{}                   & $q_3$ & 72 & 72 & 72 & 72 & 28 & 28 & 25 & 25 & 25 & 28 & 15 & 24 & 6.6 & 6.6 & 6.6 & 6.6 \\
            \multicolumn{1}{|c|}{}                   & $q_4$ & 185 & 185 & 185 & 185 & 43 & 43 & 40 & 40 & 26 & 75 & 28 & 55 & 7.4 & 7.4 & 7.4 & 7.4 \\
            \multicolumn{1}{|c|}{}                   & $q_5$ & 30 & 30 & 30 & 30 & 14 & 14 & 7 & 7 & 16 & 12 & 11 & 14 & 6.7 & 6.7 & 6.7 & 6.7 \\
            \hline
            \multicolumn{1}{|c|}{\multirow{5}{*}{$\mathsf{S}$}} & $q_1$ & 6 & 6 & 6 & 6 & 6 & 6 & 7 & 7 & 2 & 2 & 3 & 3 & 4.2 & 4.2 & 4.2 & 4.2 \\
            \multicolumn{1}{|c|}{}                   & $q_2$ & 2 & 2 & 2 & 2 & 2 & 2 & 1 & 1 & 2 & 2 & 3 & 2 & 8.2 & 8.2 & 8.2 & 8.2 \\
            \multicolumn{1}{|c|}{}                   & $q_3$ & 4 & 4 & 4 & 4 & 4 & 4 & 3 & 3 & 3 & 3 & 3 & 2 & 8.3 & 8.3 & 8.3 & 8.3 \\
            \multicolumn{1}{|c|}{}                   & $q_4$ & 4 & 4 & 4 & 4  & 4 & 4 & 2 & 2 & 3 & 3 & 4 & 5 & 8.3 & 8.3 & 8.3 & 8.3 \\
            \multicolumn{1}{|c|}{}                   & $q_5$ & 8 & 8 & 8 & 8 & 6 & 6 & 4 & 4 & 4 & 3 & 6 & 5 & 8.3 &8.3 & 8.3 & 8.3 \\
            \hline
            \multicolumn{1}{|c|}{\multirow{5}{*}{$\mathsf{U}$}} & $q_1$ & 2 & 2 & 2 & 2 & 6 & 6 & 4 & 4 & 4 & 4 & 6 & 5 & 6.2 & 6.2 & 6.2 & 6.2 \\
            \multicolumn{1}{|c|}{}                   & $q_2$ & 1 & 1 & 1 & 1 & 1 & 1 & 0 & 0 & 1 & 1 & 1 & 2 & 4.2 & 4.2 & 4.2 & 4.2 \\
            \multicolumn{1}{|c|}{}                   & $q_3$ & 4 & 4 & 4 & 4 & 9 & 9 & 5 & 5 & 4 & 4 & 3 & 5 & 8.3 & 8.3 &8.3 & 8.3 \\
            \multicolumn{1}{|c|}{}                   & $q_4$ & 2 & 2 & 2 & 2 & 5 & 5 & 4 & 4 & 3 & 2 & 4 & 3 & 8.3 & 8.3 & 8.3 & 8.3 \\
            \multicolumn{1}{|c|}{}                   & $q_5$ & 10 & 10 & 10 & 10 & 10 & 10 & 8 & 8 & 5 & 3 & 6 & 5 & 8.3 & 8.3 & 8.3 & 8.3 \\
            \hline
            \multicolumn{1}{|c|}{\multirow{5}{*}{$\mathsf{A}$}} & $q_1$ & 299 & 27 & 27 & 27 & 679 & 41 & 725 & 45 & 282 & 771 & 325 & 168 & 17.0 & 12.0 & 17.0 & 8.7 \\
            \multicolumn{1}{|c|}{}                   & $q_2$ & 94 & 50 & 50 & 50 & 1,772 & 1,431 & 1,783 & 1,456 & 853 & 1.2s & 917 & 15s & 33.5 & 33.7 & 33.2 & 32.3 \\
            \multicolumn{1}{|c|}{}                   & $q_3$ & 104 & 104 & 104 & 104 & 4,754 & 4,468 & 4,751 & 4,467 & 2.5s & 2.8s & 2.5s & 2m & 49.9 & 50.0 & 46.3 & 43.2 \\
            \multicolumn{1}{|c|}{}                   & $q_4$ & 456 & 224 & 224 & 224 & 6,740 & 3,159 & 6,838 & 3,410 & 3.4s & 3.6s & 3.4s & 1.3m & 93.9 & 111.9 & 97.5 & 50.2 \\
            \multicolumn{1}{|c|}{}                   & $q_5$ & 624 & 624 & 624 & 624 & 69,449 & 32,922 & 70,486 & 38,902 & 43.4s & 44.5s & 43.2s & \dag & 865.1 & 859.7 & 863.9 & \dag \\
            \hline
            \multicolumn{1}{|c|}{\multirow{5}{*}{$\mathsf{P5}$}} & $q_1$ & 6 & 6 & 6 & 6 & 14 & 14 & 13 & 13 & 2 & 2 & 3 & 2 & 4.2 & 4.2 & 4.2 & 4.2 \\
            \multicolumn{1}{|c|}{}                   & $q_2$ & 10 & 10 & 10 & 10 & 77 & 25 & 80 & 47 & 9 & 9 & 11 & 15 & 8.6 & 8.6 & 8.6 & 8.6 \\
            \multicolumn{1}{|c|}{}                   & $q_3$ & 13 & 13 & 13 & 13 & 400 & 60 & 413 & 208 & 61 & 52 & 53 & 303 & 11.4 & 11.4 & 11.4 & 14.95 \\
            \multicolumn{1}{|c|}{}                   & $q_4$ & 15 & 15 & 15 & 15 & 2210 & 180 & 2273 & 936 & 400 & 391 & 375 & 11s & 33.5 & 33.5 & 33.5  & 124.2\\
            \multicolumn{1}{|c|}{}                   & $q_5$ & 16 & 16 & 16 & 16 & 13085 & 725 & 13424 & 5188 & 3s & 3.1s & 3.3s & \dag & 208.6 & 208.4 & 208.3 & \dag \\
            \hline
            \multicolumn{1}{|c|}{\multirow{5}{*}{$\mathsf{SF}$}} & $q_1$ & 1 & 1 & 1 & 1 & 3 & 3 & 0 & 0 & 3 & 30 & 3 & 5 & 6.2 & 6.2 & 6.2 & 6.2 \\
            \multicolumn{1}{|c|}{}                   & $q_2$ & 125 & 125 & 125 & 125 & 15 & 15 & 12 & 12 & 6 & 97 & 11 & 11 & 6.6 & 6.6 & 6.6 & 6.6 \\
            \multicolumn{1}{|c|}{}                   & $q_3$ & 1,000 & 1,000 & 1,000 & 1,000 & 30 & 30 & 27 & 27 & 15 & 1.5s & 23 & 28 & 9.1 & 6.6 & 9.1 & 9.1\\
            \multicolumn{1}{|c|}{}                   & $q_4$ & 8,000 & 8,000 & 8,000 & 8,000 & 60 & 60 & 57 & 57 & 82 & 83s & 84 & 89 & 27.4 & 83.6 & 27.4 & 27.4 \\
            \multicolumn{1}{|c|}{}                   & $q_5$ & 27,000 & 27,000 & 27,000 & 27,000 & 39 & 39 & 33 & 33 & 472 & \dag & 427 & 415 & 121.5 & 384.5 & 121.5 & \dag \\
            \hline
            \hline
            \multicolumn{1}{|c|}{\multirow{5}{*}{$\mathsf{CLQ}$}} & $q_1$ & 38 & 38 & 38 & 38 & 38 & 38 & 57 & 57 & 8 & 41 & 28 & 45 & 4.6 & 5.1 & 5.1 & 5.1 \\
            \multicolumn{1}{|c|}{}                   & $q_2$ & 38 & 38 & 38 & 38 & 39 & 39 & 56 & 56 & 15 & 41 & 31 & 60 & 4.7 & 5.3 & 5.3 & 5.3 \\
            \multicolumn{1}{|c|}{}                   & $q_3$ & 152 & 152 & 152 & 152 & 44 & 44 & 59 & 59 & 17 & 1.3s & 38 & 56 & 5.5 & 17.7 & 6.0 & 6.0 \\
            \multicolumn{1}{|c|}{}                   & $q_4$ & 5,776 & 5,776 & 5,776 & 5,776 & 82 & 82 & 112 & 112 & 317 & \dag & 41 & 426 & 87.1 & \dag & 88.1 & 88.1 \\
            \hline
            \hline
        \end{tabular}
    }
\end{table}

\paragraph{\small{\textsf{Query Subsumption.}}}
An common way of reducing the size of the rewriting is to check for
queries that are subsumed by some other queries in the rewriting and
eliminate them. Formally, given two CQs $q_1$ and $q_2$, we say that
$q_1$ \emph{subsumes} $q_2$ if there exists a homomorphism $h$ such
that $h(\body{q_1}) \subseteq \body{q_2}$ and $h(\head{q_1}) =
\head{q_2}$. Let us clarify that such a (query) subsumption check is
not explicitly included as part of $\mathsf{XRewrite}$; it is a well-known
technique that can be exploited by any rewriting algorithm.
\irispm implements query subsumption using three different modes.
The first mode (\textsc{Tail}) consists of applying an exhaustive
subsumption check for each pair of queries in the final rewriting,
and by eliminating the subsumed ones together with all its
descendants according to the query graph. The procedure preserves
the subsumee in case it is a descendant of the subsumed query. This
mode guarantees a minimal number of CQs in the final rewriting.
The \emph{intra-decomposition} mode (\textsc{IDec}) applies the
subsumption check at the end of the rewriting of a single component
obtained after the decomposition of the input query. This mode has
the advantage that the subsumption check is applied on smaller
queries and on smaller rewriting sets; however, it does not
guarantee minimality of the final rewriting since a redundant query
may be obtained during the unfolding step. Note that, if the query
is not decomposable, then \textsc{IDec} coincides with
\textsc{Tail}.
The \emph{intra-rewriting} mode (\textsc{IRew}) applies the
subsumption check every time a new query is generated by a rewriting
step. This mode has the advantage of shrinking the rewriting search
space by pruning redundant CQs as soon as they are generated, but
has the disadvantage that it might prevent completeness. As for
\textsc{IDec}, if a query is decomposable, then \textsc{IRew} does
not guarantee minimality; otherwise, \textsc{IRew} coincides with
\textsc{Tail}.

Table~\ref{tab:subsumption-check} reports on the impact of the three modes
above on the final rewriting. The
comparison is carried out against a baseline (\textsc{Base}), where
query elimination and parallelization are applied. Notice that the
number of the explored and generated CQs is reported only for
\textsc{IRew} since is the only subsumption check mode that has a
potential effect on the exploration of the rewriting search space.
The last two groups of columns report on the effect of the different
subsumption check modes on the rewriting time and memory
consumption. The symbol ``\dag'' denotes that the rewriting
did not terminate within 15 minutes.

A first interesting observation is that the baseline algorithm
already computes a minimal rewriting in most of the cases, with the
exception of queries $q_1$, $q_2$ and $q_4$ on $\mathsf{A}$. Also,
the number of explored and generated queries matches those explored
and generated by the intra-rewriting subsumption check for all
queries in $\mathsf{S}$, $\mathsf{U}$, $\mathsf{SF}$ and
$\mathsf{CLQ}$. On the other hand, \textsc{IRew} adds a substantial
burden in terms of rewriting time, becoming impractical for $q_5$ on
$\mathsf{A}$ and $\mathsf{P5}$, where our algorithm does not
terminate within 15 minutes from its invocation. Another observation
is that, despite the fact that only \textsc{Tail} provably
guarantees the minimality of the rewriting, both \textsc{IDec} and
\textsc{IRew} produce a minimal number of CQs for the given input
queries and ontologies. Also, \textsc{Tail} becomes impractical for
complex queries on $\mathsf{SF}$ and $\mathsf{CLQ}$, whereas both
\textsc{IDec} and \textsc{IRew} terminate with timings comparable to
the baseline.

In summary, our tests indicate that \textsc{IDec} provides a good
trade-off between the need for minimization of the rewriting and
performance. Also, it seems that the amount of resources necessary
to remove redundant queries via \textsc{Tail} or \textsc{IRew} is
not justified by the gain in size, especially if we consider caching
mechanisms at the database level.

\renewcommand{\tabcolsep}{0.4pt}
\begin{table}[t]
    \centering
    \tbl{Propagation and cover graphs. \label{tab:prop-graph}}{
        \begin{tabular}{c|c|c|c|c|c|c|c|}
            \cline{2-8}
            ~ & \multicolumn{2}{c|}{~Size (\#nodes,\#edges)~} & \multicolumn{1}{c|}{LP} & \multicolumn{2}{c|}{Time (ms)} &  \multicolumn{2}{c|}{Memory}\\
            \cline{2-8}
            ~&~&~&~&~&~&~&~\\
            ~ & \textsc{P-Graph} & \textsc{C-Graph} & \textsc{C-Graph} & \textsc{P-Graph} & \textsc{C-Graph} & \textsc{P-Graph} & \textsc{C-Graph} \\
            ~&~&~&~&~&~&~&~\\
            \hline
            \multicolumn{1}{|c|}{$\mathsf{V}$} & (214,445) & (214,1194) & 7 & 4 & 158 & 190Kb & 4.7Mb \\
            \hline
            \multicolumn{1}{|c|}{$\mathsf{S}$} & (41,103) & (41,405) & 8 & 1 & 120 & 45Kb & 4,4Mb \\
            \hline
            \multicolumn{1}{|c|}{$\mathsf{U}$} & (86,189) & (86,416) & 6 & 1 & 37 & 81Kb & 4.4Mb \\
            \hline
            \multicolumn{1}{|c|}{$\mathsf{A}$} & (135,319) & (135,1133) & 10 & 43 & 708 & 154Kb &   4.7Mb \\
            \hline
            \multicolumn{1}{|c|}{$\mathsf{P5}$} & (15,32) & (15,43) & 3 & 0 & 1 & 14Kb & 4.3Mb \\
            \hline
            \multicolumn{1}{|c|}{$\mathsf{SF}$} & (100,195) & (100,1,050) & 19 & 1 & 346 & 84Kb & 4.8Mb \\
            \hline
            \multicolumn{1}{|c|}{$\mathsf{CLQ}$} & (11,143) & N/A & N/A & 2 & N/A & 39Kb & N/A \\
            \hline
            \hline
        \end{tabular}
    }
\end{table}

\subsection{Computing the Support Data
Structures}\label{sec:computing-support-data-structures}
%

%
$\mathsf{XRewrite}$ relies on a number of data structures, i.e.,
query, propagation, and cover graphs, supporting the rewriting
process. A natural question is how large such data structures can be
and how long does it take to compute them.

For the query graph the answer to such questions is straightforward
since its maximum size corresponds to the number of queries
generated by \tgdrewrite when no subsumption check is applied.
Similarly, the time to compute it and the memory consumption
correspond roughly to the rewriting time and the total memory usage
of \tgdrewrite.

Differently from the query graph, the propagation and the cover
graph depend only on the input ontology and not on the input query.
Table~\ref{tab:prop-graph} reports the characteristics of both the
propagation graph (\textsc{P-Graph}) and the cover graph
(\textsc{C-Graph}) constructed for each ontology.
In particular, we report on the size of the two structures in terms
of the number of nodes and edges, the time necessary to construct
them, and their memory footprint. For the cover graph, we also
report the length of the longest label on an edge (LP),
corresponding to the longest tight sequence of TGDs that we have to
consider during the computation of cover sets. Since query
elimination can be applied only to linear TGDs, for the sticky
ontology $\mathsf{CLQ}$ no cover graph is computed.

Apart from $\mathsf{S}$ and $\mathsf{V}$, in all other cases, the
time to compute the cover graph is either negligible or comparable
to the time to rewrite a query w.r.t. the corresponding ontology. For
$\mathsf{S}$ and $\mathsf{V}$, the reason of the higher cost
compared with the time necessary to rewrite the input queries is to
be found in the fact that these ontologies are relatively simple and
most of the machinery devised for the general case is not needed to
efficiently handle these cases. On the other hand, considered the
improvements that these two structures bring in terms of rewriting
size, rewriting time, and memory consumption for the general case,
it is certainly worthwhile to make use of them.


\renewcommand{\tabcolsep}{1pt}
\begin{table}[t]
    \centering
    \tbl{\alaska vs \irispm. \label{tab:comparative}}{
        \begin{tabular}{cc|c|c|c|c|c|c|c|c|c|c|c|}
            \cline{3-13}
            ~ & ~ & \multicolumn{2}{c|}{{Size}} & \multicolumn{2}{c|}{Explored} & \multicolumn{2}{c|}{Generated} & \multicolumn{3}{c|}{Time (ms)} &  \multicolumn{2}{c|}{Memory (MB)}\\
            \cline{3-13}
            ~&~&~&~&~&~&~&~&~&~&~&~&~\\
            ~ & ~ & \alaska & \irispm & \alaska & \irispm & \alaska & \irispm & \alaska & \irispm & S & \alaska & \irispm \\
            ~&~&~&~&~&~&~&~&~&~&~&~&~\\
            \hline
           \multicolumn{1}{|c|}{\multirow{5}{*}{$\mathsf{V}$}} & $q_1$ & 15 & 15 & 15 & 15 & 14 & 14 & 116 & 13 & \checkmark & .024 & 4.3 \\
          \multicolumn{1}{|c|}{}                       & $q_2$ & 10 & 10 & 10 & 12 & 9 & 9 & 19 & 11 & $\times$ & .024 & 6.4 \\
          \multicolumn{1}{|c|}{}                       & $q_3$ & 72 & 72 & 72 & 28 & 117 & 25 & 36 & 21 & \checkmark & .054 & 6.7 \\
           \multicolumn{1}{|c|}{}                       & $q_4$ & 185 & 185 & 185 & 43 & 328 & 40 & 60 & 37 & \checkmark & .69 & 7.4 \\
           \multicolumn{1}{|c|}{}                       & $q_5$ & 30 & 30 & 30 & 14 & 59 & 7 & 5 & 13 & $\times$ & .174 & 6.7 \\
            \hline
            \multicolumn{1}{|c|}{\multirow{5}{*}{$\mathsf{S}$}} & $q_1$ & 6 & 6 & 6 & 6 & 9 & 7 & 0 & 2 & \checkmark & .039 & 4.2 \\
            \multicolumn{1}{|c|}{}                       & $q_2$ & 2 & 2 & 48 & 2 & 288 & 1 & 7 & 2 & \checkmark & .004 & 8.2 \\
            \multicolumn{1}{|c|}{}                       & $q_3$ & 4 & 4 & 54 & 4 & 686 & 3 & 25 & 3 & \checkmark & .002 & 8.3 \\
            \multicolumn{1}{|c|}{}                       & $q_4$ & 4 & 4 & 192 & 4 & 1,632 & 2 & 56 & 4 & \checkmark & .005 & 8.3 \\
            \multicolumn{1}{|c|}{}                       & $q_5$ & 8 & 8 & 224 & 6 & 3,424 & 4 & 195 & 5 & \checkmark & .013 & 8.3 \\
            \hline
            \multicolumn{1}{|c|}{\multirow{5}{*}{$\mathsf{U}$}} & $q_1$ & 2 & 2 & 5 & 6 & 4 & 4 & 23 & 6 & \checkmark & .011 & 6.3 \\
            \multicolumn{1}{|c|}{}                       & $q_2$ & 1 & 1 & 42 & 1 & 148 & 0 & 119 & 2 & \checkmark & .002 & 4.2 \\
            \multicolumn{1}{|c|}{}                       & $q_3$ & 4 & 4 & 48 & 9 & 260 & 5 & 82 & 5 & \checkmark & .001 & 8.3 \\
            \multicolumn{1}{|c|}{}                       & $q_4$ & 2 & 2 & 1,300 & 5 & 6,092 & 4 & 2.4s & 4 & \checkmark & .006 & 8.3 \\
            \multicolumn{1}{|c|}{}                       & $q_5$ & 10 & 10 & 100 & 10 & 1,430 & 8 & 233 & 5 & \checkmark & .003 & 8.3 \\
            \hline
            \multicolumn{1}{|c|}{\multirow{5}{*}{$\mathsf{A}$}} & $q_1$ & 27 & 27 & 457 & 679 & 1,307 & 725 & 517 & 324 & \checkmark & 16 & 17.0 \\
            \multicolumn{1}{|c|}{}                       & $q_2$ & 50 & 50 & 1,598 & 1,772 & 4,658 & 4,704 & 2s & 1.21s & \checkmark & .050 & 17.85 \\
            \multicolumn{1}{|c|}{}                       & $q_3$ & 104 & 104 & 4,477 & 4,754 & 1,3871 & 4,751 & 4.5s & 2.5s & \checkmark & .697 & 46.6 \\
            \multicolumn{1}{|c|}{}                       & $q_4$ & 224 & 224 & 4,611 & 6,740 & 15,889 & 6,838 & 3.8s & 3.5s & \checkmark & .716 & 97.7 \\
            \multicolumn{1}{|c|}{}                       & $q_5$ & 624 & 624 & 50,508 & 69,449 & 231,899 & 70,486 & 12.8m & 42.4s & \checkmark& 3.5 & 863.9 \\
            \hline
            \multicolumn{1}{|c|}{\multirow{5}{*}{$\mathsf{P5}$}} & $q_1$ & 6 & 6 & 14 & 14 & 13 & 13 & 0 & 2 & \checkmark & .004 & 4.2 \\
            \multicolumn{1}{|c|}{}                        & $q_2$ & 10 & 10 & 67 & 77 & 130 & 80 & 4 & 9 & \checkmark & .007 & 8.6 \\
            \multicolumn{1}{|c|}{}                        & $q_3$ & 13 & 13 & 332 & 400 & 1,001 & 413 & 74 & 50 & $\times$ & .010 & 11.4 \\
            \multicolumn{1}{|c|}{}                        & $q_4$ & 15 & 15 & 1,647 & 2,210 & 7,065 & 2,273 & 2.6s & 378 & \checkmark & 3.5 & 33.5 \\
            \multicolumn{1}{|c|}{}                        & $q_5$ & 16 & 16 & 8,186 & 13,085 & 47,608 & 13,424 & 2m & 3s & \checkmark & .914 & 208.3\\
            \hline
            \multicolumn{1}{|c|}{\multirow{5}{*}{$\mathsf{SF}$}} & $q_1$ & 1 & 1 & 1 & 3 & 0 & 0 & 0 & 3.5 & \checkmark & 1 & 6.2 \\
            \multicolumn{1}{|c|}{}                       & $q_2$ & 125 & 125 & 125 & 15 & 300 & 12 & 10 & 7 & $\times$ & 122 & 6.6 \\
            \multicolumn{1}{|c|}{}                       & $q_3$ & 1,000 & 1,000 & 1,000 & 30 & 2,800 & 27 & 193 & 19 & \checkmark & 973 & 9.1 \\
            \multicolumn{1}{|c|}{}                       & $q_4$ & 8,000 & 8,000 & 8,000 & 60 & 23,600 & 57 & 7.1s & 93 & \checkmark & 6.4 & 27.4 \\
            \multicolumn{1}{|c|}{}                       & $q_5$ & 27,000 & 27,000 & 27,000 & 39 & 135,000 & 33 & 3.6m & 425 & \checkmark & 40.0 & 121.5  \\
            \hline
            \hline
            \multicolumn{1}{|c|}{\multirow{5}{*}{$\mathsf{CLQ}$}} & $q_1$ & 38 & 38 & 38 & 38 & 218 & 57 & 23 & 25 & \checkmark & 37 & 5.1 \\
            \multicolumn{1}{|c|}{}                         & $q_2$ & 38 & 38 & 38 & 39 & 218 & 56 & 65 & 32 & \checkmark & 41 & 5.3 \\
            \multicolumn{1}{|c|}{}                         & $q_3$ & 152 & 152 & 152 & 44 & 1,452 & 59 & 1.3s & 36 & \checkmark & 193 & 6.0 \\
            \multicolumn{1}{|c|}{}                         & $q_4$ & 5,776 & 5,776 &  82 & 82 & 112 & 112 & \dag & 346 & - & \dag &
            48.8\\
            \hline
            \hline
        \end{tabular}
    }
\end{table}

\subsection{Comparative Evaluation}\label{sec:comparative evaluation}

Although several DL-based systems exist that can deal with the
DL-Lite$_{\R}$ ontologies in our tests, to the best of our knowledge
only \alaska (i.e., the reference implementation of~\cite{KLMT12})
supports ontological query answering under general TGDs. We believe
that limiting the comparison to these two systems is fair. DL-based
systems leverage specificities of DLs, such as the limitation to
unary and binary relations only, and the absence of variable
permutations in the axioms, that enable more efficient rewriting
techniques that cannot be easily extended to more general languages
such as TGDs; in fact, DL-based systems often resort to case-by-case
analysis on the syntactic form of DL axioms. In addition to the
queries provided by the benchmarks, we also generated 492 additional
queries using \tool{SyGENiA}~\cite{ISC*12}, an automatic query
generation tool for testing the completeness of rewriting-based DL
systems. These queries do not cover the non-DL ontologies
$\mathsf{SF}$ and $\mathsf{CLQ}$. For space reasons,
Table~\ref{tab:comparative} limits the results of the evaluation to
the benchmark queries. Results for the full (internal and
comparative) evaluation are available
online.\footnote{\url{https://www.dropbox.com/s/llueoa39y9xidfa/full_evaluation.zip}.}

For \alaska we chose the setting that consistently reported the smallest size of the rewriting
and, in case of a tie, the one with lower rewriting time, namely
{\tt ar-single} in \alaska terminology. In case of \irispm, we
apply query elimination, parallelization, and
intra-decomposition subsumption check.

Transient states of the experimental machines can bias running time
and memory consumption values. For a fair comparison, we run both
systems 10 times and report the median of the values to limit biases
due to outliers. Also, since code instrumentation for running time
can interfere with memory consumption values and vice-versa, 10 runs
have been performed only with code instrumented for running time and
other 10 with code instrumented for memory consumption. Moreover,
the column (S) shows whether the difference in running time between
\alaska and \irispm is statistically significant (\checkmark) or not
($\times$). For a query $q$, we say that the difference in running
time is significant if it is greater than the maximum standard
deviation recorded for $q$ on the two systems, i.e., if
$|\mathit{time}(q,\alaska) - \mathit{time}(q,\irispm)|
> \max\{f(q,\alaska), f(q,\irispm)\}$, where $\mathit{time}(q,s)$
is the rewriting time for $q$ on system $s$, and $f(q,s)$ denotes
the standard deviation recorded for $q$ on $s$ over the 10 runs.
As before, the symbol ``\dag'' denotes test-cases where the
rewriting process either did not terminate within 15 minutes, or it
did run out of memory. Regarding the running time, a value of 0
indicates a running time below the millisecond.

A first observation is that both systems return minimal UCQ
rewritings on the given test cases. A second observation is that query elimination allows \irispm to perform a better exploration of the rewriting search space on $\mathsf{V}$, $\mathsf{S}$ and $\mathsf{U}$, where it is more effective, while \alaska explores the search space better on $\mathsf{A}$ and
$\mathsf{P5}$. This is due to the better normalization of TGDs with
multiple heads applied by \alaska that we are planning to consider
also for \irispm. On the other hand, on these ontologies caching
allows \irispm to perform better than \alaska since both query
elimination and parallelization are rather ineffective on these
ontologies. On $\mathsf{SF}$ and $\mathsf{CLQ}$, parallelization
provides a fundamental contribution towards making the rewriting
manageable as the number of explored and generated queries is
drastically reduced. As expected, \alaska consumes substantially
less memory than \irispm and delivers better performance
than \irispm on simpler queries.

By extending the comparison to the full set of
\tool{SyGENiA}-generated queries, the following facts can be
observed. All generated queries have length (i.e., number of atoms)
less than 3, and are therefore considerably simpler than those
provided by the benchmark. This is due to the fact that
\tool{SyGENiA}'s goal is to test for completeness and is not meant
to stress-test the rewriting engines. On 80\% of the test queries,
\irispm generates a rewriting of the same size as \alaska while, for
the remaining 20\%, \alaska produces smaller rewritings. This is due
to the parallelization that prevents subsumption check across
components. By running \irispm with \textsc{Tail} subsumption check,
it can be verified that the outputs of \alaska and of \irispm
coincide in size for all queries. In terms of exploration and
generation of queries, \irispm explores and generates less queries
than \alaska in 78\% of the cases, while \alaska explores the search
space better in 22\% of the cases. This is again due to the
parallelization that prevents atom coverage from identifying
redundant atoms across different components.

\section{Conclusions}\label{sec:conclusions}

The problem of designing a practical query rewriting algorithm,
which is able to treat arbitrary TGDs, has been investigated. In
particular, a resolution-based query rewriting algorithm, called
\tgdrewrite, for linear and sticky TGDs has been proposed,
and several optimization techniques have been studied.
%
%
An extensive analysis on the impact of the proposed optimizations on
the rewriting process, as well as a comparison of our system with
the only known system which supports query rewriting under arbitrary
TGDs, that is, \textsc{Alaska} (i.e., the reference implementation
of~\cite{KLMT12}), have been also performed.
In the future, we would like to study in more depth the problem of
parallelizing the rewriting process. In particular, we are planning
to investigate more sophisticated techniques of decomposing the
input query into smaller queries that can be rewritten
independently. Also, effective execution of large rewritings in
forms of UCQs as well as Datalog rewritings will be investigated.


\bibliographystyle{bib-style}
\bibliography{main}


\elecappendix

\medskip

\section{Definitions and Background}\label{appsec:preliminaries}

\subsection{Technical Definitions}\label{appsec:technical-definitions}

\paragraph{\small{\textsf{Tuple-Generating Dependencies.}}}
A set $\dep$ of TGDs is in \emph{normal form} if each of its TGDs
has a single head-atom which contains only one occurrence of an
existentially quantified variable.
As shown, e.g., in~\cite{CaGP12}, every set $\dep$ of TGDs over a
schema $\R$ can be transformed in logarithmic space into a set
$\norm{\dep}$ over a schema $\R_{\norm{\dep}}$ in normal form of
size at most quadratic in $|\dep|$, such that $\dep$ and
$\norm{\dep}$ are equivalent w.r.t.~query answering.
For a TGD $\sigma \in \dep$, if $\sigma$ is already in normal form,
then $\norm{\sigma} = \{\sigma\}$; otherwise, assuming that
$\{\atom{a}_1,\ldots,\atom{a}_k\} = \head{\sigma}$,
$\{X_1,\ldots,X_n\} = \var{\body{\sigma}} \cap \var{\head{\sigma}}$,
and $Z_1,\ldots,Z_m$ are the existentially quantified variables of
$\sigma$, let $\norm{\sigma}$ be the set
\[
\begin{array}{rcl}
\body{\sigma} &\ra& \exists Z_1 \,
\mathit{p}_{\sigma}^{1}(X_1,\ldots,X_n,Z_1)\\
\mathit{p}_{\sigma}^{1}(X_1,\ldots,X_n,Z_1) &\ra& \exists
Z_2 \, \mathit{p}_{\sigma}^{2}(X_1,\ldots,X_n,Z_1,Z_2)\\
\mathit{p}_{\sigma}^{2}(X_1,\ldots,X_n,Z_1,Z_2) &\ra& \exists
Z_{3} \, \mathit{p}_{\sigma}^{3}(X_1,\ldots,X_n,Z_1,Z_2,Z_3)\\
&\vdots&\\
\mathit{p}_{\sigma}^{m-1}(X_1,\ldots,X_n,Z_1,\ldots,Z_{m-1}) &\ra&
\exists
Z_m \, \mathit{p}_{\sigma}^{m}(X_1,\ldots,X_n,Z_1,\ldots,Z_{m})\\
\mathit{p}_{\sigma}^{m}(X_1,\ldots,X_n,Z_1,\ldots,Z_m) &\ra&
\atom{a}_1\\
&\vdots&\\
\mathit{p}_{\sigma}^{m}(X_1,\ldots,X_n,Z_1,\ldots,Z_m) &\ra&
\atom{a}_k,
\end{array}
\]
where $p_{\sigma}^{i}$ is an $(n+i)$-ary auxiliary predicate not
occurring in $\R$, for each $i \in [m]$.
Let $\norm{\dep} = \bigcup_{\sigma \in \dep} \norm{\sigma}$, and
$\R_{\norm{\dep}}$ be the schema obtained by adding to $\R$ the
auxiliary predicates occurring in $\norm{\dep}$.

\paragraph{\small{\textsf{The TGD Chase Procedure.}}}
Here is an example of how the TGD chase procedure works. Consider
the set $\dep$ of TGDs consisting of
\[
\sigma_1\ :\ p(X,Y,Z)\ \ra\ s(Y,X) \qquad \textrm{and} \qquad
\sigma_2\ :\ s(X,Y)\ \ra\ \exists Z \exists W \, p(Y,Z,W),
\]
and let $D = \{p(a,b,c)\}$.
An infinite chase of $D$ w.r.t~$\dep$ is:
\[
\begin{array}{c}
D\\
\tup{\sigma_1,h_1 = \{X \ra a,Y \ra b,Z \ra c\}}\\
D \cup \{s(b,a)\}\\
\tup{\sigma_2,h_2 = \{X \ra b,Y \ra a\}}\\
D \cup \{s(b,a),p(a,z_1,z_2)\}\\
\tup{\sigma_1,h_3 = \{X \ra a,Y \ra z_1,Z \ra z_2\}}\\
D \cup \{s(b,a),p(a,z_1,z_2),s(z_1,a)\}\\
\vdots\\
\tup{\sigma_2,h_{2i+2} = \{X \ra z_{2i-1},Y \ra a\}}\\
D \cup \{s(b,a),p(a,z_1,z_2)\} \cup \bigcup_{j=1}^{i}
\{s(z_{2j-1},a),p(a,z_{2j+1},z_{2j+2})\}\\
\vdots\\
\end{array}
\]
Clearly, $\chase{D}{\dep}$ is the infinite instance
\[
\{p(a,b,c),s(b,a),p(a,z_1,z_2)\}\ \cup\ \bigcup_{j=1}^{\infty}
\{s(z_{2j-1},a),p(a,z_{2j+1},z_{2j+2})\},
\]
where $z_1,z_2,\ldots$ are nulls of $\freshdom$.

\subsection{Query Answering via Rewriting}\label{appsec:query-rewriting}

The problem of deciding whether a set of TGDs guarantees the
first-order rewritability of CQ answering is undecidable. This
negative result holds already for the class of \emph{full} TGDs,
i.e., TGDs without existentially quantified variables. To establish
this we first need to define when a set of full TGDs is bounded.
Consider a database $D$, and a set $\dep$ of full TGDs. The
\emph{level} of an atom $\atom{a} \in \chase{D}{\dep}$ is defined
inductively as follows: if $\atom{a} \in D$, then $\level{\atom{a}}
= 0$; otherwise, if $\atom{a}$ is obtained during the chase step
$I_i \tup{\sigma,h} I_{i+1}$, then $\level{\atom{a}} =
\max_{\atom{b} \in h(\body{\sigma})} \{\level{\atom{b}}\} + 1$.
The chase of $D$ w.r.t.~$\dep$ up to level $k \geqslant 0$, denoted
$\pchase{k}{D}{\dep}$, is defined as the instance
$\{\atom{a}~|~\atom{a} \in \chase{D}{\dep} \textrm{~and~}
\level{\atom{a}} \leqslant k\}$.
A set $\dep$ of full TGDs over a schema $\R$ is \emph{bounded} if
there exists an integer constant $k \geqslant 0$ such that
$\chase{D}{\dep} = \pchase{k}{D}{\dep}$, for every database $D$ for
$\R$.

It is not difficult to show that a set of full TGDs guarantees the
first-order rewritability of CQ answering iff is bounded.
The ``only-if'' direction follows from the fact that classes of TGDs
which enjoy the so-called \emph{bounded-derivation depth property
(BDDP)} guarantee the first-order rewritability of CQ
answering~\cite{CaGL12}. The BDDP implies that, for query answering
purposes, we can consider the chase up to a level which depends only
on the query and the set of TGDs (but not on the database); clearly,
a set of full TGDs which is bounded trivially enjoys the BDDP.
The ``if'' direction is implicit in~\cite{AjGu89}, where it is shown
that each first-order expressible Datalog query is bounded.
Since the problem of deciding whether a set of full TGDs is bounded
is undecidable, which is implicit in~\cite{GMSV93} where it is shown
that the same problem for Datalog programs is undecidable, the
desired result follows.

\section{UCQ Rewriting}\label{appsec:ucq-rewriting}

\subsection{Additional Modeling
Features}\label{appsec:additional-modeling-features}

We discuss how linear and sticky sets of TGDs can be safely combined
with functional dependencies (FDs) and negative constraints, that
is, modeling features which are vital for representing ontologies.

\paragraph{\small{\textsf{Functional Dependencies.}}}
The interaction of general TGDs and FDs has been proved to lead to
undecidability of query answering. In fact, this is true even in
simple cases such that of inclusion and functional dependencies
\cite{ChVa85}, or inclusion and key dependencies, see,
e.g.,~\cite{CaLR03}. Thus, we cannot hope to safely combine the
classes of TGDs discussed above with FDs, unless suitable syntactic
restrictions are applied which would guarantee the decidability of
query answering.

A \emph{functional dependency} $\phi$ over a schema $\R$ is an
assertion $r : \insA \ra \insB$, where $r \in \R$ and $\insA,\insB$
are sets of attributes of $r$, asserting that the attributes of
$\insB$ functionally depend on the attributes of $\insA$.
Formally, $\phi$ is satisfied by an instance $I$ for $\R$ if the
following holds: whenever there exist two (distinct) atoms
$r(\tuple{t_1})$ and $r(\tuple{t_2})$ in $I$ such that
$\tuple{t_1}[\insA] = \tuple{t_2}[\insA]$, where $\tuple{t}[\insA]$
denotes the projection of tuple $\tuple{t}$ over $\insA$, then
$\tuple{t_1}[\insB] = \tuple{t_2}[\insB]$.

\begin{example}
Having the binary relation $\mathit{fatherOf}$, we can assert that
each person has at most one father by asserting that the first
attribute of $\mathit{fatherOf}$ functionally depends on the second
attribute, i.e., $\mathit{fatherOf} : \{2\} \ra \{1\}$.
\hfill\markfull
\end{example}

Note that FDs can be identified with sets of equality rules (a.k.a.
equality-generating dependencies). For example, the FD given in the
above example can be equivalently written as
$\mathit{fatherOf}(Y,X),\mathit{fatherOf}(Z,X) \ra Y=Z$.
As said, suitable syntactic restrictions are needed which would
guarantee the decidability of query answering.
A crucial concept towards this direction is separability, which
formulates a controlled interaction of TGDs and FDs; see,
e.g.,~\cite{CaGP12}.
Formally speaking, a set $\dep = \tdep \cup \fdep$ over a schema
$\R$, where $\tdep$ and $\fdep$ are sets of TGDs and FDs,
respectively, is \emph{separable} if, for every database $D$ for
$\R$, either $D \not\models \fdep$, or, for every CQ $q$ over $\R$,
$\ans{q}{D}{\dep}$\footnote{The answer to a CQ $q$ w.r.t. a database
$D$ and a set $\dep$ of TGDs can be naturally extended to sets of
TGDs and FDs (or even an arbitrary first-order theory).} $=
\ans{q}{D}{\tdep}$.
Notice that separability is a semantic notion. A sufficient
syntactic criterion for separability of TGDs and FDs is given
in~\cite{CaGP12}, and sets of TGDs and FDs satisfying this criterion
are called \emph{non-conflicting}. The formal definition of the
non-conflicting condition is beyond the scope of this paper, and for
more details we refer the reader to~\cite{CaGP12}.

Obviously, to perform query answering under non-conflicting TGDs and
FDs, we just need to apply a preliminary check whether the given
database satisfies the FDs, and if this is the case, then we
eliminate them, and proceed by considering only the set of TGDs.
This preliminary check can be reduced to the problem of CQ
evaluation. For example, given a ternary relation $r$, we can check
if the FD $r : \{1\} \ra \{3\}$ is satisfied by the database $D$ by
checking whether the CQ $q : p() \la
r(X,Y,Z),r(X,Y',Z'),\mathit{neq}(Z,Z')$ answers negatively over the
database $D_{\neq} = D \cup \{\mathit{neq}(a,b)~|~\{a,b\} \subseteq
\adom{D} \textrm{~and~} a \neq b\}$, i.e., $q(D_{\neq}) =
\emptyset$. Clearly, the atom $\mathit{neq}(a,b)$ implies that $a$
and $b$ are different constants.

\paragraph{\small{\textsf{Negative Constraints.}}}
A \emph{negative constraint} $\nu$ over a schema $\R$ is a
first-order formula of the form $\forall \insX \, \varphi(\insX) \ra
\bot$, where $\insX \subset \variables$, $\varphi$ is a conjunction
of atoms over $\R$ (possibly with constants), and $\bot$ denotes the
Boolean constant \emph{false}. Formula $\varphi$ is the \emph{body}
of $\nu$, denoted as $\body{\nu}$. Henceforth, the universal
quantifiers are omitted for brevity.

\begin{example}
With negative constraints we can assert disjointness assertions such
as students and professors are disjoint sets:
$\mathit{student}(X),\mathit{professor}(X) \ra \bot$. We can also
express non-participation assertions such as a student cannot be the
director of a research group:
$\mathit{student}(X),\mathit{directs}(X,Y) \ra \bot$.
\hfill\markfull
\end{example}

A negative constraint $\nu$ is satisfied by an instance $I$ if there
is no homomorphism $h$ such that $h(\varphi(\insX)) \subseteq I$.
Checking whether a set of negative constraints is satisfied by a
database and a set of non-conflicting TGDs and FDs is tantamount to
query answering~\cite{CaGL12}. Formally speaking, given a database
$D$, a set $\dep$ of non-conflicting TGDs and FDs, and a set
$\dep_\bot$ of negative constraints, for each $\nu \in \dep_\bot$,
we compute the answer to the CQ $q_\nu$ of the form $p() \la
\body{\nu}$ w.r.t. $D$ and $\dep$. If at least one of such queries
$q_{\nu}$ answers positively, i.e., $\tup{} \in
\ans{q_{\nu}}{D}{\dep \cup \dep_\bot}$, then there is no instance
$I$ such that $I \supseteq D$ and $I \models \dep \cup \dep_\bot$,
or, equivalently, there is no model of $D$ w.r.t. $\dep \cup
\dep_\bot$, and thus query answering is trivial since every query is
entailed; otherwise, $\ans{q}{D}{\dep \cup \dep_\bot} =
\ans{q}{D}{\dep}$, for every CQ $q$, i.e., we can answer queries by
ignoring the negative
constraints.\\

From the above discussion, we conclude that our techniques for
answering CQs under linear and sticky sets of TGDs apply immediately
even if we additionally consider FDs, providing that the
non-conflicting condition holds, and negative constraints.
Notice that the formalism obtained by taking together
non-conflicting linear or sticky sets of TGDs and FDs, and negative
constraints, is strictly more expressive than the most
widely-adopted tractable ontology languages, in particular
DL-Lite$_{\A}$, DL-Lite$_{\F}$ and DL-Lite$_{\R}$, without loosing
the desirable property of first-order rewritability; for more
details, we refer the reader to~\cite{CaGL12,CaGP12}.

\subsection{Proof of Claim~\ref{cla:auxiliary-claim}}\label{appsec:complete-auxiliary-claim}

Clearly, there exists a set $A$ such that $h(\body{p} \setminus A)
\subseteq \apchase{i-1}{D}{\dep}$ and $h(A) = \atom{a}$. Observe
that the null value that occurs in $\atom{a}$ at position
$\pi_{\exists}(\sigma)$ does not occur in $\apchase{i-1}{D}{\dep}$
or in $\atom{a}$ at a position other than $\pi_{\exists}(\sigma)$.
Therefore, the variables that occur in the atoms of $A$ at
$\pi_{\exists}(\sigma)$ do not appear at some other position.
Consequently, $A$ can be partitioned into $A_1,\ldots,A_m$, where $m
\geqslant 1$, in such a way that the following holds: for each $i
\in [m]$, in the atoms of $A_i$ at position $\pi_{\exists}(\sigma)$
the same variable $U_i$ occurs, and also $U_i$ does not occur in
$\{A_1,\ldots,A_m\} \setminus \{A_i\}$ or in $A_i$ at some position
other than $\pi_{\exists}(\sigma)$. It is easy to verify that each
set $A_i$ is factorizable w.r.t.~$\sigma$.
Suppose that we factorize $A_1$. Then, the query $p_1 =
\gamma_1(p)$, where $\gamma_1$ is the MGU for $A_1$, is obtained.
Observe that $h$ is a unifier for $A_1$. By definition of the MGU,
there exists a substitution $\theta_1$ such that $h = \theta_1 \circ
\gamma_1$. Clearly, $\theta_1(\body{p_1} \setminus \gamma_1(A)) =
\theta_1(\gamma_1(\body{p}) \setminus \gamma_1(A)) = h(\body{p}
\setminus A) \subseteq \apchase{i-1}{D}{\dep}$, $\theta_1(\insV_1) =
\theta_1(\gamma_1(\insV)) = h(\insV) = \tuple{t}$, where $\insV_1$
are the distinguished variables of $p_1$, and $\theta_1(\gamma_1(A))
= h(A) = \atom{a}$.
Now, observe that the set $\gamma_1(A_2) \subseteq \body{p_1}$ is
factorizable w.r.t.~$\sigma$. By applying factorization we get the
query $p_2 = \gamma_2(p_1)$, where $\gamma_2$ is the MGU for
$\gamma_1(A_2)$. Since $\theta_1$ is a unifier for $\gamma_1(A_2)$,
there exists a substitution $\theta_2$ such that $\theta_1 =
\theta_2 \circ \gamma_2$. Clearly, $\theta_2(\body{p_2} \setminus
\gamma_2(\gamma_1(A))) = \theta_2(\gamma_2(\body{p_1}) \setminus
\gamma_2(\gamma_1(A))) = \theta_1(\gamma_1(\body{p}) \setminus
\gamma_1(A)) = h(\body{p} \setminus A) \subseteq
\apchase{i-1}{D}{\dep}$, $\theta_2(\insV_2) =
\theta_2(\gamma_2(\insV_1)) = \theta_2(\gamma_2(\gamma_1(\insV))) =
\theta_1(\gamma_1(\insV)) = h(\insV) = \tuple{t}$, where $\insV_2$
are the distinguished variables of $p_2$, and
$\theta_2(\gamma_2(\gamma_1(A))) = \theta_1(\gamma_1(A)) = h(A) =
\atom{a}$.
Eventually, by applying the factorization step as above, we will get
the CQ $p_m = \gamma_m \circ \ldots \circ \gamma_1(p)$, where
$\gamma_j$ is the MGU for the set $\gamma_{j-1} \circ \ldots \circ
\gamma_1(A_j)$, for $j \in \{2,\ldots,m\}$ (recall that $\gamma_1$
is the MGU for $A_1$), such that $\theta_m(\body{p_m} \setminus
\gamma_m \circ \ldots \circ \gamma_1(A)) \subseteq
\apchase{i-1}{D}{\dep}$, $\theta_m(\insV_m) = \tuple{t}$, where
$\insV_m$ are the distinguished variables of $p_m$, and
$\theta_m(\gamma_m \circ \ldots \circ \gamma_1(A)) = \atom{a}$.
It is easy to verify that $\sigma$ is applicable to $\gamma_m \circ
\ldots \circ \gamma_1(A)$. The claim follows with $p' = p_m$, $S =
\gamma_m \circ \ldots \circ \gamma_1(A)$ and $\lambda = \theta_m$.

\subsection{XRewrite under More Expressive Classes of
TGDs}\label{appsec:more-expressive-classes}

\paragraph{\small{\textsf{Multi-linear.}}}
An interesting extension of linear TGDs, proposed in~\cite{CaGL12},
are the so-called multi-linear TGDs. A TGD $\sigma$ is called
\emph{multi-linear} if, for each atom $\atom{a} \in \body{\sigma}$,
$\var{\atom{a}} = \var{\body{\sigma}}$, i.e., each body-atom of
$\sigma$ contains all the body-variables of $\sigma$. The goal of
multi-linearity was the definition of a natural class of TGDs which
is strictly more expressive than DL-Lite$_{\R,\sqcap}$, that is, the
extended version of DL-Lite$_\R$ which allows for concept
conjunction~\cite{CGLL*13}.

Interestingly, our rewriting algorithm can also treat multi-linear
TGDs. Since Theorem~\ref{the:TGD-rewrite-sound-complete} holds for
arbitrary TGDs, we get that $\mathsf{XRewrite}$ is correct even if
we consider multi-linear TGDs. The non-trivial part is the
termination of $\mathsf{XRewrite}$ under this extended class.
It is possible to show that the final rewriting contains (modulo
bijective variable renaming) at most $|\mathit{terms}(q)| +
|\body{q}| \cdot \arity{\R}$ symbols (variables and constants),
where $q$ is the input query and $\R$ is the underlying schema,
which in turn implies termination of $\mathsf{XRewrite}$. This can
be established by induction on the number of atoms in the given
query.

\begin{figure}[t]
  \epsfclipon \centerline {\hbox{
      \leavevmode \epsffile{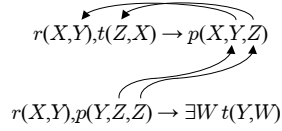} }}
  \epsfclipoff \caption{Sticky-join condition.}
  \label{fig:sticky-join-condition}
\end{figure}

\paragraph{\small{\textsf{Sticky-join.}}}
Although the class of sticky sets of TGDs is a relevant and
applicable modeling tool, it is not expressive enough to model
simple cases such as the linear TGD $r(X,Y,X) \ra \exists Z \,
s(Z,Y)$; clearly, after applying $\mathsf{SMarking}$, the variable
$X$ is marked, and thus the stickiness condition is violated. The
question whether stickiness and linearity can be safely combined was
investigated in~\cite{CaGP12}, and the class of sticky-join sets of
TGDs was proposed.
Intuitively speaking, the sticky-join condition allows a marked
variable to appear more than once in the body of a TGD $\sigma$ as
long as \emph{(i)} it appears only in one atom of $\body{\sigma}$,
and \emph{(ii)} its marking is not propagated in more than one
body-atoms of a TGD $\sigma'$ during the marking procedure (i.e.,
the situation illustrated in Figure~\ref{fig:sticky-join-condition},
where the marking of the variable $Z$ in the body of
$r(X,Y),p(Y,Z,Z) \ra \exists W \, t(Y,W)$ is propagated in two
different atoms, is forbidden).
The formal definition of this class is in the same spirit as the one
for sticky sets of TGDs, but a more involved marking procedure which
keeps track of the origin of each marking is applied; for more
details we refer the reader to~\cite{CaGP12}.

Sticky-join sets of TGDs can also be treated by our rewriting
algorithm. As for multi-linearity, the non-trivial part is the
termination of $\mathsf{XRewrite}$ under this extended class. This
can be shown by establishing a syntactic property of the rewritten
query analogous to the one for sticky sets of TGDs stated in
Lemma~\ref{lem:property}. More precisely, given a CQ $q$ over a
schema $\R$, and a sticky-join set $\dep$ of TGDs over $\R$, it can
be proved that, for each $q' \in q_{\dep}$, every variable of
$(\var{q'} \setminus \var{q})$ occurs only in one atom of
$\body{q'}$ (possibly more than once). Then, by giving an argument
similar to that in the proof of
Theorem~\ref{the:tgdrewrite-termination}, we can show that the
maximum number of CQs that can be constructed during the execution
of $\mathsf{XRewrite}$ is bounded by the number of different CQs
that can be constructed using terms of $T = \adom{q} \cup
\{\star_1,\ldots,\star_w\}$, where $w = \arity{\R}$ (recall that in
the case of sticky sets of TGDs just one special symbol is enough),
and predicates of $\R$; this immediately implies termination of
$\mathsf{XRewrite}$.

\section{Optimize the Rewriting for Linear TGDs}\label{appsec:ucq-optimization}

\subsection{Proof of Claim~\ref{cla:atom-coverage}}\label{appsec:atom-coverage-claim}

Let us first construct the TGD $\sigma$. Since $\atom{a}
\prec_{\dep}^{q} \atom{b}$, there exists a tight sequence
$\sigma_1,\ldots,\sigma_m$, for $m \geqslant 1$, of TGDs of $\dep$
which is compatible to $\atom{a}$. If $m=1$, then $\sigma =
\sigma_1$; in this case, trivially $\dep \models \sigma$, i.e., for
every instance $I$ that satisfies $\dep$, $I \models \sigma$. The
interesting case is when $m > 1$. We define $\sigma$ via an
inductive construction. Without loss of generality, we assume that
the TGDs $\sigma_1,\ldots,\sigma_m$ do not have variables in common.
By definition, there exists a homomorphism $\gamma_{12}$ such that
$\gamma_{12}(\body{\sigma_2}) = \head{\sigma_1}$. By applying the
resolution inference rule\footnote{Notice that we do not need to
Skolemise since the MGU is a homomorphism from $\body{\sigma_2}$ to
$\head{\sigma_1}$.}, we get the TGD $\sigma[12] :
\gamma_{12}(\body{\sigma_1}) \ra \gamma_{12}(\head{\sigma_2})$.
Notice that $\gamma_{12}$ is the identity on the variables of
$\sigma_1$, and hence $\sigma[12]$ is actually the TGD
$\body{\sigma_1} \ra \gamma_{12}(\head{\sigma_2})$.
Let us now show that we can obtain the TGD $\sigma[12\ldots k]$ from
$\sigma[12\ldots(k-1)]$ and $\sigma_k$ by applying the resolution
inference rule. Observe that
$\gamma_{12\ldots(k-1)}(\head{\sigma_{k-1}}) =
\head{\sigma[12\ldots(k-1)]}$. Since, by definition, there exists a
homomorphism $\gamma$ such that $\gamma(\body{\sigma_k}) =
\head{\sigma_{k-1}}$, we get that $\gamma_{12 \ldots k} =
\gamma_{12\ldots(k-1)} \circ \gamma$ maps $\body{\sigma_k}$ to
$\head{\sigma[12 \ldots (k-1)]}$.
Clearly, $\gamma_{12 \ldots k}$ is a MGU for $\body{\sigma_k}$ and
$\head{\sigma[12 \ldots (k-1)]}$. By applying the resolution
inference rule, we get $\sigma[12 \ldots k] : \gamma_{12 \ldots
k}(\body{\sigma[12 \ldots (k-1)]}) \ra \gamma_{12 \ldots
k}(\head{\sigma_k})$. Notice that $\gamma_{12 \ldots k}$ is the
identity on the variables of $\sigma[12 \ldots (k-1)]$, and thus
$\sigma[12 \ldots k] = \body{\sigma[12 \ldots (k-1)]} \ra \gamma_{12
\ldots k}(\head{\sigma_k})$. The desired TGD $\sigma$ is $\sigma[12
\ldots m]$. Notice that $\body{\sigma[12 \ldots m]} =
\body{\sigma_1}$, and hence $\sigma$ is the TGD $\body{\sigma_1} \ra
\gamma_{12 \ldots m}(\head{\sigma_m})$.

To show that $\dep \models \sigma$ it suffices to show that, given
two TGDs $\sigma'$ and $\sigma''$ such that there exists a
substitution $\gamma$ that maps $\body{\sigma''}$ to
$\head{\sigma'}$, then $\{\sigma',\sigma''\} \models \sigma'''$,
where $\sigma'''$ is the TGD $\body{\sigma'} \ra
\gamma(\head{\sigma''})$. Consider an instance $J$ that satisfies
$\{\sigma',\sigma''\}$, and assume that there exists a homomorphism
$g$ such that $g(\body{\sigma'}) \in J$ (otherwise, the claim
follows immediately). We need to show that there exists an extension
$g'$ of $g$ such that $g'(\gamma(\head{\sigma''})) \in J$. Since $J
\models \sigma'$, there exists an extension $g''$ of $g$ such that
$g''(\head{\sigma'}) \in J$. Thus, $g''(\gamma(\body{\sigma''})) \in
J$. Since $J \models \sigma''$, there exists an extension $\rho$ of
$g'' \circ \gamma$ such that $\rho(\head{\sigma''}) \in J$. Assuming
that $\head{\sigma''} = r(\insX,\insZ)$, where $\insZ =
Z_1,\ldots,Z_k$, for $k \geqslant 1$, are the existentially
quantified variables of $\sigma''$, we define the substitution $g' =
g \cup \{\gamma(Z_i) \ra \rho(Z_i)\}_{i \in [k]}$; if $\insZ =
\emptyset$, then $g' = g$. Notice that $g'$ is well-defined since
none of the variables $\gamma(Z_1), \ldots, \gamma(Z_k)$ occurs in
$g$. Clearly, $g'(\gamma(\head{\sigma''})) =
r(g'(\gamma(\insX)),g'(\gamma(\insZ))) = r(g(\insX),\rho(\insZ)) =
\rho(r(\insX,\insZ)) \in J$, as needed.

Let us now establish the existence of $\lambda$ and $\mu$. By
definition, there exists a substitution $\lambda'$ such that
$\lambda'(\body{\sigma}) = \atom{a}$. We define $\lambda$ to be the
extension of $\lambda'$ that maps each existentially quantified
variable of $\sigma$ to a ``fresh'' symbol of $\freshdom$. Let
$\mu'$ be the substitution that maps each variable occurring in
$\body{q} \setminus \{\atom{b}\}$ to itself. We obtain $\mu$ by
adding to $\mu'$ the following: for each term $t \in \adom{\atom{b}}
\setminus T(q,\atom{b})$, if $t$ occurs in $\atom{b}$ at position
$\pi$, then add $t \ra \lambda(t')$, where $t'$ is the term at
position $\pi$ in $\head{\sigma}$. By construction,
$\lambda(\head{\sigma}) = \mu(\atom{b})$, and the claim follows.

\begin{table*}[t]
  \centering
  \tbl{Case analysis in the proof of Lemma~\ref{lem:unique-elimination-strategy}. \label{tab:cases}}{
  \begin{tabular}{|c||c|c||c|}
    \hline
    Category & $\mathit{cover}(\atom{a}_{k})$ & $\mathit{cover}(\atom{a}_{k+1})$ & Result \\
    \hline
      & $= \emptyset$ & $= \emptyset$ &  \\
     & $= \emptyset$ & $\neq \emptyset$ &  \\
      & $\neq \emptyset$ & $= \emptyset$ &  \\
     A & $\supset \{\atom{a}_{k+1}\}$ & $\supset \{\atom{a}_{k}\}$ & $\mathsf{eliminate}(q,S_1,\dep) = \mathsf{eliminate}(q,S_2,\dep)$ \\
      & $\supset \{\atom{a}_{k+1}\}$ & $\not\ni \atom{a}_{k}$ &  \\
      & $\not\ni \atom{a}_{k+1}$ & $\supset \{\atom{a}_{k}\}$ &  \\
      & $\not\ni \atom{a}_{k+1}$ & $\not\ni \atom{a}_{k}$ &  \\
    \hline
    B & $= \{\atom{a}_{k+1}\}$ & $= \{\atom{a}_{k}\}$ & $\mathsf{eliminate}(q,S_1,\dep) \neq \mathsf{eliminate}(q,S_2,\dep)$ \\
    & & & $|\mathsf{eliminate}(q,S_1,\dep)| = |\mathsf{eliminate}(q,S_2,\dep)|$\\
    \hline
      & $\supset \{\atom{a}_{k+1}\}$ & $= \{\atom{a}_{k}\}$ &  \\
    C & $= \{\atom{a}_{k+1}\}$ & $\supset \{\atom{a}_{k}\}$ & not applicable \\
      & $= \{\atom{a}_{k+1}\}$ & $\not\ni  \atom{a}_{k}$, $\neq \emptyset$ &  \\
      & $\not\ni \atom{a}_{k+1}$, $\neq \emptyset$ & $= \{\atom{a}_{k}\}$ &  \\
    \hline
    \hline
  \end{tabular}}
\end{table*}

\subsection{Proof of Lemma~\ref{lem:unique-elimination-strategy}}\label{appsec:unique-elimination-lemma}

We assume that $S_1$ and $S_2$ are exactly the same except two
consecutive elements. In other words, for each $i \in
\{1,\ldots,k-1,k+2,\ldots,n\}$, $S_1[i] = S_2[i]$, $S_1[k] =
S_2[k+1]$ and $S_1[k+1] = S_2[k]$.
Notice that the above assumption does not harm the generality of the
proof since, given any two strategies $S$ and $S'$, $S$ can be
obtained from $S'$ (and vice versa) by applying finitely many times
an operator which swaps two consecutive elements of a strategy.
For example, assuming that $S_1 =
[\atom{a},\atom{b},\atom{c},\atom{d}]$ and $S_2 =
[\atom{c},\atom{a},\atom{d},\atom{b}]$, $S_2$ can be obtained from
$S_1$ as follows: $S_1 = [\atom{a},\atom{b},\atom{c},\atom{d}] \ra
[\atom{a},\atom{c},\atom{b},\atom{d}] \ra
[\atom{c},\atom{a},\atom{b},\atom{d}] \ra
[\atom{c},\atom{a},\atom{d},\atom{b}] = S_2$.
Let us now establish the claim. For notational convenience, given a
strategy $S$, let $\mathsf{eliminate}^{\ell}(q,S,\dep)$ be the
subset of $\mathsf{eliminate}(q,S,\dep)$ computed after $\ell$
applications of the for-loop; clearly,
$\mathsf{eliminate}^{k-1}(q,S_1,\dep) =
\mathsf{eliminate}^{k-1}(q,S_2,\dep)$.
In what follows, let $\atom{a}_k = S_1[k] = S_2[k+1]$ and
$\atom{a}_{k+1} = S_1[k+1] = S_2[k]$. The proof proceeds by case
analysis whether $\mathit{cover}(\atom{a}_k)$ and
$\mathit{cover}(\atom{a}_{k+1})$ are empty or not after $k-1$
applications of the for-loop.
All the possible cases are grouped in three categories which are
depicted in Table~\ref{tab:cases}. Observe that for category A,
$\mathsf{eliminate}(q,S_1,\dep)$ and
$\mathsf{eliminate}(q,S_2,\dep)$ coincide, which immediately implies
that they have the same cardinality. The interesting case is
category B where $\mathsf{eliminate}(q,S_1,\dep)$ and
$\mathsf{eliminate}(q,S_2,\dep)$ are different, but they have the
same cardinality. Finally, the cases of category C are not
applicable since it is not possible to occur. In the rest of the
proof, we prove the first case of each category; all the other cases
can be shown in a similar way.


\renewcommand{\tabcolsep}{2pt}
\begin{table*}
    \centering
    \tbl{Test queries. \label{apptab:test-queries}}{
    \begin{tabular}{|l|lll|}
        \hline
        \multicolumn{1}{|c|}{Ontology} & \multicolumn{3}{c|}{Queries} \\
        \hline
        \multicolumn{1}{|c|}{\multirow{6}{*}{$\mathsf{V}$}} & $q_1(A)$ & $\leftarrow$ & $\mathit{ Location(A).}$\\
        ~ & $q_2(A,B)$ & $\leftarrow$ & $\mathit{Military\_Person(A), hasRole(B,A), related(A,C).}$\\
        ~ & $q_3(A,B)$ & $\leftarrow$ & $\mathit{Time\_Dependant\_Relation(A), hasRelationMember(A,B), Event(B).}$\\
        ~ & $q_4(A,B)$ & $\leftarrow$ & $\mathit{Object(A), hasRole(A,B), Symbol(B).}$\\
        ~ & $q_5(A)$ & $\leftarrow$ & $\mathit{Individual(A), hasRole(A,B), Scientist(B), hasRole(A,C),}$\\
        ~&~&~& $\mathit{ Discoverer(C), hasRole(A,D), Inventor(D).}$\\
        \hline
        \multicolumn{1}{|c|}{\multirow{8}{*}{$\mathsf{S}$}} & $q_1(A)$ & $\leftarrow$ & $\mathit{ StockExchangeMember(A).}$\\
        ~ & $q_2(A,B)$ & $\leftarrow$& $\mathit{ Person(A), hasStock(A,B), Stock(B).}$\\
        ~ & $q_3(A,B,C)$& $\leftarrow$ &$\mathit{ FinantialInstrument(A), belongsToCompany(A,B), Company(B),}$\\
        ~&~&~& $\mathit{hasStock(B,C), Stock(C).}$\\
        ~ & $q_4(A,B,C)$& $\leftarrow$& $\mathit{ Person(A), hasStock(A,B), Stock(B), isListedIn(B,C),}$ \\
        ~&~&~&$\mathit{StockExchangeList(C).}$\\
        ~ & $q_5(A,B,C,D)$ &$\leftarrow$ &$\mathit{ FinantialInstrument(A), belongsToCompany(A,B), Company(B),}$\\
        ~ & ~ &~ &$\mathit{hasStock(B,C), Stock(C),isListedIn(B,D), StockExchangeList(D).}$ \\
        \hline
        \multicolumn{1}{|c|}{\multirow{6}{*}{$\mathsf{U}$}} & $q_1(A)$ &$\leftarrow$& $\mathit{ worksFor(A,B), affiliatedOrganizationOf(B,C).}$ \\
        ~ & $q_2(A,B)$ &$\leftarrow$& $\mathit{ Person(A), teacherOf(A,B), Course(B).}$ \\
        ~ & $q_3(A,B,C)$ &$\leftarrow$ &$\mathit{Student(A), advisor(A,B), FacultyStaff(B), takesCourse(A,C)}$\\
        ~&~&~& $\mathit{teacherOf(B,C), Course(C).}$\\
        ~ & $q_4(A,B)$ &$\leftarrow$& $\mathit{ Person(A), worksFor(A,B), Organization(B).}$ \\
        ~ & $q_5(A)$ &$\leftarrow$& $\mathit{ Person(A), worksFor(A,B), University(B), hasAlumnus(B,A).}$ \\
        \hline
        \multicolumn{1}{|c|}{\multirow{6}{*}{$\mathsf{A}$}} & $q_1(A)$& $\leftarrow$& $\mathit{ Device(A), assistsWith(A,B).}$ \\
        ~ & $q_2(A)$& $\leftarrow$ &$\mathit{ Device(A), assistsWith(A,B), UpperLimbMobility(B).}$ \\
        ~ & $q_3(A)$ &$\leftarrow$& $\mathit{ Device(A), assistsWith(A,B), Hear(B), affects(C,B), Autism(C).}$ \\
        ~ & $q_4(A)$& $\leftarrow$& $\mathit{ Device(A), assistsWith(A,B), PhysicalAbility(B).}$ \\
        ~ & $q_5(A)$& $\leftarrow$& $\mathit{ Device(A), assistsWith(A,B), PhysicalAbility(B), affects(C,B),}$ \\
        ~&~&~& $\mathit{Quadriplegia(C).}$\\
        \hline
        \multicolumn{1}{|c|}{\multirow{5}{*}{$\mathsf{P5}$}} & $q_1(A)$ &$\leftarrow$& $\mathit{ edge(A,B).}$ \\
        ~ & $q_2(A)$& $\leftarrow$& $\mathit{edge}(A,B), \mathit{edge}(B,C).$ \\
        ~ & $q_3(A)$& $\leftarrow$& $\mathit{edge}(A,B), \mathit{edge}(B,C), \mathit{edge}(C,D).$ \\
        ~ & $q_4(A)$& $\leftarrow$ &$\mathit{edge}(A,B), \mathit{edge}(B,C), \mathit{edge}(C,D), \mathit{edge}(D,E).$ \\
        ~ & $q_4(A)$& $\leftarrow$ &$\mathit{edge}(A,B), \mathit{edge}(B,C), \mathit{edge}(C,D), \mathit{edge}(D,E), \mathit{edge}(E,F).$ \\
        \hline
        \multicolumn{1}{|c|}{\multirow{5}{*}{$\mathsf{SF}$}} & $q_1(A,B)$ &$\leftarrow$& $\mathit{p_1(A), r_1(A,B), s_1(B).}$ \\
        ~& $q_2(A,B)$ &$\leftarrow$& $p_5(A), r_5(A,E,B), s_5(B).$\\
        ~& $q_3(A,B)$ &$\leftarrow$& $p_{10}(A), r_{10}(A,E,B), s_{10}(B).$\\
        ~& $q_4(A,B)$ &$\leftarrow$& $p_{20}(A), r_{20}(A,E,B), s_{20}(B).$\\
        ~& $q_5(A,B,C,D)$ &$\leftarrow$& $p_{10}(A), r_{10}(A,E,B), s_{10}(B), p_{20}(C), r_{20}(C,F,D), s_{20}(D).$\\
        \hline
        \multicolumn{1}{|c|}{\multirow{4}{*}{$\mathsf{CLQ}$}} & $q_1(A,B,C)$ &$\leftarrow$& $\mathit{C_{3}(A,B,C).}$ \\
        ~& $q_2(A,B,C)$ &$\leftarrow$& $C_{3}(A,B,C), Sp(A)$\\
        ~& $q_3(A,B,C)$ &$\leftarrow$& $C_{3}(A,B,C), C_{2}(A,D), Sp(A), Sp(D).$\\
        ~& $q_4(A,B,C,D,E,F)$ &$\leftarrow$& $C_{3}(A,B,C), C_{3}(D,E,F), C_{2}(A,D),Sp(A1), Sp(A2).$\\
        \hline
        \hline
    \end{tabular}}
\end{table*}


Case A1: It is not difficult to see that
$\mathsf{eliminate}^{k+1}(q,S_1,\dep)$ and
$\mathsf{eliminate}^{k+1}(q,S_2,\dep)$ coincide. Moreover, after the
$(k+1)$-th application of the for-loop, $\mathit{cover}(S_1[i])$ and
$\mathit{cover}(S_2[i])$, for each $i \in \{k+2,\ldots,n\}$, are the
same. Thus, $\mathsf{eliminate}^{n}(q,S_1,\dep)$ and
$\mathsf{eliminate}^{n}(q,S_2,\dep)$ are equal. By construction,
$\mathsf{eliminate}(q,S_i,\dep) =
\mathsf{eliminate}^{n}(q,S_2,\dep)$, for each $i \in \{1,2\}$.
Hence, $\mathsf{eliminate}(q,S_1,\dep) =
\mathsf{eliminate}(q,S_2,\dep)$, and the claim follows.

Case B1: Clearly, $\atom{a}_{k+i} \in
\mathsf{eliminate}^{k+1}(q,S_{1+i},\dep)$ and $\atom{a}_{k+i}
\not\in \mathsf{eliminate}^{k+1}(q,S_{2-i},\dep)$, for each $i \in
\{0,1\}$.
This implies that $|\mathsf{eliminate}^{k+1}(q,S_1,\dep)| =
|\mathsf{eliminate}^{k+1}(q,S_2,\dep)|$; notice that
$\mathsf{eliminate}^{k+1}(q,S_1,\dep) \neq
\mathsf{eliminate}^{k+1}(q,S_2,\dep)$ but they have the same
cardinality.
Now, consider an atom $\atom{b} \in \body{q}$. If $\atom{a}_k
\prec_{\dep}^{q} \atom{b}$, since $\atom{a}_{k+1} \prec_{\dep}^{q}
\atom{a}_{k}$, by transitivity of $\prec_{\dep}^{q}$, we get that
$\atom{a}_{k+1} \prec_{\dep}^{q} \atom{b}$. Conversely, if
$\atom{a}_{k+1} \prec_{\dep}^{q} \atom{b}$, since $\atom{a}_{k}
\prec_{\dep}^{q} \atom{a}_{k+1}$, we get that $\atom{a}_{k}
\prec_{\dep}^{q} \atom{b}$. Hence, $\atom{a}_{k} \prec_{\dep}^{q}
\atom{b}$ iff $\atom{a}_{k+1} \prec_{\dep}^{q} \atom{b}$, or,
equivalently, $\atom{a}_k \in \mathit{cover}(\atom{b})$ iff
$\atom{a}_{k+1} \in \mathit{cover}(\atom{b})$. Observe that, after
the $(k+1)$-th application of the for-loop, for each $i \in
\{k+2,\ldots,n\}$, either $\mathit{cover}(S_1[i]) =
\mathit{cover}(S_2[i])$, or $\mathit{cover}(S_1[i]) \setminus
\mathit{cover}(S_2[i]) = \{\atom{a}_{k+1}\}$ and
$\mathit{cover}(S_2[i]) \setminus \mathit{cover}(S_1[i]) =
\{\atom{a}_{k}\}$. Consequently, $|\mathsf{eliminate}(q,S_1,\dep)| =
|\mathsf{eliminate}(q,S_2,\dep)|$.

Case C1: Clearly, there exists an atom $\atom{b} \in \body{q}$,
other than $\atom{a}_k$ and $\atom{a}_{k+1}$, such that $\atom{b}
\prec_{\dep}^{q} \atom{a}_k$. Since $\atom{a}_k \prec_{\dep}^{q}
\atom{a}_{k+1}$, by transitivity of $\prec_{\dep}^{q}$, we get that
$\atom{b} \prec_{\dep}^{q} \atom{a}_{k+1}$. This implies that
$\mathit{cover}(\atom{a}_{k+1}) \supset \{\atom{a}_k\}$ which
contradicts our hypothesis that $\mathit{cover}(\atom{a}_{k+1}) =
\{\atom{a}_k\}$.

\section{Experimental Evaluation}\label{sec:experimental-evaluation-appendix}

\subsection{Test Queries}\label{appsec:test-queries}

Each ontology that we consider in our experimental evaluation has an
associated set of test queries (see Table~\ref{apptab:test-queries})
either obtained via an analysis of query logs or manually created.

\subsection{Remark on Non-recursive Datalog Rewritings}\label{appsec:nr-datalog-rewriting}

\renewcommand{\tabcolsep}{2pt}
\begin{table}[t]
    \centering
    \tbl{Non-recursive Datalog rewritings. \label{tab:dtg-rewriting}}{
        \begin{tabular}{cc|c|c|c|c|c|c|c|c|c|c|}
            \cline{3-12}
            ~ & ~ & \multicolumn{2}{c|}{Size} & \multicolumn{2}{c|}{\#Atoms} & \multicolumn{2}{c|}{\#Joins} & \multicolumn{2}{c|}{Time (ms)} & \multicolumn{2}{c|}{Memory (MB)}\\
            \cline{3-12}
            ~&~&~&~&~&~&~&~&~&~&~&~\\
            ~ & ~ & \textsc{Base} & \textsc{Datalog} & \textsc{Base} & \textsc{Datalog} & \textsc{Base}& \textsc{Datalog} & \textsc{Base}& \textsc{Datalog} & \textsc{Base}& \textsc{Datalog} \\
            ~&~&~&~&~&~&~&~&~&~&~&~\\
            \hline
            \multicolumn{1}{|c|}{\multirow{4}{*}{$\mathsf{V}$}} & $q_2$ & 10 & 13 & 30 & 15 & 30 & 3 & 4 & 7 & 6.4 & 6.4 \\
            \multicolumn{1}{|c|}{}                   & $q_3$ & 72 & 29 & 216 & 31 & 144 & 2 & 25 & 15 & 6.7 & 6.7 \\
            \multicolumn{1}{|c|}{}                   & $q_4$ & 185 & 44 & 555 & 46 & 370 & 2 & 26 & 19 & 7.4 & 7.4 \\
            \multicolumn{1}{|c|}{}                   & $q_5$ & 30 & 15 & 210 & 21 & 270 & 9 & 16 & 7 & 6.7 & 6.7 \\
            \hline
            \multicolumn{1}{|c|}{\multirow{2}{*}{$\mathsf{S}$}} & $q_4$ & 4 & 5 & 8 & 6 & 4 & 1 & 3 & 3 & 8.3 & 8.3 \\
            \multicolumn{1}{|c|}{}                   & $q_5$ & 8 & 7 & 24 & 12 & 24 & 5 & 4 & 14 & 8.3 & 8.3 \\
            \hline
            \multicolumn{1}{|c|}{\multirow{3}{*}{$\mathsf{U}$}} & $q_1$ & 2 & 4 & 4 & 5 & 2 & 1 & 4 & 6 & 6.2 & 6.2 \\
            \multicolumn{1}{|c|}{}                   & $q_3$ & 4 & 8 & 16 & 11 & 20 & 5 & 4 & 3 & 8.3 & 8.3 \\
            \multicolumn{1}{|c|}{}                   & $q_5$ & 10 & 8 & 20 & 9 & 20 & 2 & 5 & 5 & 8.3 & 8.3 \\
            \hline
            \multicolumn{1}{|c|}{\multirow{2}{*}{$\mathsf{A}$}} & $q_3$ & 104 & 107 & 520 & 317 & 520 & 210 & 2.5s & 2.4s & 49.9 & 49.9 \\
            \multicolumn{1}{|c|}{}                   & $q_5$ & 624 & 626 & 3,120 & 2,499 & 3,120 & 2,497 & 43s & 42s & 865.1 & 865.1 \\
            \hline
            \multicolumn{1}{|c|}{\multirow{5}{*}{$\mathsf{SF}$}}& $q_1$ & 1 & 4 & 3 & 6 & 2 & 2 & 3 & 3 & 6.2 & 6.2 \\
            \multicolumn{1}{|c|}{}                   & $q_2$ & 125 & 11 & 375 & 13 & 250 & 2 & 6 & 8 & 6.6 & 6.6 \\
            \multicolumn{1}{|c|}{}                   & $q_3$ & 1,000 & 31 & 3,000 & 33 & 2,000 & 2 & 15 & 6 & 9.1 & 9.1 \\
            \multicolumn{1}{|c|}{}                   & $q_4$ & 8,000 & 61 & 24,000 & 63 & 16,000 & 2 & 82 & 12 & 27.4 & 27.4 \\
            \multicolumn{1}{|c|}{}                   & $q_5$ & 27,000 & 37 & 162,000 & 42 & 108,000 & 4 & 472 & 6 & 121.5 & 121.5 \\
            \hline
            \hline
            \multicolumn{1}{|c|}{\multirow{4}{*}{$\mathsf{CLQ}$}}& $q_2$ & 38 & 40 & 140 & 105 & 166 & 97 & 15 & 16 & 4.8 & 4.8 \\
            \multicolumn{1}{|c|}{}                   & $q_3$ & 152 & 45 & 864 & 112 & 1,248 & 100 & 17 & 13 & 5.5 & 5.5 \\
            \multicolumn{1}{|c|}{}                   & $q_4$ & 5,776 & 83 & 48,336 & 215 & 83,296 & 198 & 317 & 15 & 87.1 & 87.1 \\
            \hline
            \hline
        \end{tabular}
    }
\end{table}

If we consider those tests where decomposition is more effective,
e.g., $\mathsf{SF}$ and $\mathsf{CLQ}$, we observe that most of the
time is spent unfolding the rewritten components into a UCQ. A possible way of tackling this problem is to keep the rewriting
``folded'', i.e., as a non-recursive Datalog rewriting. As mentioned
before, Datalog queries are in theory more complicated to execute
than UCQs; however, by effect of our decomposition technique, all Datalog rewritings constructed by our algorithm have a particular
shape. In fact, they consists of a set of UCQs, obtained by the
independent rewriting of the components, plus a (single)
reconciliation query (i.e., a view) that joins a number of relations
equal to the number of components constructed by the decomposition.
We conjecture that Datalog queries of this form do not represent a
major problem for current DBMSs, since they can be executed as a
simple two-levels nested SQL query.
Table~\ref{tab:dtg-rewriting} reports on the size of these
non-recursive Datalog rewritings, and on the effort to compute them
for the ontologies (and queries) where decomposition is effective.
Note that we do not report on the number of explored and generated
queries since they coincide with the corresponding values in the
column \textsc{Para} of Table~\ref{tab:decomposition}. The
comparison is carried out against a baseline (\textsc{Base}), where
query elimination and parallelization are applied, but the target
language for the rewriting is UCQs. Again, for $\mathsf{CLQ}$ query
elimination is not applied.

As expected, Datalog rewritings deliver, in average, a smaller
number of CQs to be executed. They also drastically reduce the
number of joins to be performed. The maximum gain we have observed
is for $q_5$ on $\mathsf{SF}$, where we have a gain of 99.8\% in
terms of queries to be executed and 99.99\% in terms of joins. This
has also impact on the rewriting time that is reduced by 99\%. On
the other hand, there are cases where the computation of a Datalog
rewriting increases the number of queries to be executed; e.g., for
$q_3$ on $\mathsf{U}$, where the size of the UCQ rewriting
is already smaller than the number of components computed from the
input query. It is worth noting that, even for these cases, the
number of atoms and joins is always reduced. As expected, on the $\mathsf{A}$
ontology, Datalog rewritings are not particularly effective since
almost the entire rewriting search-space is explored by one of the
components, while the others do not produce any rewritings.

\end{document}